\documentclass[11pt]{article}
\usepackage{amsmath,amssymb,array}
\usepackage{amsthm}
\usepackage{graphicx}
\usepackage{hyperref}
\usepackage{pgf,pgfarrows,pgfautomata,pgfheaps,pgfnodes,pgfshade}
\newtheorem{theorem}{Theorem}[section]

\newcommand\su{\mathrm{SU}(2,1)}
\newcommand\supp{\mathrm{SU}(2,1)^{++}}
\newcommand\suppp{\mathrm{SU}(2,1)^{+++}}
\newcommand\asu{\mathfrak{su}(2,1)}
\newcommand\asupp{\mathfrak{su}(2,1)^{++}}
\newcommand\asuppp{\mathfrak{su}(2,1)^{+++}}
\newcommand\dd{\mathrm{d}}
\newcommand\mbb{\mathbb}
\newcommand\mc{\mathcal}

\newcommand\m{\mu}
\newcommand\n{\nu}
\newcommand\p{\partial}
\newcommand\s{\sigma}

\newcommand\nn{\nonumber}
\newcommand\be{\begin{equation}}
\newcommand\ee{\end{equation}}

\begin{document}

\thispagestyle{empty}
\setcounter{page}{0}
\renewcommand{\theequation}{\thesection.\arabic{equation}}

{\hfill{\tt ULB-TH/09-10}}

\begin{center} {\bf \Large Finite and infinite-dimensional symmetries 

\vspace{0.3cm}

of pure $\mathcal{N}=2$ supergravity in $D=4$}

\vspace{.9cm}

Laurent Houart, Axel
Kleinschmidt, \\Josef Lindman H\"ornlund, Daniel Persson\footnote{Also at \emph{Fundamental Physics, Chalmers University of Technology, SE-412 96, G\"oteborg, Sweden}} and Nassiba
Tabti

\footnotesize
\vspace{.9 cm}

{\em Service de Physique Th\'eorique et Math\'ematique,\\
Universit\'e Libre de Bruxelles \& International Solvay Institutes\\ Campus Plaine C.P. 231, Boulevard du
Triomphe, B-1050 Bruxelles, 
Belgium}

\vspace{.4 cm}

 {\tt
 lhouart, axel.kleinschmidt, jlindman, dpersson,  ntabti@ulb.ac.be } 

\end{center}

\vspace {.7cm}

\begin{center}
{\bf Abstract} 
\end{center}

\vspace{0.3 cm}
{\small
\noindent We study the symmetries of  pure $\mathcal{N}=2$ supergravity in $D=4$. As is known, this theory reduced on one Killing vector is characterised by a non-linearly realised symmetry $\mathrm{SU}(2,1)$ which is a non-split real form of $\mathrm{SL}(3, \mathbb C)$. We consider the BPS brane solutions of the theory preserving half of the supersymmetry and the action of $\mathrm{SU}(2,1)$ on them. Furthermore we provide evidence that the theory exhibits an underlying algebraic structure described by the Lorentzian Kac-Moody group $\mathrm{SU}(2,1)^{+++}$. This evidence arises both from  the correspondence  between the bosonic space-time fields of  $\mathcal{N}=2$ supergravity in $D=4$ and a one-parameter sigma-model based on the hyperbolic group $\mathrm{SU}(2,1)^{++}$, as well as from the fact that the structure of BPS brane solutions is neatly encoded in $\mathrm{SU}(2,1)^{+++}$.  As a nice by-product of our analysis, we obtain a regular embedding of the Kac-Moody algebra $\asuppp$ in $\mathfrak e_{11}$ based on brane physics. }\\

\newpage
\tableofcontents
\newpage
\setcounter{equation}{0}


\section{Introduction and discussion}

Many supergravity theories exhibit continuous global symmetries. These are of importance both for the generation of solutions to the field equations and also for the study of quantization. The symmetries are either inherent in the formulation of the theory, as in the case of type IIB supergravity in $D=10$ which admits an $\mathrm{SL}(2,\mathbb{R})$ symmetry, or act on certain subclasses of solutions admitting Killing vectors. This was first noticed in the case of non-supersymmetric $D=4$ gravity with one Killing vector by Ehlers in~\cite{Ehlers:1957zz} where it is also  $\mathrm{SL}(2,\mathbb{R})$ that acts on the set of solutions. Many other instances of this phenomenon are known, the most prominent being ${\cal N}=8$ supergravity in $D=4$ possessing an exceptional $\mathrm E_{7(7)}$ symmetry~\cite{Cremmer:1978ds}. This symmetry can also be viewed as a symmetry of the solutions of ${\cal N}=1$ supergravity in $D=11$ admitting seven commuting space-like Killing vectors, in agreement with the construction of the ${\cal N}=8$ theory in $D=4$ by dimensional reduction of the $D=11$ theory on a seven-torus $T^7$. 
It is also known that further dimensional reduction of the ${\cal N}=8$ theory to $D=3$ leads to an $\mathrm E_{8(8)}$ symmetry~\cite{Marcus:1983hb} and to an infinite-dimensional affine $\mathrm E_{9(9)}$ symmetry in $D=2$ ~\cite{Julia:1982gx,Nicolai:1987kz}. A longstanding conjecture is that upon further reduction to $D=1$ this yields the hyperbolic Kac--Moody symmetry group $\mathrm E_{10(10)}$~\cite{Julia:1982gx,Mizoguchi:1997si}. 

This idea has received new impetus in a modified form over the last years. Concretely, it has been suggested that there should be ways to reformulate the {\em unreduced} maximally supersymmetric $D=11$ theory such that it becomes invariant under $\mathrm E_{10(10)}$~\cite{Damour:2002cu} or even $\mathrm E_{11(11)}$~\cite{West:2001as}. In both cases, the symmetry acts non-linearly on the fields of the theory. An important difference, however, is the implementation of space-time in the two proposals. In the case of $\mathrm E_{10(10)}$ the ten spatial directions are thought to be rearranged in the infinitely many fields contained in the hyperbolic symmetry group $\mathrm E_{10(10)}$ and all the fields in the proposed non-linear sigma model only depend on a single parameter identified with time~\cite{Damour:2002cu}. We will refer to this as the cosmological $\mathrm E_{10(10)}$ model. In contrast, in the case of $\mathrm E_{11(11)}$, all fields depend on the eleven-dimensional space-time coordinates (or even additional coordinates implied by $\mathrm E_{11(11)}$ covariance~\cite{West:2003fc,Kleinschmidt:2003jf}). By embedding the cosmological $\mathrm E_{10(10)}$ model in a one-parameter sigma-model based on $\mathrm E_{11(11)}$ one can obtain variants of the model where the parameter is not identified with time but with a spatial direction. This model permits one to describe smeared BPS brane solutions~\cite{Englert:2003py,Englert:2004ph}. We refer to this model as the brane $\mathrm E_{10(10)}$ model.

This picture has been generalized to the case of any simple split symmetry group $\mathrm G$ in a theory of gravity coupled to matter in $D=3$ in~\cite{Englert:2003zs}. There is a general construction of extending the symmetry group $\mathrm G$ to an affine group $\mathrm G^+$, a hyperbolic or Lorentzian group $\mathrm G^{++}$ and a Lorentzian group $\mathrm G^{+++}$~\cite{Gaberdiel:2002db} and in the discussion above $\mathrm E_{10(10)}$ and $\mathrm E_{11(11)}$ have to be replaced by $\mathrm G^{++}$ and $\mathrm G^{+++}$ to obtain a set of more general conjectures for a wider class of theories. It has been verified that for all simple $\mathrm G$ the extended symmetry group describes the correct field content to make the conjecture work~\cite{Kleinschmidt:2003mf} but a full dynamical confirmation of the conjectures is still an open problem. Many aspects of these ideas have been discussed and instead of reviewing this work we refer the reader to~\cite{Englert:2007qb,Damour:2007dt,Riccioni:2009hi,Bergshoeff:2009zq} and references therein for further information.

In the present paper, we investigate these ideas in the context of four-dimensional pure ${\cal N}=2$ supergravity which is of interest for several reasons. The bosonic sector of this model consists of gravity coupled to a single Maxwell field. It admits half-BPS solutions like the extremal Reissner-Nordstr\"om black hole. In fact, a general half-BPS solution can be described by four charges $m$, $n$, $q$ and $h$ subject to the constraint $m^2+n^2=q^2+h^2$~\cite{Kallosh:1994ba,Argurio:2008zt}. The first two charges are gravitational mass and NUT charge and the latter two correspond to the electric and magnetic charges under the Maxwell field. In addition, it has been known for a long time that the solutions of this model with one Killing vector transform under the group $\mathrm G=\mathrm{SU}(2,1)$~\cite{Kinnersley}. This symmetry group is not in its split real form (which would be $\mathrm{SL}(3,\mathbb{R}))$ and one of our motivations for this work was to investigate whether the conjectures discussed above apply also in this case (see~\cite{HenneauxJulia,deBuyl:2003ub,Henneaux:2007ej,Riccioni:2008jz} for related work). In particular, the theory of real forms for the extended infinite-dimensional symmetries $\mathrm G^{++}$ and $\mathrm G^{+++}$ is not as well-developed as for finite-dimensional groups but see~\cite{Rousseau1989,Rousseau1995,BenMessaoud} for some mathematical results. Since the symmetry $\mathrm{SU}(2,1)$ mixes the two gravitational charges one can study the question of gravitational dualities~\cite{Hull:1997kt,Hull:2000zn,Henneaux:2004jw,Bunster:2006rt,Boulanger:2008nd,Bergshoeff:2009zq,Argurio:2008zt,Nieto:1999pn} analogous to electromagnetic duality in this simple model.

In more detail, we show the following in this paper. We first review some facts about pure ${\cal  N}=2$ supergravity in $D=4$ and the group $\mathrm{SU}(2,1)$ acting on its solutions with one time-like or 
one space-like Killing vector in Section~\ref{sec:EinsteinMaxwell}. Then, we go on to study the action of the finite-dimensional $\mathrm{SU}(2,1)$ on the BPS solutions in Section~\ref{sec:GroupAction}. There we show that the four charges transform linearly under the non-compact subgroup $\mathrm{SL}(2,\mathbb{R})\times \mathrm{U}(1)$ of $\mathrm{SU}(2,1)$. In particular, we show that the moduli space of half-BPS solutions can be described as a certain coset space, in agreement with recent results in the literature, and discuss the extension to the quantum theory from a string theory perspective. By analysing the Lie algebra of $\mathrm{SU}(2,1)^{+++}$ we then demonstrate that the field content of the extended symmetry group is correct in Section~\ref{sec:su21+++}. This requires understanding which generators are present in this particular real form of the Kac-Moody algebra. Starting from this observation, one can construct a correspondence between the one-parameter cosmological model based on $\supp$ and ${\cal N}=2$ supergravity in exactly the same way as for $\mathrm E_{10(10)}$ and this is shown in Section~\ref{sec:su21++}. We demonstrate also how the algebraic structure of $\suppp$ captures the half-BPS solutions in Section~\ref{sec:su21++}. This provides a detailed study of the proposed infinite-dimensional symmetries of ${\cal N}=2$. The extremal BPS solutions that occur in ${\cal N}=2$ supergravity can be derived from intersecting brane construction in M-theory and this leads us to an embedding of the non-split $\mathrm{SU}(2,1)^{+++}$ in the split $\mathrm  E_{11(11)}$, which is described in Section \ref{sec:Embedding}, thus nicely unifying our analysis with existing results. Questions not addressed in this paper are the supersymmetric deformations of $\mathcal{N}=2$ supergravity (e.g. adding a cosmological constant) and their consistency of the algebraic structure of $\suppp$ via higher rank forms~\cite{Bergshoeff:2007qi,Riccioni:2007au,Gomis:2007gb,Kleinschmidt:2008jj} as well as the coupling of the fermionic sector.

Our results can be taken as evidence that the conjectured $\mathrm{G}^{++}$ and $\mathrm{G}^{+++}$ also appear in situations when $\mathrm{G}$ is not in split real form. Their full verification is subject to the same restrictions regarding the correct interpretation of the infinity of their generators as in the case when $\mathrm{G}$ is split. One can establish a correspondence (or dictionary) between the cosmological coset model based on $\mathrm{G}^{++}$ and the supergravity equations at low levels and account for the algebraic structure of  half-BPS solutions in $\mathrm{G}^{+++}$. The finite $\mathrm{G}$ part of the symmetry acts as a solution generating group in $D=3$. In particular, there are non-linear transformations acting as gravitational dualities on BPS solutions. Furthermore, the construction of ${\cal N}=2$ supergravity as a truncation of the maximal ${\cal N}=8$ theory has an algebraic counterpart since $\mathfrak{su}(2,1)^{+++}$ is contained in $\mathfrak{e}_{11}$ as a subalgebra.

\setcounter{equation}{0}
\section{Symmetries and BPS solutions of pure $\mathcal{N}=2$  supergravity}
\label{sec:EinsteinMaxwell}

Pure $\mathcal{N}=2$ supergravity in four dimensions is the natural supersymmetric completion of Einstein-Maxwell theory. To set the scene, we shall in this section present our conventions for this theory, and in particular discuss its underlying symmetries in the presence of a space-like or a time-like  Killing vector. The presence of these Killing vectors is equivalent to performing a Kaluza-Klein reduction of the theory on a space-like or a time-like circle, respectively. This process reveals a hidden global symmetry in $D=3$, described by the group SU(2,1) \cite{Kinnersley,Breitenlohner:1987dg}. In this section, we also discuss some important properties of this group and its associated Lie algebra $\mathfrak{su}(2,1)$, which will be of importance in the remainder of this paper.
\subsection{Einstein--Maxwell in $D=4$}

The field content of four-dimensional $\, \mc{N}=2$ supergravity consists of a gravity multiplet, with a graviton $g_{\alpha \beta}$, two gravitino $\Psi^a_\alpha$ ($a=1,2$) and a Maxwell field $A_\alpha$. The bosonic part of the theory is then described by the standard Einstein-Maxwell Lagrangian:
\begin{eqnarray}
\label{eqn:SugraAction4d}
\mathcal{L}_{4d}=  \frac{1}{4}\, \sqrt{- g} \, \Big(R- F_{\alpha\beta} F^{\alpha \beta}\Big)\,, 
\end{eqnarray}
such that the Maxwell field is minimally coupled to gravity, and where $F_{\alpha \beta} = 2\,  \p_{\, [ \alpha} \,A_{\beta ]}$, locally. We will take space-time $\mathcal{M}_4$ to be Lorentzian with signature $(-, +, +, +)$.\footnote{Regarding index notation, Greek letters $\alpha, \beta...$ will indicate the four-dimensional curved space-time indices, $\mu, \nu...$ the three-dimensional curved indices, $A,B,...$ flat space-time indices, and $a,b...$ flat space indices.} The equations of motion derived from (\ref{eqn:SugraAction4d}), written in flat coordinates, are for the metric
\begin{equation}
\label{eqn:Einstein}
R_{AB} + \frac{1}{2} \eta_{AB} F_{CD} F^{CD} - 2 F_{AC}{F_B}^{C} = 0, 
\end{equation}
and for the Maxwell field
\begin{equation}
\label{eqn:Maxwell}
D^A F_{AB} = 0 .
\end{equation}
Here $D$ is the covariant derivative with respect to the spin connection. From the symmetry properties of the fields, we can derive the following Bianchi-identities for the Riemann-tensor and the field strength
\begin{eqnarray}
\label{eqn:MaxwellBianchi}
\epsilon^{ABCD}D_B F_{CD} &= &0 ,\\
\label{eqn:RiemannBianchi1}
\epsilon^{ABCD} R_{BCDE} &=& 0 .
\end{eqnarray}

For the analysis of finite symmetries, we will mainly be concerned with space-times preserving some subgroup of the diffeomorphism group of $\mathcal{M}_4$. These residual symmetries are described by the existence of Killing vectors $\kappa$. The Maxwell field will also preserve this symmetry if its Lie derivative with $\kappa$ vanishes. The dynamics of such solutions can be described from a three-dimensional perspective, formally reducing (\ref{eqn:SugraAction4d}) on the orbits of $\kappa$. This three-dimensional theory is then living on an orbit space $\mathcal{M}_3 =  \mathcal{M}_4/\Sigma$, where $\Sigma$ is the exponentiated action of $\kappa$ on $\mathcal{M}_4$.\footnote{Note that generally $\kappa$ will vanish on certain submanifolds of $\mathcal{M}_4$, and when defining its orbit space, we choose a component of $\mathcal{M}_4$ where $\kappa$ is non-vanishing and connected to infinity.}

In three dimensions vector fields have only one propagating degree of freedom and are hence equivalent to scalar fields. One can therefore dualise a vector -- using the Hodge star on the corresponding field strength -- to a scalar by explicitly imposing its Bianchi-identity and consequently write the three-dimensional Lagrangian only in terms of a metric and a set of scalars. For example, a four-dimensional stationary Maxwell-field will in three dimensions be described by two scalars (one from the component of the potential in the time-direction, and one from dualisation). As we will see below, this will concretely realize electromagnetic duality as well as a gravitational duality such as the Ehlers symmetry. As a consequence, the three-dimensional theory allows for a big set of symmetries, acting on the set of solutions preserving the given Killing vector. In fact, the moduli space of solutions  (almost) realizes a generally non-linear representation of this symmetry group. We will discuss this in more detail in Section \ref{sec:GroupAction}.

\subsection{$\mathrm{SU}(2,1)$ and coset models}
\label{SU(2,1)SigmaModel}

In the following, the group $\mathrm{SU}(2,1)$ and some of its subgroups will play an important role since $\mathrm{SU}(2,1)$ is the global symmetry group of Einstein-Maxwell theory in the presence of a Killing vector~\cite{Kinnersley}. We briefly discuss its definition and the construction of coset models with $\mathrm{SU}(2,1)$ symmetry, relegating more details and explicit expressions of the generators to the Appendices \ref{app:SigmaModel} and \ref{app:su21}.

In our conventions the group $\mathrm{SU}(2,1)$ is defined by the set of all unit-determinant complex $(3\times 3)$ matrices $g$ that preserve a metric $\eta$ of signature $(2,1)$;
\begin{equation}
{\mathrm {SU}}(2,1) = \left\{ g \in \mathrm{SL}(3, \mathbb{C}) \,:\, g^\dagger \eta g = \eta \right\}\quad\text{with}\quad
\eta = \left(\begin{array}{ccc}0&0&-1\\0&1&0\\-1&0&0\end{array}\right)\,,
\end{equation}
and we denote the associated Lie algebra by $\mathfrak{su}(2,1)$. This Lie algebra is a real form\footnote{We refer the reader to \cite{Helgason:1978,Henneaux:2007ej,deBuyl:2006gp} for introductions on real forms of complex Lie algebras.} of $\mathfrak{sl}(3,\mathbb{C})$ which may be described via the Tits-Satake diagram shown in Figure~\ref{fig:su21}. The labelling of nodes in Figure~\ref{fig:su21} is chosen to leave room for the extension of $\mathfrak{su}(2,1)$ to the Kac-Moody algebra $\mathfrak{su}(2,1)^{+++}$ to be discussed in Section~\ref{sec:su21+++}. 

\begin{figure}[t]
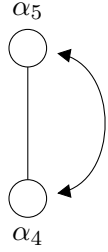

\begin{center}
\begin{pgfpicture}{1cm}{0cm}{1cm}{2cm}

\pgfnodecircle{Node1}[stroke]{\pgfxy(1,0.5)}{0.25cm}
\pgfnodecircle{Node2}[stroke]
{\pgfrelative{\pgfxy(0,2)}{\pgfnodecenter{Node1}}}{0.25cm}

\pgfnodebox{Node6}[virtual]{\pgfxy(1,0)}{$\alpha_{4}$}{2pt}{2pt}
\pgfnodebox{Node7}[virtual]{\pgfxy(1,3)}{$\alpha_{5}$}{2pt}{2pt}
\pgfnodeconnline{Node1}{Node2} 

\pgfsetstartarrow{\pgfarrowtriangle{4pt}}
\pgfsetendarrow{\pgfarrowtriangle{4pt}}
\pgfnodesetsepend{5pt}
\pgfnodesetsepstart{5pt}
\pgfnodeconncurve{Node2}{Node1}{-10}{10}{1cm}{1cm}

\end{pgfpicture}
\caption {\label{fig:su21} \sl \small The Tits-Satake diagram of the real form $\asu$. This real form is in one to one correspondence with a conjugation $\sigma$ of the complex Lie algebra $A_2=\mathfrak{sl}(3,\mathbb{C})$ which fixes completely $\asu$. The Tits-Satake diagram precisely gives  the action of  $\sigma$ on the simple roots of $A_2$. The presence of the double arrow means that $\sigma(\alpha_4)= \alpha_5$ and $\sigma(\alpha_5)= \alpha_4$. See Section \ref{sec:su21+++} and Appendix \ref{app:su21} for more details.} 
\end{center}
\end{figure}
The Lie algebra $\mathfrak{su}(2,1)$ has two maximal subalgebras that will play a central role in what follows. The first one is the maximal compact subalgebra $\mathfrak{k}=\mathfrak{su}(2)\oplus \mathfrak{u}(1)$, defined by the subset of generators which are pointwise fixed by the so-called Cartan involution $\theta$:
\begin{equation}
\mathfrak{k}=\mathfrak{su}(2)\oplus \mathfrak{u}(1)=\{ x\in \mathfrak{su}(2,1)\,: \, \theta(x)=x\}.
\end{equation}
The other (non-compact) maximal subalgebra is $\mathfrak{k}^{*}=\mathfrak{sl}(2,\mbb{R})\oplus \mathfrak{u}(1)$, which is similarly defined with respect to a ``temporal involution'' $\Omega_4$. The two involutions $\theta$ and $\Omega_4$ are discussed in more detail in Appendices \ref{app:k}, \ref{app:k*} and in Section \ref{sec:carttemporalin}. The Cartan involution induces the following Cartan decomposition in terms of vector spaces (see, e.g.,~\cite{Helgason:1978})
\begin{equation}
\mathfrak{su}(2,1)=\mathfrak{k}\oplus \mathfrak{p},
\end{equation}
where $\mathfrak{p}$ is the subspace which is anti-invariant under $\theta$, corresponding to the orthogonal complement of $\mathfrak{k}$ with respect to the Killing form on $\mathfrak{su}(2,1)$. Similarly, the temporal involution $\Omega_4$ induces the analogous decomposition
\be \mathfrak{su}(2,1)=\mathfrak{k}^{*}\oplus \mathfrak{p}^{*}.
\ee 
Note that $\mathfrak{p}$ and $\mathfrak{p}^{*}$ transform respectively  in representations of $\mathfrak{k}$ and $\mathfrak{k}^{*}$ but are not subalgebras of $\mathfrak{su}(2,1)$. For later reference, let us also give another useful decomposition of $\mathfrak{su}(2,1)$, known as the algebraic Iwasawa decomposition in terms of vector spaces
\begin{equation}
\mathfrak{su}(2,1)=\mathfrak{k}\oplus \mathfrak{a} \oplus \mathfrak{n}_+,
\label{Iwasawa}
\end{equation}
where $\mathfrak{a}$ is the non-compact part of the Cartan subalgebra $\mathfrak{h}\subset \mathfrak{su}(2,1)$ and $\mathfrak{n}_+$ is the nilpotent subspace of upper-triangular matrices. The subspace $\mathfrak{b}_+=\mathfrak{a} \oplus \mathfrak{n}_+$ is known as the standard Borel subalgebra.

At the group level, we then have the corresponding maximal compact subgroup  
\begin{equation}
\mathrm{K} = \mathrm{SU}(2)\times \mathrm{U}(1)
\end{equation}
and non-compact subgroup 
\begin{equation}
\mathrm{K}^* = \mathrm{SL}(2, \mathbb{R})\times \mathrm{U}(1)\, .
\end{equation}
Similarly, the subspaces $\mathfrak{p}$ and $\mathfrak{p}^{*}$ correspond to the two coset spaces
\begin{equation}
\mathcal{C} = \mathrm{G}/\mathrm{K} = \frac{\mathrm{SU}(2,1)}{\mathrm{SU}(2)\times \mathrm{U}(1)}
\quad\text{and}\quad
\mathcal{C}^* = \mathrm{G}/\mathrm{K}^* = \frac{\mathrm{SU}(2,1)}{\mathrm{SL}(2, \mathbb{R})\times \mathrm{U}(1)}\,.
\end{equation}
Physically, $\mathcal{C}$ and $\mathcal{C}^{*}$ arise, respectively, as the moduli spaces of scalars upon reduction to $D=3$ of the Einstein-Maxwell Lagrangian on a space-like or a time-like circle. 

The coset space $\mathcal{C}=\mathrm{G}/\mathrm{K}$ is a Riemannian symmetric space of dimension $\dim(\mathrm{G})-\dim(\mathrm{K})=4$, matching the combined number of degrees of freedom contained in the metric and the Maxwell field in $D=4$. To describe a coset model on this space one can choose a map $\mathcal{V}: \mathcal{M}_3  \rightarrow \mathrm{G}/\mathrm{K}$ in a fixed gauge, that transforms under global transformations $g\in\mathrm{G}$ as $\mathcal{V}(x) \to k(x) \mathcal{V} g^{-1}$, where $k(x)\in\mathrm{K}$ is a local compensating transformation required to restore the chosen gauge for the coset representative.

A manifestly SU(2,1)-invariant Lagrangian can now be constructed using the Maurer-Cartan form $\mathrm{d}\mathcal{V} \mathcal{V}^{-1}$ as follows. Its projection
\begin{equation}
\mathcal{P}= \frac{1}{2} \big(\dd\mathcal{V} \mathcal{V}^{-1}\,  - \,  \theta(\dd\mathcal{V} \mathcal{V}^{-1})\big) \ \in \mathfrak p
\end{equation}
 along the coset transforms $\mathrm{K}$-covariantly under the global $\mathrm{G}$ action as $\mathcal{P}\to k \mathcal{P}k^{-1}$ and the (invariant) trace of its square can be used as a $\mathrm{G}$-invariant Lagrangian that is second order in derivatives: 
 \begin{equation}
 \mathcal{L}=\sqrt{|h|}h^{\mu\nu}  (\mathcal{P}_{\mu} | \mathcal{P}_{\nu}) ,
 \end{equation}
 where $h_{\mu\nu}$ is the metric on $\mathcal{M}_3$. To make this concrete, we shall extend the decomposition (\ref{Iwasawa}) to the group level using the Iwasawa theorem, so that
\begin{equation}
\mathrm{SU}(2,1)=\mathrm{KAN},
\end{equation}
where $\mathrm A$ is the abelian group with the Lie algebra $\mathfrak{a}$ and $\mathrm N$ is the nilpotent group corresponding to the subspace $\mathfrak{n}_+$.
This ensures that we may choose a coset representative $\mathcal{V}\in \mathrm{AN}$ of upper-triangular matrices, traditionally referred to as the ``Borel gauge''. As a consequence of this gauge choice, the coset element  $\mathcal{V}$ can be parametrized using four scalar fields (coordinates on G/K), to be illustrated in detail below.

On the other hand, the coset space $\mathcal{C}^*=\mathrm{G}/\mathrm{K}^*$ is a homogeneous space of dimension $\dim(\mathrm{G})-\dim(\mathrm{K}^*)=4$, but is no longer a Riemannian manifold. Rather, it has signature $(2,2)$, and is usually referred to as a ``pseudo-Riemannian'' symmetric space \cite{Breitenlohner:1987dg}. The general construction of an SU(2,1)-invariant coset model discussed above is still applicable, although in this case the global Iwasawa decomposition is no longer valid. 

For both choices of subgroup, $\mathrm{K}$ and $\mathrm{K}^*$, there is a Noether current 
\begin{equation}
\label{eqn:Noether}
\mathcal J^{\mu} =  \sqrt{|h|} h^{\mu \nu} \mathcal{V}^{-1} \mathcal P_{\nu} \mathcal{V} ,
\end{equation}
associated to the global $\mathrm{G}$ symmetry. We will see later that its values ``at infinity''  for $\mathcal{V}$ describing a half-BPS solutions can be related to the four charges describing the most general such solution.

\subsection{Solutions with space-like Killing vector}
\label{sec:SpacelikeKillingVector}

To give a flavor of the relevance of coset models, we will first quickly consider the case of solutions for which the preserved Killing vector $\kappa$ is space-like. After choosing suitable coordinates so that $\kappa = \partial_x$, the reduction of the Einstein-Maxwell Lagrangian (\ref{eqn:SugraAction4d}) yields in three dimensions, after dualisation, an Einstein plus scalar Lagrangian, with scalar part given by:
\be \label{eqn:lag3ds} \begin{split}
\mathcal{\tilde{L}}_{\mathrm{scal}}= -\,  \frac{\sqrt{ -h}}{4}\,  \Big( \frac{1}{2}\,  \p_{ \m} \phi\,  \p^{\m} \phi &+ 2\,  e^{ \phi} (\p_{ \m} \chi_e\,  \p^{\m} \chi_e+ \p_{ \m} \chi_m\,  \p^{\m} \chi_m) \\ 
&+ e^{2 \phi} (\p_{ \m} \psi +  \sqrt{2}\,  \chi_m \, \p_{\mu} \chi_e - \sqrt{2} \, \chi_e\, \p_{\mu} \chi_m)^2 \Big).
\end{split}
\ee

This reduced Lagrangian contains a dilaton $\phi$ and three axions: $\chi_e$ coming from the component $A_x$ of the Maxwell vector potential, $\chi_m$ coming from the dualisation of the Maxwell vector potential in three dimensions and $\psi$ arising from the dualization of the graviphoton. 
These four scalar fields parametrize the coset  space $\mathcal{C} = \su/(\mathrm{SU}(2)\times \mathrm{U}(1))$ \cite{Kinnersley,Julia:1980gr}. More concretely, the scalar dynamics given by (\ref{eqn:lag3ds}) is equivalent to a non-linear $\sigma$-model describing maps $\mathcal{V}$ from $\mathcal{M}_3$ to the homogeneous space $\mathcal{C}$ as described above. 
This map $\mathcal{V}$ is the composition of a map into the tangent space of $\mathcal{C}$ together with the exponential map from this tangent space to the coset. Naturally parametrized by the four scalar fields, it is therefore given by the expression
\be
\label{eqn:CosetElement} \begin{split}
\mathcal{V}&= e^{\frac{1}{2} \phi\, h_4}\, e^{\sqrt{2} \,\chi_e\,   e_4 + \sqrt{2} \,\chi_{m} \,  e_5 + \sqrt{2}\, \psi \,  e_{4,5}}\\
&= \left(\begin{array}{ccc}e^{\frac{\phi}{2}} & e^{\frac{\phi}{2}}(\sqrt{2} \, \chi_e + i \, \sqrt{2}\,  \chi_m) & e^{\frac{\phi}{2}}(\chi_e^{\  2} +  \chi_m^{\  2 }+ i\,  \sqrt{2} \, \psi) \\0 & 1 & \sqrt{2} \, \chi_e - i \, \sqrt{2}\,  \chi_m \\0& 0 & e^{- \frac{\phi}{2}} \end{array}\right)\, ,
\end{split}
\ee
where $  h_4,\,   e_4,\,   e_5$ and  $  e_{4,5}$ (see (\ref{eqn:su21gens})) are the generators of the Borel subalgebra $\mathfrak{b}_+=\mathfrak{a}\oplus \mathfrak{n}_+$.\footnote{Recall from the previous section that the Borel subalgebra $\mathfrak{b}$ of a non-split real form does not contain the full Cartan subalgebra $\mathfrak{h}$ but only its non-compact part $\mathfrak{a}$ (see Appendix \ref{app:su21} for more details). }  The scalar fields here depend on the coordinates $x^\mu$ of $\mathcal{M}_3$.

\subsection{Solutions with time-like Killing vector}
\label{sec:StationarySolutions}
Let us now repeat this discussion in a little bit more detail in a case that will be more interesting for us, namely BPS solutions. These solutions  preserve a time-like Killing vector, which is most easily seen by considering the fact that a BPS-solution necessarily preserve a Killing spinor $\epsilon$. Forming the supersymmetry generator $Q = Q^{\mu}\epsilon_{\mu}$, the supersymmetry algebra -- in which $Q$ squares to the generator of time-translation -- shows that the solution must be preserved under time-translation. The set of single centered BPS-solutions is a subset of a general set of generalized Reissner-Nordstr\"om black hole-like solutions to the equations of motion (\ref{eqn:Einstein}) and (\ref{eqn:Maxwell}) with mass $m$, NUT charge $n$, and electric and magnetic Maxwell charges $q$ and $h$. This solution can be written as~\cite{Argurio:2008zt}
\be\label{eqn:TaubNUTsolution} \begin{split}
\dd s^2 =& - \frac{{\tilde r}^2-n^2- 2\,m\,{\tilde r}+q^2+h^2}{{\tilde r}^2+n^2}\,  ( \, \dd t + 2 \,n \,\cos \theta\, \dd \phi \,)^2  \\ 
&+ \frac{\tilde{r}^2+n^2}{{\tilde r}^2-n^2- 2\,m\, {\tilde r}+q^2+h^2}\, \dd {\tilde r}^2 + ({\tilde r}^2+n^2)(\dd \theta^2+ \sin^2 \theta \, \dd \phi^2)\, , 
\end{split}
\ee
\begin{eqnarray} \label{eqn:TaubNUTsolution2}
& A_t& = \frac{q\,{\tilde r} +n\,h}{{\tilde r}^2+n^2}, \hspace{1cm}   A_{\phi}= \frac{2\, n\, q\,  {\tilde r} - h\, ({\tilde r}^2-n^2)}{{\tilde r}^2 +n^2}\, \cos \theta \, .
\end{eqnarray}
For our purpose it is convenient to introduce isotropic coordinates defined by ${\tilde r}= r+ m$, in which  the metric in (\ref{eqn:TaubNUTsolution}) becomes
\begin{equation}
\dd s^2=-\frac{\lambda(r)}{R^2(r)} (\dd t+2\,n \, \cos \theta \, \dd \phi \, )^2+ \frac{R^2(r)}{\lambda(r)} \dd r^2 +
R^2(r) (\dd \theta^2+ \sin^2 \theta \, \dd \phi^2)\, .
\ee
Here the functions $\lambda(r)$ and $R^2(r)$ are given by 
\begin{eqnarray}
&&\lambda(r)= r^2-\Delta^2, \qquad \Delta^2 \equiv m^2+n^2-q^2-h^2 \label{lambda}, \\
&& R^2(r)=r^2+2mr+m^2+n^2 \label{R2} .
\end{eqnarray}

We will be interested in the subclass of solutions (\ref{eqn:TaubNUTsolution}) which are BPS, namely the ones which preserve 1/2 of the supersymmetry. These solutions are characterised by the following  BPS condition among the four charges \cite{Kallosh:1994ba}  (see also \cite{Argurio:2008zt}):
\begin{equation}
\label{BPScondp}
\Delta^2 =0   \quad \Leftrightarrow \quad m^2+n^2=q^2+h^2.
\end{equation}
Using  (\ref{lambda}) and (\ref{BPScondp}), we have $\lambda_{BPS}=r^2$ and the BPS metric is:
\begin{equation}
\label{Isom}
\dd s^2_{BPS}=-\frac{r^2}{R^2(r)} (\dd t+2\,n \, \cos \theta \, \dd \phi \, )^2 + \frac{R^2(r)}{r^2} \sum_{a=1}^{3}\dd x_a^2,
\end{equation} 
where the $x_a$'s are the flat Euclidean space coordinates.
In the particular case $n=0$, one finds again the extremal Reissner-Nordstr\"om black hole (or an extremal $0$ brane in four dimensions) characterised by the harmonic function $1+\frac{m}{r}$ (see for instance \cite{Argurio:1998cp}). Note that upon dimensional reduction of (\ref{eqn:TaubNUTsolution}) on time the three-dimensional Euclidean metric is
\begin{equation}
\dd s^2_{3D}= \dd r^2 + (r^2-\Delta^2) \dd \Omega_2^2.
\label{redu3d}
\end{equation}
When the BPS condition (\ref{BPScondp}) is fulfilled, the line element (\ref{redu3d}) is just the line element for three-dimensional Euclidean flat space in spherical coordinates.

This solution is an example of a solution to the Einstein-Maxwell system which allows for a time-like Killing vector. 
These solutions are referred to as {\it stationary}. More generally, choosing suitable coordinates, in this case so that $\kappa = \partial_t$, a convenient metric ansatz for this type of solutions is\footnote{In the original work on time-like reductions to $D=3$ \cite{Breitenlohner:1987dg}, it was noted that by further assuming spherical symmetry for the three-dimensional metric $g_3$, the Einstein-scalar Lagrangian in $D=3$ describes the geodesic motion of a fiducial particle moving on (a cone over) the moduli space $\mathcal{C}^{*}=\mathrm{SU(2,1)}/(\mathrm{SL}(2,\mathbb{R})\times \mathrm{U(1)})$. The dynamics of the particle on $\mathcal{C}^{*}$ thus corresponds to motion in the space of stationary, spherically symmetric solutions to Einstein-Maxwell theory. This point of view has been extended recently in \cite{Gaiotto:2007ag,Bergshoeff:2008be} in the context of solution-generating techniques. Moreover, for the special case of BPS solutions this philosophy was elaborated upon in \cite{Pioline:2006ni,Neitzke:2007ke}, where it has been shown that the classical phase space of the particle dynamics coincides with the six-dimensional coset space $\mathcal{Z}=\mathrm{SU(2,1)/(U(1)}\times \mathrm{U(1)})$ (known as the twistor space of $\mathcal{C}$, see e.g. \cite{Gunaydin:2007qq}). This result has been used as a starting point for quantizing BPS black hole solutions by (radial) quantization of the particle dynamics on $\mathcal{C}^{*}$ \cite{Gunaydin:2007bg}.} 
\begin{eqnarray}
\label{redukill}
\dd s^{2}_{4D} = - e^{- \phi} (\dd t + \omega)^2+ e^{\phi} \dd s_{3D}^2\, ,
\end{eqnarray}
where $\dd s^2_{3D}$ is the invariant line element on $\mathcal{M}_3$ corresponding to the pullback metric $g_3$ of the four dimensional metric $g$ to $\mathcal{M}_3$.\footnote{In the case of extremal solutions this will give a flat $g_3$ metric, and for extremal solutions with horizon this will generally give $\mathcal{M}_3$ a topology homeomorphic to $\mathbb{R}^3\backslash \{0\}$.} Here we see explicitly the origin of the graviphoton, given as the 1-form $\omega$.  The dynamics of these particle-like solutions is now governed by (\ref{eqn:SugraAction4d}) reduced on the orbits of the Killing vector $\kappa$. Explicitly
\be \label{eqn:3dlagt} \begin{split}
\mathcal{L}_{3D}' =  \frac{1}{4} \sqrt{g_3} \,\Big[\ &\big(R_3 - \tfrac{1}{2}  \p_{ \m} \phi\,  \p^{\m} \phi+ \tfrac{1}{4} e^{-2 \phi} \, F^2_{(2)} \big)\\
&-\, \big(e^{- \phi} \tilde{F}^2_{(2)} - 2\,  e^{ \phi}\, \p_{ \m} \chi_e\,  \p^{\m} \chi_e  \big)   \   \Big]\, ,
\end{split}\ee
where $F_{(2)}= \dd\omega$ and $\tilde{F}_{(2)}= \dd A- \dd\chi_e\wedge \omega  $. After dualization of the two field strengths,
\begin{eqnarray}
\label{eqn:DualisingMaxwell}
\tilde{F}^{\lambda \nu }&=& - \frac{ \epsilon^{\mu \lambda \nu}}{\sqrt{g_3}}\, e^{\phi}\, \p_{\mu}\chi_m\, , \\
\label{eqn:DualisingEinstein}
F^{\lambda \nu }&=&  \frac{ \epsilon^{\mu \lambda \nu}}{\sqrt{g_3}}\, e^{2 \phi}\, \Big( 2\, (\chi_m\, \p_{\mu} \chi_e - \chi_e\,\p_{\mu} \chi_m ) + \sqrt{2} \, \p_{\mu} \psi \Big)\, ,
\end{eqnarray}
with $ \epsilon^{r \theta \phi} =-1 $, we can  rewrite the three-dimensional Lagrangian (\ref{eqn:3dlagt}) using only the metric in three dimensions and the scalar fields $\phi, \chi_e, \chi_m$ and $\psi$. We hence get
\be  \label{eqn:3dStationaryLagrangian} \begin{split}
 \mathcal{L}_{3D}=  \frac{1}{4}\, \sqrt{g_3} \, \Big[ \ R_3 &- \tfrac{1}{2}\,  \p_{ \m} \phi\,  \p^{\m} \phi + 2\,  e^{ \phi}\, (\p_{ \m} \chi_e\,  \p^{\m} \chi_e+ \p_{ \m} \chi_m\,  \p^{\m} \chi_m) \\
&- e^{2 \phi}\, (\p_{ \m} \psi +  \sqrt{2}\,  \chi_m \, \p_{\mu} \chi_e - \sqrt{2} \, \chi_e\, \p_{\mu} \chi_m)^2 \, \Big]\, .
\end{split}
\ee
One sees directly that $\chi_e$ and $\chi_m$ appear completely symmetrically. The duality between the two gravitational scalars is less apparent, but is in fact present as we will see in Section \ref{sec:GroupAction}. Note the change of signs in front of the kinetic terms for the Maxwell scalars in comparison with \eqref{eqn:lag3ds}, revealing that the two scalar actions for space-like and time-like reductions are related by a ``Wick rotation'' of the Maxwell scalars $\chi_e$ and $\chi_m$.

The scalar part of (\ref{eqn:3dStationaryLagrangian}) can now be identified with a non-linear $\sigma$-model constructed on the coset $\mathcal{C}^* = \su/(\mathrm{SL}(2, \mathbb{R})\times \mathrm{U}(1))$, where the change in the quotient group has its origin in the different kinetic terms. Hence the construction of this theory is the same as the one used for space-like reduction, except that when deriving the coset Lagrangian we replace the Cartan involution $\theta$ by the temporal involution $\Omega_4$ (introduced above and defined in Appendix \ref{app:k*} ), having as fixed subalgebra $\mathfrak{k}^* = \mathfrak{sl}(2,\mathbb{R}) \oplus \mathfrak{u}(1)$.

Generally we still write maps $\mathcal{V}$ in the Borel gauge by using the expression (\ref{eqn:CosetElement}). The solution (\ref{eqn:TaubNUTsolution}) and (\ref{eqn:TaubNUTsolution2})  in terms of our four scalar fields can then be rewritten as:
\be \begin{split}\label{eqn:BPSscalars}
\phi(r)&= \ln \Big(\frac{(r+m)^2 + n^2 }{ r^2}\Big)\, , \\
\chi_e(r)&= \frac{h n + q (m+r)}{(r+m)^2 + n^2 } \,,  \\
\chi_m(r)&= \frac{nq - h(m+r)}{(r+m)^2 + n^2}\, , \\
\psi(r)&=   -\frac{\sqrt{2} n r}{(r+m)^2 + n^2}\, .
\end{split}
\ee
Using the Noether current of the coset Lagrangian, we can now relate the three-dimensional ``conserved'' $\sigma$-model quantities to the four-dimensional charges. At infinity ($r\to\infty$) the coset element $\mathcal{V}$ parametrized by (\ref{eqn:BPSscalars}) tends towards the identity element ${\bf{1}}$ of $\mathrm{SU}(2,1)$, implying that $\mathcal{J} \rightarrow \mathcal{P}$ when $r\to\infty$. 

Furthermore, one can compute using \eqref{eqn:Noether} (see also \cite{Bossard:2009at})
\begin{equation}\label{eqn:charges} \begin{split} \
Q &= \int_{S^{\infty}} \mathcal{J} =  \int_{S^{\infty}} \mathcal{P} \\
&=  - m\, h_4+ n\, (e_{4,5}- f_{4,5})- \frac{q}{\sqrt{2}} \, (e_4- f_4) +\frac{h}{\sqrt{2}}\, (e_5+f_5) \\
&= \left(\begin{array}{ccc}-m & \frac{i(h+iq)}{\sqrt{2}} & in \\\frac{ih + q}{\sqrt{2}} & 0 & \frac{-ih-q}{\sqrt{2}} \\-in & \frac{-ih+q}{\sqrt{2}} & m\end{array}\right),
\end{split}
\end{equation} 
where $m,n,q$ and $h$ are the four-dimensional charges. For the derivation of the elements in $\mathfrak{p}^*$ which is the orthogonal complement of $\mathfrak{k}^*$ with respect to the Killing form, see Appendix \ref{app:su21}.

The form (\ref{eqn:charges}) is preserved by coset transformations belonging to $\mathrm{K}^*$ since $\mathcal{J}=\mathcal{P}\in \mathfrak{p}^*$ and the reductive homogeneous space decomposition ensures that $[\mathfrak{k}^*,\mathfrak{p}^*]\subset\mathfrak{p}^*$. As we will argue below, the transformations from $\mathrm{K}^*$ preserve the asymptotic conditions on the BPS solution and therefore act (linearly) on the four BPS charges. The transformations that belong to $\mathrm{K}^*$ also preserve the asymptotic condition $\mathcal{V}\to \bf{1}$ as $r\to\infty$. The $\mathrm{SU}(2,1)$ transformations that are not part of $\mathrm{K}^*$ violate this asymptotic condition on the coset element and also map the Noether current $\mathcal{J}$ out of $\mathfrak{p}^*$. This makes the identification of the four BPS charges from the Noether current less evident. However, as the physical fields are related to the scalar fields of the coset $\mathcal{C}^*$ (mostly) by duality relations these transformations do preserve the asymptotic conditions on the physical fields. In fact, we will show below that  Iwasawa decomposable  transformations of $\mathrm{SU}(2,1)$ outside $\mathrm{K}^*$ act as gauge transformations on the physical fields and do not change the BPS charges. For this reason it will turn out to be sufficient to use the Noether charges from (\ref{eqn:charges}) and their transformation under $\mathrm{K}^*$ to find the orbits of BPS solutions under $\mathrm{SU}(2,1)$.

\setcounter{equation}{0}
\section{Action of $\mathrm{SU}(2,1)$ on BPS solutions}
\label{sec:GroupAction}

Let us now proceed to discuss the action of SU(2,1) on the stationary solution (\ref{eqn:TaubNUTsolution}). The $\sigma$-model is $\mathrm G$-invariant by construction and acting with $\su$ on the coset with the natural action from the right, gives an action on the maps $\mathcal{V}$. We thus generate new solutions when lifting the transformed $\mathcal{V'}$ back to four dimensions, using the explicit form (\ref{eqn:CosetElement}) of $\mathcal{V}$ and the dualisation relations (\ref{eqn:DualisingMaxwell}) and (\ref{eqn:DualisingEinstein}). Furthermore, we know that every single centered extremal solution is uniquely fixed by the values of the scalar fields at infinity, in terms of the four charges $m,\, n,\, q$ and $h$. This induces a representation of $\su$ on these four charges. Now, as the coset space $\mathcal{C}^*$ in the case of stationary solutions is not a Riemannian symmetric space, there is not a single coordinate system covering the whole coset \cite{Bossard:2009at}. However, our $\sigma$-model describes maps to a given coordinate patch. If the action of $\mathrm G$ takes us outside of this patch, we have no way of relating the new $\mathcal{V}'$ to the four-dimensional fields. Constructing our coordinate system on $\mathcal{C}^*$ via the Borel gauge (i.e. treating $\mathcal{V}$ as the composition of the exponential map and a map from $\mathcal{M}_3$ to $\mathfrak{a}\oplus \mathfrak{n}_+$), we will only consider the subspace of $\mathrm G$ where the elements are decomposable in Iwasawa form. These elements are exactly the ones that preserve our coordinate patch. Hence we can consider the action of $\su$ in three different cases, one for each of the subgroups $\mathrm N,\mathrm A$ and $\mathrm K^*$ in the local Iwasawa decomposition.\footnote{A similar analysis was recently done for five-dimensional minimal supergravity which gives rise to a $\mathrm G_2$ $\sigma$-model when  this theory is reduced on two commuting Killing vectors \cite{Compere:2009zh}.} Our four-dimensional interpretation will differ in all of these cases. Solution generation in the case of Einstein-Maxwell theory has been considered also in~\cite{Neugebauer:1969wr,Kinnersley,Bossard:2009at}.

\subsection{Action of the nilpotent generators}

Let us begin with the analysis of the nilpotent group $\mathrm N$. As elements in $\mathrm N$ do not take us outside of the Borel gauge, the analysis of how the scalar fields change is simply done by multiplying $\mathcal{V}$ described by (\ref{eqn:CosetElement}) from the right by elements in the group $\mathrm N$ of nilpotent elements, i.e. if $n \in \mathrm N$, $\mathcal{V} \rightarrow \mathcal{V}' = \mathcal{V}n^{-1}$.

As described in Appendix \ref{app:restricted}, the Lie algebra $\mathfrak{n}_+$ of $\mathrm N$ is generated by the three nilpotent generators $e_4,e_5$ and $e_{4,5}$. Under the three corresponding nilpotent 1-parameter subgroups (with real parameters $\alpha,\beta$ and $\gamma$), the scalar fields given by (\ref{eqn:BPSscalars}) transform as follows; Under the group generated by $e_4$: 
\be \begin{split} 
\chi_e &\rightarrow \ \chi_e - \frac{1}{\sqrt{2}} \alpha\, ,\\
\psi &\rightarrow \ \psi - \alpha \chi_m .
\end{split}\ee
Under the group generated by $e_5$,
\be \begin{split} 
\chi_m &\rightarrow \ \chi_m - \frac{1}{\sqrt{2}} \beta \, ,\\
\psi &\rightarrow \  \psi + \beta \chi_e .
\end{split}\ee
Finally, under the group generated by $e_{4,5}$,
\begin{equation}
\psi \rightarrow \ \psi -  \frac{1}{\sqrt{2}}\gamma .
\ee
Looking at the dualisation relations (\ref{eqn:DualisingMaxwell}) and (\ref{eqn:DualisingEinstein}) we see that these transformations simply vanish when lifting the fields back to four dimensions. We can hence interpret the symmetry group $\mathrm N$ as appearing from realizing an inherent redundancy in the formulation of the three-dimensional theory, and is therefore not visible in four dimensions. Equivalently, the action of the nilpotent group $\mathrm N$ corresponds to gauge transformations.

\subsection{Action of the non-compact Cartan generator}

The action of the abelian group $\mathrm A$, with Lie algebra $\mathfrak{a}$, is generated by the non-compact Cartan generator $h_4 \in \mathfrak{p}^*$. The action of $\mathrm A$, parametrized by $d \in \mathbb{R}$ is
\be \begin{split} 
\phi &\rightarrow \  \phi-2d \, , \\
\chi_e& \rightarrow \ e^{d} \chi_e \, , \\
\chi_m& \rightarrow \  e^{d} \chi_m \, , \\
\psi &\rightarrow \  e^{2d} \psi\, .
\end{split}\ee
Lifting this transformation back into the four-dimensional metric and Maxwell potential, we see that it is just a coordinate transformation coming from a rescaling of the time and space coordinates $t \rightarrow e^d t$ and $r \rightarrow e^{-d} r$. The solution is therefore unchanged. 

\subsection{Action of $\mathrm K^*$}
\label{sec:ActionOnCharges}

We have now discussed from a physical perspective why the generators in $\mathrm A \mathrm N \subset \mathrm G$ act trivially on a given solution. By restricting to transformations that stay in our coordinate patch on $\mathrm G/\mathrm K^*$, what is left to consider is now the non-compact group $\mathrm K^*$. It turns out that it is $\mathrm K^*$ that realizes electromagnetic and gravitational duality. From the expression of the Noether charge (\ref{eqn:charges}), we see that $\mathrm K^*$ transforms non-trivially the set of conserved four-dimensional charges, and it is natural to ask precisely how this action is realized. This is done by extracting the transformed charges as the coefficients in front of the generators of $\mathfrak{p}^*$ just as in the expression (\ref{eqn:charges}). The algebra $\mathfrak{k}^*$ of $\mathrm{K}^*$ is generated by the elements $u$, and $t_i$, $i=1,2,3$, where the $t_i$'s generate an $\mathfrak{sl}(2,\mathbb{R})$, commuting with $u$. The definition of $\mathfrak{k}^*$ is described in Appendix \ref{app:k*}. Treating these four Lie algebra generators separately, as we did in the case of $\mathrm N$ above, we find that the 1-parameter subgroup generated by $u$, with parameter $a$ generates the transformation
\begin{equation}
\left(\begin{array}{c}m \\n \\q \\h\end{array}\right) \rightarrow \left(\begin{array}{cccc}\cos(a) & \sin(a) & 0 & 0 \\-\sin(a) & \cos(a) & 0 & 0 \\0 & 0 & \cos(a) & \sin(a) \\0 & 0 & -\sin(a) & \cos(a)\end{array}\right) \left(\begin{array}{c}m \\n \\q \\h\end{array}\right) \, ,
\end{equation}
under finite transformations generated by $t_1$ with parameter $b$,
\begin{equation}
\left(\begin{array}{c}m \\n \\q \\h\end{array}\right) \rightarrow \left(\begin{array}{cccc}\cos(b) & -\sin(b) & 0 & 0 \\\sin(b) & \cos(b) & 0 & 0 \\0 & 0 & \cos(b) & \sin(b) \\0 & 0 & -\sin(b) & \cos(b)\end{array}\right) \left(\begin{array}{c}m \\n \\q \\h\end{array}\right) \, ,
\end{equation}
under $t_2$ with parameter $c$,
\begin{equation}
\left(\begin{array}{c}m \\n \\q \\h\end{array}\right) \rightarrow \left(\begin{array}{cccc}\cosh(c) & 0 & -\sinh(c) & 0 \\0 & \cosh(c) & 0 & -\sinh(c) \\-\sinh(c) & 0 & \cosh(c) & 0 \\0 & -\sinh(c) & 0 & \cosh(c)\end{array}\right) \left(\begin{array}{c}m \\n \\q \\h\end{array}\right) \, ,
\end{equation}
and finally under $t_3$ with parameter $d$,
\begin{equation}
\left(\begin{array}{c}m \\n \\q \\h\end{array}\right) \rightarrow \left(\begin{array}{cccc}\cosh(d) & 0 & 0 &\sinh(d)  \\0 & \cosh(d) & -\sinh(d)  & 0\\ 0 & -\sinh(d)  & \cosh(d) & 0 \\\sinh(d) \ &0 & 0 & \cosh(d)\end{array}\right) \left(\begin{array}{c}m \\n \\q \\h\end{array}\right) .
\end{equation}

We see here that $\mathrm K^*$ realizes a linear representation $R$ on the charges, identified with $R= {\bf 2}_{1} \oplus {\bf 2}_{-1}$ (decomposed with respect to $\mathfrak{sl}(2,\mathbb{R})\oplus \mathfrak{u}(1)$, where the subscript indicates the charge under $\mathrm{U}(1)$), acting as two boosts and two rotations. In particular we see that $u + t_1= \tfrac{2}{3}\, h_5$ acts as a rotation of electric and magnetic charges. This will be in agreement with further discussion in Section \ref{sec:su21+++}, considering the commutation relations (see (\ref{comutenc45}))
\begin{equation}
[h_5,r^a] =3\, \tilde{r}_a \quad, \quad [h_5,\tilde{r}^a] =- 3\, r_a\, ,
\end{equation} 
and the identification in the dictionary (Table \ref{tab:Dictionary}) stating that the generators $r^a$ and $\tilde{r}^a$ correspond to the electric and magnetic parts of the Maxwell field. We can also see that $u-t_1$ acts as gravitational duality rotation (see for instance \cite{Bunster:2006rt,Argurio:2008zt}).

\subsection{Describing $\mathrm K^*$ as a subgroup of $\mathrm{SO}(2,2)$}
\label{sec:SO22Embedding}

From group theoretic considerations one can derive the above conclusions using rather general arguments. Let $\Delta^2 :\mathbb{R}^4 \rightarrow \mathbb{R} $ be the homogeneous quadratic form  defined  by
\begin{equation}
\Delta^2(v) = m^2+n^2-q^2-h^2,
\end{equation}
for $v=(m,n,h,q)\in\mathbb{R}^4$, and let 
\begin{equation}
B = \{ v \in \mathbb{R}^4 \backslash \{0\}; \Delta^2(v) = 0\}
\end{equation}
be the set of zeros of $\Delta^2$. We know from Section \ref{sec:StationarySolutions} that $B$ is isomorphic to the set of single centered BPS-solutions via the maps (\ref{eqn:BPSscalars}). The set $B$ is by definition preserved by the group $\mathrm S = \mathrm{SO}(2,2)$. The question of how $\mathrm K^*$ acts on the set of charges can then be transformed into the question of how $\mathrm K^*$ embeds into $\mathrm S$ as we know from the previous section that $\mathrm K^*$ preserves the BPS-condition. The Lie algebra isomorphism $\mathfrak{so}(2,2) \cong \mathfrak{sl}(2,\mathbb{R}) \oplus \mathfrak{sl}(2,\mathbb{R})$ induces the Lie group isomorphism $ \mathrm{SO}(2,2)_0 \cong \mathrm{SL}(2,\mathbb{R}) \times \mathrm{SL}(2,\mathbb{R})$. Here  $\mathrm{SO}(2,2)_0$ indicates the component connected to the identity. We also know that $ \mathfrak{sl}(2,\mathbb{R})$ contains two non-compact generators and one compact. Comparing with $\mathrm K^*$ whose Lie algebra we know contains two compact generators and two non-compact generators, forming $\mathfrak{k}^* = \mathfrak{sl}(2,\mathbb{R}) \oplus \mathbb{R}u$, the embedding $I : \mathrm K^* \hookrightarrow \mathrm S$ is therefore given by lifting the natural (up to automorphisms) differential $\mathrm{d}I_e : \mathfrak{k}^* \rightarrow \mathfrak{so}(2,2)$ at the identity $e\in\mathrm{K}^*$, mapping compact generators to compact generators. More concretely, if $b_{i=1,...6}$ are the generators of $\mathfrak{so}(2,2)$ (for definition of the algebra $\mathfrak{so}(2,2)$ see Appendix \ref{app:so22}),
\be \begin{split}
\mathrm{d}I_e(t_i)& = b_i \, ,\\
\mathrm{d}I_e(u) &= b_4 \, ,
\end{split}\ee
where $b_1,b_2$ and $b_3$ generate one $\mathfrak{sl}(2)$ summand in $\mathfrak{so}(2,2)$, and $b_4$ is the compact generator in the other. The normalization is not fixed, but is up to redefinition of the generators of the two Lie algebras. By looking at the action of $I(\mathrm K^*)$, we now see that $b_1-b_4$ generates an Ehlers $\mathrm{U}(1)$-group rotating $m,n$ into each other, $b_1+b_4$ generates a $\mathrm{U}(1)$ rotating $q,h$ and the non-compact $b_2$ and $b_3$ act as boosts. This is in complete agreement with the analysis in Section \ref{sec:ActionOnCharges} above.

\subsection{The space of BPS solutions}

Now as we know how $\su$, or more precisely, how $\mathrm K^*$ acts on the set of  single centered BPS-solutions we can ask the question about duality orbits. Namely, starting with one BPS-solution, can we generate all the others by acting with $\mathrm K^*$? If $\su$ is to be considered as a spectrum generating group \cite{Cremmer:1997xj}, this must clearly be the case. Here the fact that $\mathrm K^*$ is non-compact will be of crucial importance. In fact, we have the following result.

\begin{theorem}
\label{thm:modulispace}
The group $\mathrm K^*$ acts transitively on the set of single centered BPS-solutions, so that $B \cong \mathrm K^*/\mathbb{R}$.
\end{theorem}
\begin{proof}
Consider the set $B$. As it is defined by the homogeneous form $\Delta^2$, we can consider the projective descendant of $B$, namely $\mathrm{P}B =  \{ v \in \mathrm{P}\mathbb{R}^3 ; \Delta^2(v) = 0\}$, where $\mathrm{P}\mathbb{R}^3$ is the three-dimensional projective space. In analogy with the isomorphism $\mathrm{SO}(2,2)_0 \cong \mathrm{SL}(2,\mathbb{R}) \times \mathrm{SL}(2,\mathbb{R})$, we get the isomorphism $\mathrm{P}B \cong \mathrm{P}\mathbb{R}^1\times  \mathrm{P}\mathbb{R}^1$ via a bijection $F :  \mathrm{P}\mathbb{R}^1\times  \mathrm{P}\mathbb{R}^1 \rightarrow \mathrm{P}B$ given by the expression
\begin{equation}
F([x_0,x_1], [y_0,y_1]) = [x_0 y_0+x_1 y_1, y_0x_1-y_1x_0,x_0y_0-x_1y_1, x_0y_1+x_1y_0] .
\end{equation}
The action of $I(\mathrm K^*)$ on $B$ descends to an action on $\mathrm{P}B$, and hence to an action on $\mathrm{P}\mathbb{R}^1\times  \mathrm{P}\mathbb{R}^1$ by $F$. Furthermore, $\mathrm{P}\mathbb{R}^1 \cong S^1$ via the map $f([x_0,x_1]) = \arctan(x_0/x_1)$ (schematically), so that we have a diffeomorphism $\mathrm{P}B \cong S^1 \times S^1$. In fact, the Lie subgroup $\mathrm{U}(1) \times \mathrm{U}(1) \subset \mathrm K^*$, generated by the subalgebra $\mathbb{R} t_1 \oplus \mathbb{R} u$, acts transitively on these two circles by complex multiplication. We conclude that $\mathrm K^*$ acts transitively on $\mathrm{P}B$. Furthermore $\mathrm{U}(1) \times \mathrm{U}(1)\subset \mathrm K^*$ contains an element acting as $v \rightarrow -v$ for $v \in B$. Let us now turn to the action of $\mathrm K^*$ on $B$. Due to the above analysis, it is sufficient to consider charge vectors with all charges positive and equal. As $\mathrm K^*$ contains non-compact generators it is now in fact possible to reach all these charge vectors, being given one. The explicit 1-parameter Lie subgroup is $S(\lambda) = \exp\big(-\lambda b_2\big)$, acting so that $(k,k,k,k) \mapsto  \exp{\lambda} (k,k,k,k)$. This proves our assertion, noting that the 1-parameter subgroup stabilizing a diagonal vector  $(k,k,k,k)$ is $\mathrm{Stab}(c) \equiv \exp{c (b_1+b_3)} \cong \mathbb{R}$.
\end{proof}

If we consider this proof from the physical point of view, it may seem surprising that $\mathrm K^*$ is transitive on $\mathrm{P}B$ by only using $b_1$ and $b_4$ as these do not mix gravitational and electromagnetic degrees of freedom. This is in fact true as for four charges $m,n,q,h$ to fulfill the BPS-condition we need both non-zero gravitational and non-zero electromagnetic charges and to generate new solutions we can treat these two sectors separately.  Furthermore, we can compare the result of Theorem \ref{thm:modulispace} with the expression for the $1/2$-BPS strata in \cite{Bossard:2009at}, (equation (5.5)) and see that the two results are in full agreement.

\subsection{The quantum moduli space and string theory}

Our analysis so far has been performed purely at the classical level. In the full quantum theory it is expected that the classical moduli space  is affected by quantum corrections. These can be both of perturbative and of non-perturbative origin and they are not understood generally. The only exceptions are cases where there are additional duality symmetries that constrain them.

In general, electric and magnetic charges are subject to quantization in the sense of Dirac. For example, the $28+28$ electric and magnetic charges in type II string theory on a six-torus break the  classical continuous $\mathrm E_{7(7)}$ symmetry group to the discrete subgroup $\mathrm E_{7(7)}(\mbb{Z})$ \cite{HullTownsend}
\be
\mathrm E_{7(7)}(\mbb{Z})=\mathrm E_{7(7)}(\mbb{R})\cap\mathrm{Sp}(56, \mbb{Z})\,,
\ee
where $\mathrm{Sp}(56, \mbb{Z})$ is the symmetry group of the 56-dimensional symplectic lattice of electric and magnetic charges, associated with the 28 abelian vector fields in $D=4$. 

It has furthermore been speculated that after further reduction of this maximal supergravity theory on a space-like circle $S^1$ to $D=3$, the duality group should be enhanced to some discrete subgroup $\mathrm G(\mbb{Z})$ of the classical hidden symmetry group $\mathrm E_{8(8)}(\mbb{R})$ \cite{HullTownsend}. However, in three dimensions it is by no means clear how to define the group $\mathrm G(\mbb{Z})$, since there are no vector fields whose associated charge lattice provides a natural integral structure. Moreover, in $D=3$ one is forced to take into account gravitational effects since the moduli space includes components of the four-dimensional metric.
It was recently argued that the three-dimensional duality groups that arise in this way do not act nicely on the gravitational part of the moduli space, and there is therefore no natural candidate for a discrete subgroup $\mathrm G(\mbb{Z})$ which should be preserved in the quantum theory in $D=3$~\cite{Bossard:2009at}.

Returning to the ${\cal N}=2$ theory discussed in this paper, the situation is not very different at face value. However, we propose that the $c$-map \cite{CecottiFerraraGirardello,Behrndt,VandorenVroome} improves the situation. The $c$-map can be thought of as a type of $T$-duality in $D=3$, where it exchanges the moduli space obtained from the reduction of the Einstein-Maxwell sector with that obtained by the reduction of a universal hypermultiplet sector that can be added to the ${\cal N}=2$ theory in $D=4$ and that is present in any Calabi-Yau reduction of type IIA superstring theory \cite{CecottiFerraraGirardello}.\footnote{We ignore, i.e. set to zero, the effects of the other hyper- and vectormultiplets that arise in the reduction.} The point here is that the universal hypermultiplet in $D=4$ is also described classically by a coset space $\mathrm{SU}(2,1)/\mathrm{SU}(2)\times\mathrm{U}(1)$. The quantum corrections to this universal hypermultiplet moduli space are not fully understood, but recently~\cite{UHPaper} it has been proposed that a promising candidate for the discrete group $\mathrm{G}(\mathbb{Z})$ in this case is the so-called the Picard modular group $\mathrm{SU}(2,1;\mathbb{Z}[i])$, whose generators can be given an intuitive physical interpretation in terms of Peccei-Quinn symmetries, electric-magnetic duality and S-duality. Assuming this to be the correct quantum duality group of the universal hypermultiplet and the validity of the $c$-map at the quantum level would imply that the correct moduli space and quantum symmetry group of Einstein-Maxwell theory with one Killing vector is also encoded in the Picard group. A further verification of these claims is outside the scope of this paper. See \cite{UHPaper} for more detailed discussions of these issues.

 \setcounter{equation}{0}
 
\section{On $\asuppp$} \label{sec:su21+++}

So far we have analyzed the role of the duality group SU(2,1) for understanding BPS solutions in $\mathcal{N}=2$ supergravity in four dimensions. This was done by performing a dimensional reduction to three dimensions, where the Lagrangian corresponds to Einstein gravity coupled to scalars parametrizing a Riemannian coset space $\mathcal{C}$ in the case of space-like reduction, and a pseudo-Riemannian coset space $\mathcal C^*$ in the case of time-like reduction.\\
 
Motivated by this, it is interesting to assume that the Einstein-Maxwell theory exhibits a hidden nonlinearly realized Kac-Moody symmetry group SU(2,1)$^{+++}$, formally arising in the reduction to zero dimensions~\cite{Julia:1982gx}, but as a conjectured symmetry of the full model~\cite{West:2001as}. The associated Kac-Moody algebra $\mathfrak{su}(2,1)^{+++}$ can be obtained by adding three nodes $\alpha_1, \alpha_2$ and $\alpha_3$ to the Tits-Satake diagram of $\asu$ displayed in Figure \ref{fig:su21}. The Tits-Satake diagram of  $\mathfrak{su}(2,1)^{+++}$ is given in Figure \ref{fig:su21+++}b.\\

In this section, we will give some basic properties of $\asuppp$, and explain the construction of a non-linear $\s$-model on the infinite-dimensional coset space $\suppp/\mathrm K^{*\,+++}$, generalizing the finite-dimensional $\sigma$-model on G/K$^{*}$ considered in Section \ref{SU(2,1)SigmaModel}. Here $\mathrm K^{*\,+++}$ is the subgroup of $\suppp$ consisting of those generators which are pointwise fixed by the temporal involution $\Omega_1$, defined such that we may identify the index $1$ by a time coordinate. To this end we shall slice up the adjoint representation of $\asuppp$ in a level decomposition, suitable to reveal the field content of the bosonic part of pure $\mathcal{N}=2$ supergravity. We will also define the action of  a general temporal involution $\Omega_i$ on the generators of $\asuppp$. 

In the same way as $\mathfrak{su}(2,1)$ is a real form of the complex Lie algebra $A_2=\mathfrak{sl}(3,\mathbb{C})$, the Kac-Moody algebra $\mathfrak{su}(2,1)^{+++}$ is a real form of the complex algebra $A_2^{+++}$. To construct $\mathfrak{su}(2,1)^{+++}$ it is therefore illuminating to first consider some relevant properties of $A_2^{+++}$. 
\subsection{Generalities on $A_2^{+++}$}
\label{A2+++}
The rank 5 Kac-Moody algebra $A_2^{+++}$ can be obtained by adjoining three extra nodes to the Dynkin diagram of the finite-dimensional Lie algebra $A_2$. The resulting Dynkin diagram is displayed in Figure \ref{fig:su21+++}a, where the nodes $\alpha_4$ and $\alpha_5$ correspond to the underlying $A_2$-algebra, while $\alpha_1, \alpha_2$ and $\alpha_3$ provide the extension to $A_2^{+++}$. From the Dynkin diagram we may construct the associated Kac-Moody algebra by introducing five triples of generators $\{E_i, F_i, H_i\}, \, i=1,\dots, 5,$ known as Chevalley generators, such that each triple generate an $\mathfrak{sl}_{(i)}(2,\mathbb{C})$-subalgebra corresponding to the nodes $1, \dots, 5$ in the Dynkin diagram. The Chevalley generators are subject to the commutation relations (no summation on repeated indices)
\begin{eqnarray}
{} [H_i, E_j]= A_{ij} E_j, &  & [H_i, F_j]=-A_{ij} F_j,
 \nn  \\
{} [E_i, F_j]=\delta_{ij} H_j, &  & [H_i, H_j]=0,
\end{eqnarray}
where $A_{ij}$ is the Cartan matrix encoding the structure of the Dynkin diagram in Figure \ref{fig:su21+++}a. The sets $\{E_i\}$  and $\{F_i\}$ correspond, respectively, to positive and negative step-operators which generate the nilpotent subspaces $\mathfrak{\tilde{n}}^{+++}_+$ and $\mathfrak{\tilde{n}}_-^{+++}$, modulo the so-called Serre relations (see \cite{Kac:book}). In addition, the set $\{H_i\}$ generates the Cartan subalgebra $\tilde{\mathfrak{h}}^{+++}$, providing the full Kac-Moody algebra with a triangular structure (direct sums of vector spaces) \footnote{ The different subspaces of  $A_2^{+++}$ are denoted with a $\ \tilde{}\ $  to distinguish them from the different subspace of $\mathfrak{su}(2,1)^{+++}$ to be introduced below.}
\begin{equation}
A_2^{+++}=   \mathfrak{\tilde{n}}_-^{+++}\oplus \tilde{\mathfrak{h}}^{+++}\oplus \mathfrak{\tilde{n}}^{+++}_+.
\end{equation}
In the following subsection, we shall use these properties of $A_2^{+++}$ to define the non-split real form $\mathfrak{su}(2,1)^{+++}$ from the Tits-Satake diagram in Figure \ref{fig:su21+++}b.

\subsection{Definition of $\mathfrak{su}(2,1)^{+++}$}
The Tits-Satake diagram in Figure \ref{fig:su21+++}b differs from the standard Dynkin diagram of $A_2^{+++}$ (see Figure \ref{fig:su21+++}a) by the extra decoration afforded by the double arrow connecting the nodes $\alpha_4$ and $\alpha_5$.\footnote{We note that the theory of real forms of Kac-Moody algebras has one important difference to the theory of real forms for finite-dimensional algebras. Since not any two Borel subalgebras are conjugate in the Kac-Moody case there are different classes of real forms. Indeed, the standard upper triangular and lower triangular Borel subalgebras, $\mathfrak{b}_+$ and $\mathfrak{b}_-$, cannot be conjugated into one another~\cite{PetersonKac} and depending on whether the involution fixing the real form maps $\mathfrak{b}_+\to\mathfrak{b}_+$ or $\mathfrak{b}_+\to\mathfrak{b}_-$ the real form is called almost split or almost compact~\cite{Rousseau1989,Rousseau1995,BenMessaoud}. Almost split algebras are under better control and the fact that here we have an involution, given by the arrow in the diagram in Figure~\ref{fig:su21+++}b, acting only on a finite-dimensional subalgebra ensures that we are constructing an almost split real form.}   This implies that the $A_2$-part of the diagram is transformed into the non-split real form $\asu$ of $A_2$ such that on the simple roots one has
\be \label{eqn:sigmaa4a5}
\sigma (\alpha_4)= \alpha_5, \qquad \sigma (\alpha_5)= \alpha_4\, ,
\ee
where $\sigma$ is the conjugation which fixes the real form. More details on $\asu$ can be found in Appendix \ref{app:su21}.
\begin{figure}[t]
\begin{center}
\begin{pgfpicture}{0cm}{-2cm}{1cm}{2.5cm}

\pgfnodecircle{Node12}[stroke]{\pgfxy(-6,0.5)}{0.25cm}
\pgfnodecircle{Node22}[stroke]
{\pgfrelative{\pgfxy(1.5,0)}{\pgfnodecenter{Node12}}}{0.25cm}
\pgfnodecircle{Node32}[stroke]
{\pgfrelative{\pgfxy(1.5,0)}{\pgfnodecenter{Node22}}}{0.25cm}
\pgfnodecircle{Node42}[stroke]
{\pgfrelative{\pgfxy(1.5,1)}{\pgfnodecenter{Node32}}}{0.25cm}
\pgfnodecircle{Node52}[stroke]
{\pgfrelative{\pgfxy(1.5,-1)}{\pgfnodecenter{Node32}}}{0.25cm}

\pgfnodebox{Node62}[virtual]{\pgfxy(-6,0)}{$\alpha_{1}$}{2pt}{2pt}
\pgfnodebox{Node72}[virtual]{\pgfxy(-4.5,0)}{$\alpha_{2}$}{2pt}{2pt}
\pgfnodebox{Node82}[virtual]{\pgfxy(-3,0)}{$\alpha_{3}$}{2pt}{2pt}
\pgfnodebox{Node92}[virtual]{\pgfxy(-1.5,-1)}{$\alpha_{4}$}{2pt}{2pt}
\pgfnodebox{Node102}[virtual]{\pgfxy(-1.5,2)}{$\alpha_{5}$}{2pt}{2pt}

\pgfnodeconnline{Node12}{Node22} \pgfnodeconnline{Node22}{Node32}
\pgfnodeconnline{Node32}{Node42} \pgfnodeconnline{Node32}{Node52} 
\pgfnodeconnline{Node42}{Node52}

\pgfnodebox{Node63}[virtual]{\pgfxy(-4,-1.5)}{\textbf{a}}{2pt}{2pt}

\pgfnodebox{Node64}[virtual]{\pgfxy(4,-1.5)}{\textbf{b}}{2pt}{2pt}

\pgfnodecircle{Node1}[stroke]{\pgfxy(2,0.5)}{0.25cm}
\pgfnodecircle{Node2}[stroke]
{\pgfrelative{\pgfxy(1.5,0)}{\pgfnodecenter{Node1}}}{0.25cm}
\pgfnodecircle{Node3}[stroke]
{\pgfrelative{\pgfxy(1.5,0)}{\pgfnodecenter{Node2}}}{0.25cm}
\pgfnodecircle{Node4}[stroke]
{\pgfrelative{\pgfxy(1.5,1)}{\pgfnodecenter{Node3}}}{0.25cm}
\pgfnodecircle{Node5}[stroke]
{\pgfrelative{\pgfxy(1.5,-1)}{\pgfnodecenter{Node3}}}{0.25cm}

\pgfnodebox{Node6}[virtual]{\pgfxy(2,0)}{$\alpha_{1}$}{2pt}{2pt}
\pgfnodebox{Node7}[virtual]{\pgfxy(3.5,0)}{$\alpha_{2}$}{2pt}{2pt}
\pgfnodebox{Node8}[virtual]{\pgfxy(5,0)}{$\alpha_{3}$}{2pt}{2pt}
\pgfnodebox{Node9}[virtual]{\pgfxy(6.5,-1)}{$\alpha_{4}$}{2pt}{2pt}
\pgfnodebox{Node10}[virtual]{\pgfxy(6.5,2)}{$\alpha_{5}$}{2pt}{2pt}

\pgfnodeconnline{Node1}{Node2} \pgfnodeconnline{Node2}{Node3}
\pgfnodeconnline{Node3}{Node4} \pgfnodeconnline{Node3}{Node5} 
\pgfnodeconnline{Node4}{Node5}
\pgfsetstartarrow{\pgfarrowtriangle{4pt}}
\pgfsetendarrow{\pgfarrowtriangle{4pt}}
\pgfnodesetsepend{5pt}
\pgfnodesetsepstart{5pt}
\pgfnodeconncurve{Node4}{Node5}{-10}{10}{1cm}{1cm}
\end{pgfpicture}
\caption {  \small \textbf{a}: {\sl Dynkin diagram of $A_2^{+++}$. }\textbf{b}: {\sl Tits-Satake diagram of $\asuppp$.} }
\label{fig:su21+++}
\end{center}
\end{figure}

Analogously to the construction of $A_2^{+++}$ in Section \ref{A2+++}, to construct $\asuppp$ we  extend the Tits-Satake diagram of $\asu$ (see Figure \ref{fig:su21}) with three non-compact simple roots  $\alpha_j$ $(j=1, \ldots, 3)$ such that 
\be \label{eqn:sigmalphanc}
\s (\alpha_j)= \alpha_j \quad(j=1,2,3)\, .
\ee

The action of $\s$
can be extended from the space of roots to the entire algebra. For $\sigma$ this yields
\be\label{eqn:sigmasu21+++} \begin{split}
\begin{aligned}
\s (H_j)&=H_j, \ &\s (E_j)&=E_j, \ &\s (F_j)&=F_j\,, \\
\s (H_4)&=H_5, \ &\s (E_4)&=E_5, \ &\s (F_4)&=F_5\,, \\
\s (H_5)&=H_4, \ &\s (E_5)&=E_4, \ &\s (F_5)&=F_4\,, 
\end{aligned} \end{split}
\ee
where $\{H_i, E_i, F_i\}$ are the Chevalley generators of $A_2^{+++}$ introduced in Section \ref{A2+++}.
The generators of $\asuppp$ then correspond to the subset of $A_2^{+++}$-generators left invariant under $\sigma$. 
They can be written
 in terms of the Chevalley generators of $A_2^{+++}$ as follows
\be \label{eqn:su21+++Chev}
\begin{split}
\begin{aligned}
e_1&= E_1, \qquad&&f_1=F_1, \qquad  &&h_1= H_1\, ,\\
e_2&= E_2, \qquad&&f_2=F_2,  \qquad&&h_2= H_2\, ,\\
e_3&= E_3, \qquad &&f_3=F_3, \qquad&&h_3= H_3\, ,\\
e_4&= E_4 + E_5, \qquad &&f_4=F_4 +F_5,\qquad &&h_4= H_4+ H_5\, ,\\
e_5&=i\,( E_4 - E_5), \qquad &&f_5=i\,(F_4 -F_5), \qquad &&h_5= i \,(H_4- H_5)\, .
\end{aligned}
\end{split}
\ee
We stress that for these generators there is no set of standard Chevalley--Serre relations defining the commutators between these elements.

The involution $\theta$ that fixes the maximal compact subalgebra $\mathfrak{k}^{+++}$ of $\asuppp$ is defined by

\be \begin{split}
\begin{aligned}
\label{eqn:thetasu21+++}
\theta (H_j)&=- H_j \ &\theta (E_j)&=- F_j, \ & \theta (F_j)&=- E_j\,,\\
\theta (H_4)&=- H_5 \ &\theta (E_4)&=- F_5, \ &\theta (F_4)&=- E_5\,,\\
\theta (H_5)&=- H_4 \ &\theta (E_5)&=- F_4, \ & \theta (F_5)&=- E_4\,.
\end{aligned} \end{split}
\ee

The Cartan subalgebra of $\asuppp$ is given by
\be
\mathfrak{h}^{+++}= \bigoplus _{i=1}^{5} \mathbb{R} h_i,
\ee
of which $h_1, \ldots, h_4$ are non-compact, while $h_5$ is compact. The generators $h_1, \ldots, h_4$ are diagonalizable over the real numbers, and generate the non-compact part $\mathfrak{a}^{+++}$ of the full Cartan subalgebra $\mathfrak{h}^{+++}$. 

Recall from Section \ref{SU(2,1)SigmaModel} that the construction of the $\s$-model on the coset space $\suppp/\mathrm K^{*\, +++} $ will only involve the non-compact Cartan generators by virtue of the algebraic Iwasawa decomposition 
\be \label{Iwasuppp}
\asuppp = \mathfrak{k}^{*\, +++} \oplus \mathfrak{a}^{+++} \oplus \mathfrak{n}^{+++},
\ee
where $\mathfrak{k}^{* \, +++}$ is the non-compact subalgebra of $\asuppp$ corresponding to the group $\mathrm K^{*\, +++}$, and  $\mathfrak{n}^{+++}$ is the nilpotent subalgebra generated by the set $\{ e_i \}$.

\subsection{Level decomposition} \label{sec:levdecompo}
In Section \ref{sec:CosmologicalModel}, we will give the correspondence between the field content of the bosonic part of pure $\mathcal{N}=2$ supergravity and the infinite-dimensional algebra $\asupp$. To this end, we will perform a decomposition of the adjoint representation of $\asuppp$ into representations of an $\mathfrak{sl}(4, \mathbb{R})$ subalgebra defined by the nodes $\alpha_1, \alpha_2$ and $\alpha_3$ in Figure \ref{fig:su21+++}b.  All step operators may then be written as irreducible tensors of the $\mathfrak{sl}(4, \mathbb{R})$ subalgebra. Their symmetry properties are fixed by the Young tableaux describing the irreducible representations appearing at a given level.
\subsubsection{Level decomposition of $A_2^{+++}$}

In order to understand the level decomposition of $\asuppp$, we must first consider the level decomposition of the complex algebra $A_2^{+++}$ under a $A_3\cong \mathfrak{sl}(4,\mathbb{R})$ subalgebra. This level decomposition of $A_2^{+++}$  up to level $\ell= (\ell_1,\ell_2)= (2,2)$ can be obtained for example from the SimpLie program \cite{SimpLie} and  it is shown in Table \ref{tab:levdeca2} . The levels $\ell_1$ and $\ell_2$ are respectively associated to the roots $\alpha_4$ and $\alpha_5$ in Figure \ref{fig:su21+++}a. This level decomposition will induce a grading of $A_2^{+++}$ into an infinite set of finite-dimensional subspace $\mathfrak{g}^{+++}_{(\ell_1, \ell_2)}$ such that
\be
A_2^{+++}= \bigoplus_{(\ell_1,\ell_2)} \mathfrak{g}^{+++}_{(\ell_1, \ell_2)}\, ,
\ee 
where the levels $\ell_1$ and $\ell_2$ are either both non-positive or both non-negative.\\

 At level $\ell=(0,0)$, we have a $\mathfrak{gl}(4, \mathbb{R})= \mathfrak{sl}(4,\mathbb{R}) \oplus\, \mathbb{R}$ algebra generated by $K^{a}_{\ b}$ $(a, b = 1, \ldots, 4)$, as well as an extra scalar generator $T$ which enlarges the  $\mathfrak{gl}(4, \mathbb{R})$ algebra by the addition of a $\mbb R$-factor. The commutation relations at this level are
\be \begin{split}\label{eqn:kab}
\left[K^{a}_{\ b}, K^{c}_{\ d}\right] &= \delta^c_b\,  K^{a}_{\ d} - \delta^a_d\,  K^{c}_{\ b}\, ,\\
\left[T, K^{a}_{\ b}\right]&=0 \, ,
\end{split}
\ee
and the bilinear forms  reads
\be \label{eqn:bilinearzero}
(K^{a}_{\ b} |K^{c}_{\ d}) = \delta^a_d \delta^c_b- \delta^a_b\delta^c_d, \ 
(T|T)= \frac{2}{9}  , \ (T | K^{a}_{\ b})=0\, .
\ee
The positive level generators are obtained through multiple commutators between the generators $R^a$ and $\tilde{R}^a$ on levels $(1,0)$ and $(0,1)$ respectively. They transforms as $\mathfrak{gl}(4, \mathbb{R})$ tensors in the obvious way. The level $(1,1)$ generators $R^{ab}$ and $S^{ab}$ are obtained through the commutator
\be
\left[ R^a, \tilde{R}^b\right]=  R^{ab} + S^{ab} \,  ,
\ee
where the individual projections are:
\be\label{eqn:rabsab}
S^{ab}= \left[ R^{ ( a}, \tilde{R}^{b ) }\right] , \  R^{ab}= \left[ R^{ [ a}, \tilde{R}^{b ] }\right] \  .
\ee
The Chevalley generators of $A_2^{+++}$ and its relevant commutators and bilinear forms up to level $(1,1)$ are given in Appendix \ref{a2+++}. Negative step operators are defined with lower indices such that the bilinear form evaluated on a positive step operator and its corresponding negative step operator is positive, e.g. $(R^a|R_b)=\delta^a_b$.

\begin{table}[h]
\begin{center}
\begin{tabular}{|c|c|c|}
\hline
$(\ell_1,\ell_2)$ &$ \mathfrak{sl}(4,\mathbb{R})$ Dynkin labels& Generator of $A_2^{+++}$\\
\hline \hline
$(0,0)$ &$[ 1,0,1] \oplus [ 0,0,0] $ & $K^a_{\ b}$\\
$(0,0)$ &$[ 0,0,0]  $ & $T$\\
$(1,0)$ &$[ 0,0,1]  $ & $R^{\, a}$\\
$(0,1)$ &$[ 0,0,1]  $ & $\tilde{R}^{\, a}$\\
$(1,1)$ &$[ 0,0,2]  $ & $S^{\, s_{1} s_{2}}$\\
$(1,1)$ &$[ 0,1,0]  $ & $R^{\,a_{1} a_{2}}$\\
$(2,1)$ &$[ 0,1,1]  $ & $R^{\, a_{0}|a_{1} a_{2}}$\\
$(1,2)$ &$[ 0,1,1]  $ & $\tilde{R}^{\, a_{0}|a_{1} a_{2}}$\\
$(2,2)$ &$[ 1,0,1]  $ & $R^{\, a_{0}|a_{1} a_{2} a_{3}}$\\
$(2,2)$ &$[ 0,2,0]  $ & $R^{\, a_{1} a_{2} | a_{3} a_{4}}$\\
$(2,2)$ &$[ 0,1,2]  $ & $R^{\,s_{1} s_{2} |a_{3} a_{4} }$\\
$\vdots$& $\vdots$ &$\vdots$\\
\hline
\end{tabular}
\caption{\sl \small Level decomposition of $A_2^{+++}$ under $\mathfrak{sl}(4,\mathbb{R})$ up to level $(2,2)$. The  indices $a_{i}$ are antisymmetric  while the indices $s_{i}$ are symmetric. Note that the generators from the level $(2,1)$ with mixed Young symmetries are subject to constraints.  }
\label{tab:levdeca2}
\end{center}
\end{table}

\subsubsection{Level decomposition of $\asuppp$}
We shall now apply the construction of $\asuppp$ to the level decomposition of $A_2^{+++}$. In this context, we define the level $L$ such that $L=\ell_1 +\ell_2$ and such that the grading of the $\asuppp$ algebra is written as
\be
\asuppp= \bigoplus_L \mathfrak{g}^{+++}_{L}.
\ee

 At level zero, we have the $\mathfrak{gl}(4, \mathbb{R})$-subalgebra associated to the nodes $\alpha_1, \alpha_2$ and $\alpha_3$. These nodes are non-compact and hence are, as we have seen in (\ref{eqn:sigmalphanc}), invariant under $\sigma$. Thus, the $\mathfrak{gl}(4,\mathbb{R})$ part at $L=0$ is the same as for $A_2^{+++}$. The extra Cartan generators associated to $\alpha_4$ and $\alpha_5$ are however affected by the conjugation $\sigma$. Using (\ref{eqn:su21+++Chev}) and (\ref{eqn:chevbis}), the invariant combinations are 
\begin{eqnarray}\label{eqn:h4}
h_4&=& H_4+ H_5 = - K + 2 K^4_{\ 4}  ,\\
h_5&=& i (H_4-H_5) = i 6\, T \label{eqn:h5},
\end{eqnarray}
where $K= \sum_{a=1}^{4} K^a_{\ a}$.
We have already seen that the first one is non-compact, while the second one is compact, i.e.
\be
\theta(h_4)= -h_4\,, \quad \theta(h_5)= h_5.
\ee
The effect of the algebraic Iwasawa decomposition (\ref{Iwasuppp}) will therefore be to project out the compact Cartan $h_5$. This was anticipated since the generator $T$ is associated with a dilaton which does not exist in four-dimensional Maxwell-Einstein gravity. We further define the invariant generators at level $L=1$
\be \label{eqn:rarta} \begin{split}
r^a &:= R^a+ \tilde{R}^a,\\
\tilde{r}^a &:= i (R^a - \tilde{R}^a).
\end{split} \ee
The corresponding negative step operators at level $L=-1$ are defined by
\be \label{eqn:rartaneg} \begin{split}
r_a &:= R_a+ \tilde{R}_a,\\
\tilde{r}_a &:= i (R_a - \tilde{R}_a).
\end{split} \ee
More generally, the negative step operators are obtained from the positive ones by lowering the indices as in \eqref{eqn:rartaneg}. The bilinear forms at this level reads
\be \label{eqn:bililevel1}
( r^a | r_b)= 2\, \delta^a_b, \quad ( \tilde{r}^a | \tilde{r}_b) = - 2\, \delta^a_b\, .
\ee
Using (\ref{eqn:su21+++Chev}) and (\ref{eqn:chevbis}), we get that the invariant combinations of the Chevalley generators at level $L=1$ are
\be \label{eqn:e4ande5}
e_4= r^4, \quad e_5= \tilde{r}^4.
\ee
That all other components of $r^a$ and $\tilde{r}^a$ are also invariant follows from the fact that they may be written as commutators between $\mathfrak{gl}(4, \mathbb{R})$ and $r^4$ and $\tilde{r}^4$ which are all invariant. The two Chevalley generators $e_4$ and $e_5$ have identical eigenvalues with respect to the four noncompact Cartan
 \be \label{comute4} \begin{split} \begin{aligned}
 \left[h_1, e_4\right]&=0,  & [h_2, e_4]&=0,& [h_3, e_4]&=- e_4, &[h_4, e_4]&=e_4, \\
 [h_1, e_5]&=0, & [h_2, e_5]&=0,& [h_3, e_5]&=- e_5, &[h_4, e_5]&=e_5, 
\end{aligned} \end{split} \ee
implying that these generators project into the same root $\vec{\lambda}$ in the restricted root system (see Appendix \ref{app:restricted} for more details),
\be
\vec{\lambda}= \vec{\alpha}_{e_4}= \vec{\alpha}_{e_5}= (0,0,-1,1).
\ee
The generator $h_5$, being compact, is not diagonalizable over $\mathbb{R}$. Indeed, we have the following commutation relations with $e_4$ and $e_5$
\be \label{comutenc45}
[h_5,e_4] =3\, e_5 \quad, \quad [h_5,e_5] =- 3\, e_4.
\ee
The generators at level $L=2$ are obtained as
\be \label{eqn:sabrab} \begin{split}
s^{ab}&:= [r^a, \tilde{r}^b]\, ,\\
r^{ab}&:=  [r^a,r^b]=  [\tilde{r}^a, \tilde{r}^b].
\end{split} \ee
These generators are separately invariant under $\sigma$. In terms of $A_2^{+++}$ generators, using (\ref{eqn:sabrab}), \eqref{eqn:rarta}, and \eqref{eqn:rabsab} we get
\begin{eqnarray}
s^{ab}&=& - 2i\,  S^{ab},\\
r^{ab}&=& 2\, R^{ab}\, .
\end{eqnarray}
These generators are normalized as
\be
(s^{ab} | s_{cd})= -4\, \bar{\delta}^{ab}_{cd}\,, \quad (r^{ab} | r_{cd})= 12 \, \delta^{ab}_{cd}\, ,
\ee
where  $\delta^{ab}_{cd} :=  \tfrac{1}{2} (\delta^a_c\, \delta^b_d- \delta^b_c\, \delta^a_d)$ and
$\bar{  \delta}^{ab}_{cd} :=  \tfrac{1}{2} (\delta^a_c\, \delta^b_d+ \delta^b_c\, \delta^a_d)$.\\

The level decomposition of $\asuppp$ under the $A_3\cong \mathfrak{sl}(4,\mathbb{R})$ subalgebra up to level $L=4$ is shown in Table \ref{tab:levdecsu}. Note that this level decomposition presents the same Young tableaux as in the $A_2^{+++}$ case. We will see in Section \ref{sec:CosmologicalModel} that this representation content up to level $L=2$ where the generator $r^{a_1a_2} $ is projected out,  can be associated with the bosonic field content of pure $\mathcal N=2$ supergravity in $D=4$ . The relevant commutators and bilinear forms of $\asuppp$ up to level $L=2$ are given in Appendix \ref{app:comsu21+++}.

\begin{table}[h]
\begin{center}
\begin{tabular}{|c|c|}
\hline
$L= \ell_1 +\ell_2$ &  Generator of $\asuppp$ \\
\hline \hline
$0$  & $K^a_{\ b}$\\
$0$  & $i \, T$\\
$1$  & $r^{\, a} = R^a + \tilde{R}^a$\\
$1$  & $\tilde{r}^{\, a} = i (R^a - \tilde{R}^a) $\\
$2$  & $s^{\,s_{1} s_{2}} = - 2\, i\, S^{\, s_{1} s_{2}}$\\
$2$  & $r^{\,a_{1} a_{2}}=  2\, R^{\, a_{1} a_{2}}$\\
$3$  & $r^{\, a_{0}|a_{1} a_{2}}$\\
$3$   & $\tilde{r}^{\, a_{0}|a_{1} a_{2}}$\\
$4$  & $r^{\, a_{0}|a_{1} a_{2} a_{3}}$\\
$4$  & $r^{\, a_{1} a_{2} | a_{3} a_{4}}$\\
$4$  & $r^{\,s_{1} s_{2} |a_{3} a_{4} }$\\
$\vdots$ & $\vdots$\\
\hline
\end{tabular}
\caption{\sl \small Level decomposition of $\asuppp$ under $\mathfrak{sl}(4,\mathbb{R})$ up to level $4$. The indices  $a_{i}$ are antisymmetric  while the indices  $s_{i}$ are symmetric. Note that the generators from the level $L=3$ with mixed Young symmetries are subject to constraints.  }
\label{tab:levdecsu}
\end{center}
\end{table}

\subsection{ Cartan and  temporal involutions} \label{sec:carttemporalin}

For the $\sigma$-models to be constructed in the next section we also need to fix a (local) subgroup of $\mathrm{SU}(2,1)^{+++}$. We require two different choices, denoted $\mathrm{K}^{+++}$ and $\mathrm{K}^{*+++}$, leading to different coset spaces and that are defined by appropriate involutions at the level of the $\mathfrak{su}(2,1)^{+++}$ Lie algebra. The level decomposition discussed above does not depend on the choice of this subalgebra but the $\sigma$-model to be studied below does.

The first choice of subalgebra, $\mathfrak{k}^{+++}$, is defined by the Cartan involution $\theta$. Its action on $\asuppp$ may be read off from the Tits-Satake diagram of $\asuppp$ (see (\ref{eqn:thetasu21+++})). It has the following action on the generators of $\asuppp$, 
\be \label{eqn:thetainv} \begin{split} \begin{aligned}
\theta (r^a)&= - r_a, &\quad \theta(r_a)&= -r^a,\\
\theta (\tilde{r}^a)&= \tilde{r}_a, &\quad \theta(\tilde{r}_a)&= \tilde{r}^a,\\
\theta (s^{ab})&= s_{ab}, &\quad \theta(s_{ab})&= s^{ab},\\
\theta (r^{ab})&=- r_{ab}, &\quad \theta(r_{ab})&=- r^{ab},
\end{aligned}   \end{split} \ee
while on level $L=0$ it has the familiar action
\be \label{eqn:thinkab}
\theta (K^a_{\  b}) = -\, K^b_{\  a}\quad \theta(iT) = iT\,.
\ee
The Cartan decomposition therefore reads
\be
\asuppp = \mathfrak{k}^{+++} \oplus \mathfrak{p}^{+++},
\ee
where the subalgebra $\mathfrak{k}^{+++}$ is defined as the fixed point set under the Cartan involution, while $\mathfrak{p}^{+++}$ contains the generators which anti-invariant under $\theta$.
The generators of $\mathfrak{k}^{+++}$ reads
\be \label{eqn:k3p} \begin{split} 
\mathfrak{k}^{+++}&= \{ x \in \asuppp\,  :\,  \theta(x)= x \}
\\
&= \{ iT,\,j^{ab},\, (r^a- r_a),\,  (\tilde{r}^a+ \tilde{r}_a),\, (s^{ab}+ s_{ab}),\, (r^{ab}- r_{ab}), \ldots  \}\, ,
\end{split}
\ee
where $j^{ab}= K^a_{\ b }- K^{b}_{\ a }$, and those of $\mathfrak{p}^{+++}$ are
\be \label{eqn:p3p} \begin{split} 
\mathfrak{p}^{+++}&= \{ x \in \asuppp \, :\,  \theta(x)= - x \}
\\
&= \{ k^{ab},\,  (r^a+ r_a),\,  (\tilde{r}^a-\tilde{r}_a),\,  (s^{ab}- s_{ab}),\, (r^{ab}+ r_{ab}), \ldots  \}\, ,
\end{split}
\ee
where $k^{ab}= K^a_{\ b }+ K^{b}_{\ a }$.\\

The second choice of subalgebra, $\mathfrak{k}^{*+++}$, is introduced via the so-called temporal involution \cite{Englert:2003py}.
The possible existence of a Kac-Moody symmetry $\mathrm G^{+++}$ motivated the construction of a Lagrangian formulation explicitly invariant under $\mathrm G^{+++}$ . This Lagrangian $\mathcal{S}_{\suppp}$ is defined in a reparametrisation invariant way on a world-line parameter $\xi$, apriori unrelated to space-time, in terms of fields living in a coset $\suppp/ \mathrm K^{*\,+++}$. As the metric $g_{\mu \nu}$ at a fixed space-time parametrises the coset $\mathrm{GL}(D)/ \mathrm{SO}(1,D-1)$, the subgroup $ \mathrm K^{*\,+++}$ must contain the Lorentz group. As $\mathrm{SO}(1, D-1)$ is non-compact, we cannot use the Cartan involution $\theta$ to construct $\mathrm K^{*\,+++}$ that is now non-compact. Rather we will use the temporal involution $\Omega_i$ from which the required non-compact generators of $\mathrm K^{*\,+++}$ can be selected. The temporal involution $\Omega_i$ generalises the Cartan involution $\theta$ described in (\ref{eqn:thetainv}) and \eqref{eqn:thinkab}  to allow the identification of the index $i$ as a time coordinate. It is defined by
\be \label{eqn:temporal} \begin{split}
\Omega_i (iT) &= iT,\\
\Omega_i(K^a_{\ b })&= - \epsilon_a  \epsilon_b\,  K^b_{\ a}, \\
\Omega_i (r^a)&= - \epsilon_a\,  r_a, \\
\Omega_i (\tilde{r}^a)&= \epsilon_a\,  \tilde{r}_a,\\
\Omega_i  (s^{ab})&= \epsilon_a \epsilon_b\,  s_{ab}, \\
\ \Omega_i  (r^{ab})&=-\epsilon_a \epsilon_b \,  r_{ab}, 
\end{split}
\ee
with $\epsilon_a= -1$ if $a=i$ and $\epsilon_a= 1$ otherwise.

 \setcounter{equation}{0}
\section{On $\asupp \subset \asuppp$ and $\sigma$-models } \label{sec:su21++}
We now turn our attention to one-dimensional $\sigma$-models based on the group $\supp$. The $\mathfrak{g}^{++}$ content of the $\mathfrak{g}^{+++}$-invariant actions $\mathcal{S}_{\mathrm G_{+++}}$ has been analysed in reference \cite{Englert:2004ph} where it was shown that two distinct actions invariant under the overextended Kac-Moody algebra $\mathfrak{g}^{++}$ exist.  We will apply this analysis to $\mathfrak g=\asu$ and study the two actions invariant under $\supp$.

The first one $\mathcal{S}_{\supp_C}$ is called the cosmological $\sigma$-model and constructed from $\mathcal{S}_{\suppp}$ by performing a truncation putting consistently to zero some fields.  The corresponding $\asupp$  algebra is obtained from $\asuppp$ by deleting the node $\alpha_1$ from the Tits-Satake diagram of $\asuppp$ depicted in Figure \ref{fig:su21+++}b. The involution used to construct the action $\mathcal{S}_{\suppp}$  is the temporal involution $\Omega_1$ (defined in (\ref{eqn:temporal})) such that  coordinate $1$ is time-like. This implies that the truncated theory $\mathcal{S}_{\supp_C}$ carries a Euclidean signature in space-time. The $\mathcal{S}_{\supp_C}$ is the generalisation of the $\mathrm E_8^{++}=\mathrm E_{10}$ invariant action of reference \cite{Damour:2002cu} proposed in the context of M-theory and cosmological billiards. The parameter $\xi$ along the world-line will then be identified with the time coordinate and we will see in Section \ref{sec:CosmologicalModel} that this action restricted to a defined number of  levels is equal to the corresponding $\mathcal{N}=2$ supergravity in $D=4$ in which the fields depend only  on this time coordinate.

A second $\supp$-invariant action $\mathcal{S}_{\supp_B}$, called the brane model, is obtained from $\mathcal{S}_{\suppp}$ by performing the same consistent truncation {\it after} conjugation by the Weyl reflection $s_{\alpha_1}$ in $\asuppp$. Here, $s_{\alpha_1}$ is the Weyl reflection in the hyperplane perpendicular to  the simple root $\alpha_1$ corresponding to the node 1 of Figure \ref{fig:su21+++}. The non-commutativity of the temporal involution $\Omega_1$ with the Weyl reflection \cite{Keurentjes:2004bv, deBuyl:2005it, Keurentjes:2004xx} implies that this second action is inequivalent to the first one (see Section \ref{sec:signa} where it is recalled the consequence of  $s_{\alpha_1}$ on the time identification).  In $\mathcal{S}_{\supp_B}$, $\xi$ is identified with a space-like direction.  For a generic $\mathrm G$, the $\mathrm G^{++}$-brane model describes intersecting extremal brane configurations smeared in all directions but one \cite{Englert:2003py, Englert:2004it}.

\subsection{Infinite-dimensional cosmological $\sigma$-model }
\label{sec:CosmologicalModel}

In this section we will analyze how well the suggestions in \cite{Damour:2002cu} apply to the pure $\mathcal{N}=2$ theory. More concretely, we will investigate what features of this theory can be described using a non-linear $\sigma$-models over an infinite-dimensional coset space, as a generalization of what we have seen in the case of the scalar Lagrangian (\ref{eqn:3dStationaryLagrangian}). In fact, we will find a correspondence between the supergravity fields and the parameters in a one-dimensional $\sigma$-model. For example, as we will see, the dynamics of some solutions to the supergravity equations of motion can be modelled by motion on a coset space $\supp/\mathrm K^{++}$, where $\mathrm K^{++}$ is the compact subgroup with Lie algebra $\mathfrak{k}^{++} \subset \mathfrak{su}(2,1)^{++}$. The results of this section will hence be a map between parts of the cosmological $\sigma$-model and parts of the supergravity. This confirms that the general conjecture describing supergravities with overextended Kac-Moody groups holds, to the same extent, also in the present case of pure $\mathcal{N}=2$ supergravity where the symmetry group is in a non-split form. In analogy with the discussion in Section \ref{sec:SpacelikeKillingVector} and \ref{sec:StationarySolutions}, the dynamics will be modelled by a non-linear $\sigma$-model of maps from $\mathcal{M}_1 \cong \mathbb{R}$ to $\supp/\mathrm K^{++}$. We will now formally construct this $\sigma$-model, and perform a check (as for example done in \cite{Damour:2004zy} in the case eleven-dimensional supergravity), to see if the corresponding equations of motion match with the dynamics given by the supergravity equations of motion (\ref{eqn:Einstein}) and (\ref{eqn:Maxwell}), when restricting to spatially constant solutions (in a sense to made more clear below). The action for the $\sigma$-model (given generally by (\ref{eqn:SigmaModelAction})) is
\begin{equation}
\mathcal{S}_{\supp_C} = \int_{\mathcal{M}_1} \frac{1}{n(t)} (\mathcal{P}(t) | \mathcal{P}(t)) \, \dd t\, ,
\end{equation}
where $n(t) = \sqrt{h}$ is the lapse function and necessary for reparametrization invariance on the world-line. The function $h(t)$ is the metric on the one-dimensional manifold $\mathcal{M}_1$ and $( \cdot | \cdot )$ is an invariant bilinear form of $\mathfrak{su}(2,1)^{++}$, formed for example by restriction from $\mathfrak{su}(2,1)^{+++}$. As described previously, $\mathcal{P}(t)$ is the component along the coset of the Maurer-Cartan form defined by maps into the coset. In the case of a one-dimensional base-manifold the $\sigma$-model equations of motion are (see e.g. (\ref{eqn:SigmaMotion}))
\begin{equation}
\label{eqn:cosmsigmamotion}
n\,  \partial_t (n^{-1} \, \mathcal{P}) + [\mathcal{Q}, \mathcal{P}] = 0 ,
\end{equation}
where $\mathcal P$ and $\mathcal Q$ are defined in \eqref{eqn:pandq}.
Now, as we are dealing with an infinite-dimensional coset space we cannot directly realize this model. What we will do is to use the level decomposition as described previously in Section \ref{sec:su21+++}, and perform a truncation of the Kac-Moody algebra by throwing away all the generators above a certain level. This truncation can be shown to be a consistent truncation of the $\sigma$-model equations of motion \cite{Damour:2004zy}. Before performing this truncation however, we have to describe the level decomposition of $\asupp$ in terms of the decomposition of $\asuppp$, given in Section \ref{sec:su21+++}. By defining $\mathfrak{su}(2,1)^{++}$ as a regular subalgebra of $\asuppp$, the level decomposition given in Table \ref{tab:levdecsu} descends to $\asupp$ by restricting the indices to not run over $1$, or equivalently by generating the representations at every level by acting on the highest weight with the regular $\mathfrak{sl}(3, \mathbb{R})$ subalgebra of $\mathfrak{sl}(4, \mathbb{R})$. In this section, the $\mathfrak{sl}(4, \mathbb{R})$ indices $a,b...$ will therefore only take the values $2,\,3$ and $4$. 

We can hence realize a suitable truncation of $\asupp$ by setting all $\mathfrak{g}_L = 0$ for $|L| > 2$, and furthermore set the antisymmetric generator $r^{ab}$ at level $L=2$ to zero, as this generator has no clear interpretation in terms of supergravity quantities. We can therefore write general $\mathcal{P}$ and $\mathcal{Q}$ as
\begin{equation}
\mathcal{P} = \frac{1}{2} p_{ab} k^{ab} + \frac{1}{2} P_a (r^a + r_a) + \frac{1}{2} \tilde{P}_a (\tilde{r}^a - \tilde{r}_a) + \frac{1}{2} P_{ab}(s^{ab} - s_{ab})\,,
\end{equation}
and
\begin{equation}
\mathcal{Q} = \frac{1}{2} q_{ab} j^{ab} + \frac{1}{2} P_a (r^a - r_a) + \frac{1}{2} \tilde{P}_a (\tilde{r}^a + \tilde{r}_a) + \frac{1}{2} P_{ab}(s^{ab} + s_{ab}) ,
\end{equation}
where we have expanded in the basis given in (\ref{eqn:k3p}) and (\ref{eqn:p3p}), in parameters $p_{ab}, P_a$ and so on, depending only on the time coordinate $t$. We have chosen to put different parameters in front of the generators at level zero in the expressions for $\mathcal{P}$ and $\mathcal{Q}$ , considering that $k^{ab}$ and $j^{ab}$ are symmetric and anti-symmetric respectively. Using the commutation relations in Appendix \ref{app:su21+++}, the equations of motion are now (given by inserting the expressions for $\mathcal{P}$ and $\mathcal{Q}$ in (\ref{eqn:levelzeromotion}) and (\ref{eqn:higherlevelmotion}))
\begin{eqnarray}
\label{eqn:pab}
n\, \partial_t (n^{-1} p_{ab} ) - q_{ca} {p^c}_b -  q_{cb}{p^c}_a + P_a P_b -\frac{1}{2} \delta_{ab} P_c P^c + \tilde{P}_a \tilde{P}_b \nonumber \\ 
- \frac{1}{2} \delta_{ab} \tilde{P}_c \tilde{P}^c - 2 \delta_{ab} P_{cd}P^{cd} + 4 P_{ac} {P_b}^c = 0\, ,
\end{eqnarray}
for the field $p_{ab}$, 
\begin{equation}
\label{eqn:Pa}
n\,  \partial_t (n^{-1} P_a) - p_{ac} P^c + q_{ac} P^c + 2 P_{ac} \tilde{P}^c = 0\, ,
\end{equation}
for the field $P_a$, 
\begin{equation}
\label{eqn:tildePa}
n \,\partial_t (n^{-1} \tilde{P}_a) - p_{ac} \tilde{P}^c + q_{ac} \tilde{P}^c - 2 P_{ac} P^c = 0\, ,
\end{equation}
for $\tilde{P}_a$ and finally
\begin{equation}
\label{eqn:Pab}
n\, \partial_t (n^{-1} P_{ab}) - 2 p_{ac} {P_b}^c +2 q_{ac} {P_b}^c = 0\, ,
\end{equation}
for $P_{ab}$. Regarding notation, we will in the following assume that indices are symmetrized or anti-symmetrized according to the tensor appearing linearly in expressions as these ones. For example, in (\ref{eqn:Pab}) the term $2 p_{ac} {P_b}^c$ is then short for $\frac{1}{2}(2 p_{ac} {P_b}^c + 2 p_{bc} {P_a}^c)$, the parameter $P_{ab}$ being a symmetric $\mathfrak{sl}(3, \mathbb{R})$ tensor.

\subsubsection{Dictionary} 

Let us now begin to compare the above $\sigma$-model dynamics with the dynamics given by our supergravity theory (\ref{eqn:SugraAction4d}). We will do this by choosing the parameters in $\mathcal{P}$ such that the $\sigma$-model equations of motion (\ref{eqn:pab})-(\ref{eqn:Pab}) match with the equations of motion on the supergravity side. Due to the construction of the $\sigma$-model, the natural framework for doing this is in the ADM-formalism where we will split the Einstein-Maxwell equations into dynamical equations and constraints/initial conditions. Concretely, we will only consider the dynamical supergravity equations. For the comparison it will be convenient to treat the spin connection $\omega_{ABC}$ and the field strength $F_{AB}$ as the fundamental fields of the Einstein-Maxwell theory and also considered as constant by letting them be space-independent. This is suitable because no spatial derivatives exist on the $\sigma$-model side. Locally we will also split the flat space coordinates $x^{a = 2,3,4}$ from the flat time coordinate $x^1$. The spin connection and the field strength transform under the Lorentz group $\mathrm{SO}(3,1)$ and we can use use this freedom, and a coordinate transformation, to put $\omega_{ABC}$ in a pseudo-Gaussian gauge by setting the metric shift functions to zero. This leads to $\omega_{ab1}$ being symmetric, and we also assume $\omega_{11c} = 0$, corresponding to space-independent gravity lapse $N$ . This gauge corresponds to a vielbein of the form
\begin{equation}
\label{eqn:vielbein}
{e_{\alpha}}^A = \left(\begin{array}{cc}N & 0 \\0 & {e_{\mu}}^a\end{array}\right)\,.
\end{equation}
In fact, the spin connection can be defined in terms of a tensor $\Omega_{ABC}$, called the anholonomy, by the relation
\begin{equation}
\label{eqn:anholonomy}
\omega_{ABC} = \frac{1}{2}(\Omega_{ABC} - \Omega_{BCA} +\Omega_{CAB}),
\end{equation}
and such that the anholonomy is given in terms of the vielbein by
\begin{equation}
{\Omega_{AB}}^C = 2\, {e_A}^{\alpha}{e_{B}}^{\beta}\partial_{[\alpha}{e_{\beta]}}^C.
\end{equation}
This pseudo-Gaussian gauge breaks the Lorentz group down to $\mathrm{SO}(3)$, acting on the spatial vielbein $ {e_{\mu}}^a$. Note also that we can rewrite the covariant derivative with respect to the spin connection using the vielbein, i.e.
\begin{eqnarray}
\label{eqn:SpinCovDerivative}
e^{-1} \partial_t (e \omega_{ab1})&  = & e^{-1}\,   \partial_t e\,  \omega_{ab1} +  \partial_t  \omega_{ab1} \nonumber \\
& = &  {e_{m}}^c \, \partial_t {e^{m}}_c\,  \omega_{ab1} +   \partial_t  \omega_{ab1}  \\
& = &  N {\omega^c}_{c1} \omega_{ab1} +  N \partial_1  \omega_{ab1} \nonumber .
\end{eqnarray}
Here we have used that $\partial_1 = N^{-1} \partial_t$. We will now consider the different parts of the supergravity equations of motion (\ref{eqn:Einstein}), (\ref{eqn:Maxwell}) and the Bianchi identity (\ref{eqn:MaxwellBianchi}), one at the time.

In addition to the gauges in the gravity sector, we also have to adopt a temporal gauge for the Maxwell field, corresponding with our choice of time coordinate to 
\be
A_1 = 0\,.
\ee

\begin{itemize}
\item{{\bf Ricci-tensor}}

First, let us consider the Ricci-tensor. Our Riemann-tensor written with flat indices is given in terms of the spin connection and the anholonomy by
\begin{equation} \begin{split}
\label{eqn:Riemann}
R_{ABCD} =\  & \partial_{A} \omega_{BCD} - \partial_{B} \omega_{ACD} + {\Omega_{AB}}^E \omega_{ECD} \\ &+ {\omega_{AC}}^E \omega_{BED} - {\omega_{BC}}^E \omega_{AED}\, .
\end{split} \end{equation}
From this expression we can derive the purely spatial Ricci-tensor appearing in (\ref{eqn:Einstein}) with our gauge choice,
\begin{equation}
\label{eqn:spatialRicci} 
R_{ab} = \partial_1 \omega_{ab1} +  \omega_{ab1} {\omega^c}_{c1} +  {\omega_{ab}}^d {\omega^c}_{dc} - \omega_{1ca} {\omega^c}_{b1} +  \omega_{ca1} {\omega_{1b}}^c +  {\omega^c}_{da} {\omega^d}_{bc}. 
\end{equation}
Now, to match with the $\sigma$-model equation (\ref{eqn:pab}), we make the ansatz $\omega_{ab1} = c_1\, p_{ab}$ and $\omega_{1ab} = c_2 \,q_{ab}$. Using this ansatz, one rewrites (\ref{eqn:spatialRicci}) as
\begin{equation}
\label{eqn:RewrittenSpatialRicci} \begin{split}
R_{ab} = \ &(Ne)^{-1} \partial_t ( e\, c_1 \,p_{ab}) - c_1\, c_2\, q_{ca}\,{p^c}_b -  c_1\, c_2 \,q_{cb}\,{p^c}_a \\ &+  {\omega_{ab}}^d {\omega^c}_{dc} + {\omega^c}_{da} {\omega^d}_{bc}\, . \end{split}
\end{equation}
Comparing with (\ref{eqn:pab}) we conclude that $c_1 = c_2 = N^{-1}$, and $n = e^{-1} N$ (multiply  (\ref{eqn:Einstein}) with $N^2$ to make the identification easier). Consider now the last term in  (\ref{eqn:RewrittenSpatialRicci}). Somehow we need to match $\omega_{abc}$ with the parameter $P_{ab}$. We do this by the ansatz
\begin{equation}
\label{eqn:omegaansatz}
\Omega_{abc} = c_3 \epsilon_{abd}{P^d}_c .
\end{equation}
There is here a mismatch in the number of degrees of freedom between these two tensors. From the symmetry of $P_{ab}$ we see that the anholonomy must obey a trace condition,
\begin{equation}
{\Omega_{ab}}^b = c_3 \epsilon_{abc} P^{cb} = 0 .
\end{equation}
This removes three of the nine degrees of freedom in the purely spatial anholonomy and with this condition its degrees of freedom equals the number of components of $P_{ab}$. It is generally assumed that this trace condition always can be imposed \cite{Damour:2002cu}. Observe that the tracelessness $\Omega_{abb} = 0$ is equivalent to $\omega_{bba} = 0$. Hence the second to last term in (\ref{eqn:RewrittenSpatialRicci}) vanishes. From the expression (\ref{eqn:anholonomy}) of the spin connection in terms of the anholonomy, the last term in (\ref{eqn:RewrittenSpatialRicci}) can be rewritten as
\begin{equation}
\label{eqn:SpinAndHolonomy}
{\omega^c}_{da} {\omega^d}_{bc} = \frac{1}{4} \Omega_{cda} {\Omega^{cd}}_b - \frac{1}{2} {\Omega_{da}}^c {\Omega^{d}}_{bc}- \frac{1}{2} \Omega_{adc} {\Omega_b}^{cd}.
\end{equation}
Let us take a closer look at the last term $\Omega_{adc} {\Omega_b}^{cd}$. The first two indices $a$ and $d$ in the first anholonomy has no matching index with the first two indices $b$ and $c$ of the second anholonomy. With our ansatz (\ref{eqn:omegaansatz}) this kind of index structure is impossible to match with  anything in the $\sigma$-model at low levels, as there is no such term in (\ref{eqn:pab}). This is a general phenomena when matching infinite-dimensional coset space $\sigma$-models with supergravity theories. In particular it is true also for the well studied case of eleven-dimensional supergravity. For example in \cite{Damour:2004zy} it is suggested that this term comes from terms in the $\sigma$-model that we in the current truncation have thrown away but the confirmation of this claim is still an open problem. Inserting our ansatz (\ref{eqn:omegaansatz}) in (\ref{eqn:SpinAndHolonomy}) we get (leaving the last term as it is)
\begin{equation}
{\omega^c}_{da} {\omega^d}_{bc} = - \frac{c^2}{2} \delta_{ab} P_{cd}P^{cd} + c^2 P_{ac} {P_b}^c - \frac{1}{2} \Omega_{adc} {\Omega_b}^{cd}. 
\end{equation}
Looking at (\ref{eqn:pab}) we get precise matching if ${c_3}^2 = 4N^{-2}$ and if we ignore the anomalous monomial in the anholonomy. The sign of $c_3$ remains unfixed, so we define $c_3 = 2N^{-1}c_3'$, where $|c_3'| = 1$. Let us now turn to the rest of the terms in the equation of motion (\ref{eqn:Einstein}) for the metric.

\item{$\mathbf{Maxwell\ field}$}

From the Einstein-Maxwell equation (\ref{eqn:Einstein}) we get when looking at the spatial part of the two monomials in the field strength (remembering that we multiplied with $N^2$),
\be \begin{split}
\frac{N^2}{2}\, \delta_{ab} F_{CD} F^{CD} - 2\, N^2 F_{aC}{F_b}^{C} = & \ \frac{N^2}{2} \delta_{ab}\, (-2 F_{1c}{F_1}^{c}  + F_{cd}F^{cd}) \\ 
&+ 2\, N^2 (F_{a1} F_{b1} -F_{ac} {F_b}^c). 
\end{split} \ee
A reasonable here ansatz is
\be
\label{eqn:fieldstrengthansatz} \begin{split}
F_{1c} &= c_4 P_c\, , \\
F_{ab} &= c_5 \epsilon_{abc} \tilde{P}^c\, ,
\end{split} \ee
giving
\be \begin{split} 
\frac{N^2}{2} \delta_{ab} F_{CD} F^{CD} - 2 N^2F_{aC}{F_b}^{C}  = &\quad   N^2 {c_4}^2(- \delta_{ab} P_c P^c + 2 P_a P_b)  \\
  & + N^2{c_5}^2(- \delta_{ab} \tilde{P}_c \tilde{P}^c + 2\tilde{P}_a \tilde{P}_b) .
\end{split} \ee
Comparing with (\ref{eqn:pab}) we see that we must put ${c_4}^2 = {c_5}^2 = \frac{1}{2N^2}$. As in the case for $c_3$, the signs of $c_4$ and $c_5$ is still unfixed, so we define $c_4 = \frac{1}{\sqrt{2}N} c_4'$ and $c_5 = \frac{1}{\sqrt{2}N} c_5'$ where again $|c_4'| = |c_5'| = 1$. Note that whether $F_{1c}$ should be proportional to $P_a$ or $\tilde{P}_a$ is up till now not fixed as they have appeared symmetrically so far. In fact, when we now turn to consider the equation of motion for $F_{AB}$ \eqref{eqn:Maxwell} and its Bianchi identity \eqref{eqn:MaxwellBianchi} it turns out that neither of these equations will fix this arbitrariness or the signs of the functions $c_i$. This is due to the symmetry between the roots $\alpha_4$ and $\alpha_5$ and can be interpreted physically as electromagnetic duality.

\item{{\bf Equations of motion and Bianchi identities for $F_{AB}$ }}

Finally we consider the equation of motion for the field strength $F_{AB}$ (\ref{eqn:Maxwell}). Explicitly the covariant derivative becomes
\begin{equation}
D^A F_{AB} = \partial^A F_{AB} - {\omega_{C}}^{AC}F_{AB} -  \omega_{ACB}F^{AC} = 0.
\end{equation}
Looking at the spatial dynamics, putting $B = b$ and splitting the sums over space and time we get
\begin{equation}
D^A F_{Ab} = -\partial_1 F_{1b} -  {\omega^e}_{e1} F_{1b} + \omega_{1cb}{F_1}^c -  \omega_{ab1}{F^a}_1 -  \omega_{acb}F^{ac} = 0 .
\end{equation}
Again we recognize the ``time'' covariant derivative from (\ref{eqn:SpinCovDerivative}). Hence we have
\begin{equation}
e^{-1} N^{-1} \partial_t (e  F_{1b})  = {\omega^e}_{e1} F_{1b} + \partial_1 F_{1b} . 
\end{equation}
Note also that $\omega_{acb}F^{ac} = \frac{1}{2} \Omega_{acb}F^{ac}$. With the expressions for the anholonomy (\ref{eqn:omegaansatz}), we derived above, we rewrite the equation for the field strength as
\begin{equation}
e^{-1} N^{-1} \partial_t (e N^{-1} {c_4}' P_b) + {c_4}' N^{-2} (q_{bc} P^c - p_{cb} P^c )+ 2{c_3}'{c_5}' N^{-2} P_{cb} \tilde P^c = 0. 
\end{equation}
This agrees with the $\sigma$-model equation (\ref{eqn:Pa}) if ${c_3}' {c_5}' = {c_4}'$.  Consider now the Bianchi-identity (\ref{eqn:MaxwellBianchi}). Letting $A=a$ we get
\be \begin{split}
\epsilon^{aBCD}D_B F_{CD} &=  \epsilon^{a1bc}D_1 F_{bc} + 2 \epsilon^{abc1} D_b F_{c1}\, , \\
&=  \epsilon^{abc} \partial_1 F_{bc} - 2 \epsilon^{abc}  \omega_{1db} {F^d}_c + 2\epsilon^{abc} \omega_{bd1} {F^d}_c - 2 \epsilon^{abc} \omega _{bdc} {F^d}_0    \\ 
& = 0.
\end{split}\ee
As $\epsilon^{abc} \omega_{bdc}  = \frac{1}{2} \epsilon^{abc} \Omega_{cbd}$,
\begin{equation} \nonumber
 \partial_1 ({c_5}' N^{-1} \tilde{P}_a) + {c_5}' N^{-2} ( q_{ab}\tilde{P}_b - p_{ab}\tilde{P}_b + N {w^b}_{b1}\tilde{P}_a )- 2 {c_3}' {c_4}' N^{-2} P_{ac} P_c =0 . 
\end{equation}
Again using (\ref{eqn:SpinCovDerivative}) and multiplying everything with $N^2$ we find
\begin{equation}
n\,  \partial_t(n^{-1} \tilde{P}_a) - p_{ab}\tilde{P}^b + q_{ab}\tilde{P}^b - 2\frac{{c_3}'{c_4}'}{{c_5}'} P_{ac} P^c = 0 .
\end{equation}
This is precisely the corresponding equation (\ref{eqn:tildePa}). 

\item{{\bf Riemann Bianchi}}

What remains to analyze is the last equation of the $\sigma$-model (\ref{eqn:Pab}), which is to be matched with the algebraic Bianchi identity (\ref{eqn:RiemannBianchi1}) for the Riemann tensor on the supergravity side. The component of (\ref{eqn:RiemannBianchi1}) to be considered is the symmetric purely spatial part. These turn out to be exactly equivalent in the current truncation, automatically by the above mapping of fields. This is a consistency check of our analysis.

\item{{\bf Summary}}

Hence, our analysis has given us an almost complete correspondence between the parameters of the truncated $\supp_C$-model and the dynamics of certain spatially constant solutions of pure $\mathcal{N} = 2$ supergravity. This result summarized in the Table \ref{tab:Dictionary} is what was expected from the structure of the low-lying $\mathfrak{sl}(3,\mathbb{R})$ representations. The map is similar to those already constructed for other supergravity theories, and succeeds and fails at the same points. We point out that in addition to the dynamical equations, there are in general constraint equations to be verified, for example the (spatial) diffeomorphism constraint and Gauss constraints. We expect that they are satisfied in the same way as for the maximally supersymmetric case~\cite{Damour:2007dt}.

\begin{table}[h]
\begin{center}
\begin{tabular}{|c|c|c|c|}
\hline
Level $L$ & Supergravity field & $\supp$ field & $ \asupp$ generator\\
\hline \hline
$0 $ &$ \omega_{abt} $ & $ p_{ab}$ & $k^{a b}$\\
$0 $ &$\omega_{tab}$ & $ q_{ab}$ & $j^{a b}$\\
$1$ &$F_{tc} $ & $ \frac{1}{\sqrt{2}} {c_4}' P_c $ & ${r^a}$\\
$1$ &$N F_{ab} $ & $\frac{1}{\sqrt{2}} {c_5}' \epsilon_{abc} \tilde{P}^c$ & $\tilde{r}^a$\\
$2$ &$N \Omega_{abc}$ & $2{c_4}' {c_5}' \epsilon_{abd}{P^d}_c $& $s^{a b}$\\
$-$ & $Ne^{-1}$& $n$ & -\\
\hline
\end{tabular}
\caption{\sl \small Correspondence between the bosonic fields in the supergravity theory and the Kac-Moody $\sigma$-model. The parameters ${c_4}'$ and ${c_5}'$ are unfixed and are $\pm 1$. All the supergravity quantities are assumed to be evaluated at a fixed spatial point.}
\label{tab:Dictionary}
\end{center}
\end{table}

\end{itemize}

 \subsection{The $\asuppp$ algebraic structure of BPS branes}

In this subsection, we show that the BPS solutions ({\ref{Isom}) which are upon dimensional reduction on time described in the $\asu$ $\sigma$-model by equation ({\ref{eqn:BPSscalars}), are in fact  completely algebraically described in $\asuppp$. In order to do so we now choose the time coordinate to be the direction $x_4$. More precisely, we show that the full space-time solution (\ref{Isom}) can be reconstructed (i.e. not only the part which correspond to scalars upon dimensional reduction) by demanding that the $\asu$ is regularly embedded\footnote{We recall that a subalgebra $\bar{\mathfrak{g}}\subset\mathfrak{g}$ is \emph{regularly embedded} in $\mathfrak{g}$ if the root vectors of $\bar{\mathfrak{g}}$ are root vectors of
$\mathfrak{g}$, and the simple roots of $\bar{\mathfrak{g}}$ are real roots of
$\mathfrak{g}$. Of particular relevance for our analysis is that, as a consequence, the Weyl group $\mathcal{W}(\bar{\mathfrak{g}})$ of $\bar{\mathfrak{g}}$ is a subgroup of $\mathcal{W}({\mathfrak{g}})$. For finite-dimensional Lie algebras the concept of a regular embedding was introduced by Dynkin in \cite{Dynkin:1957um}, and was subsequently extended to the infinite-dimensional case by Feingold and Nicolai \cite{Feingold}.} in $\asuppp$. The regular embedding is defined by erasing the nodes $\alpha_1, \alpha_2$ and $\alpha_3$  in Figure~\ref{fig:su21+++}b.

We first recall that we can describe the non-compact Cartan fields of $\asuppp$ in two bases, the $\mathfrak{gl}(4, \mathbb R)$ one described by the generators  $K^a_{\ a}$,   (see (\ref{eqn:kab})) and the Chevalley base given by the $h_m, \ m=1 \dots 4$ (see (\ref{eqn:su21+++Chev})). The fields corresponding to the former are denoted $p_a$ and the ones corresponding to the latter denoted $q_a$ ($a=1 \dots4)$. The relation between these two bases is:
\begin{equation}
\sum_{a=1}^{4} \, p_a K^a_{\ a}=\sum_{a=1}^{4} \, q_a\,  h_a,
\label{basecartsu21}
\end{equation}
where the $p_a$'s encode the diagonal metric in $\asuppp$. We have indeed $p_a= \frac{1}{2} \ln g_{aa}$ where $g_{aa}$ is the four-dimensional metric. This follows for instance from \cite{Englert:2003zs} or also from the results of the preceding section, summarized in Table~\ref{tab:Dictionary}.

We are now in position to impose the regular embedding which amounts at the level of the Cartan to enforce
\begin{equation}
q_1=q_2=q_3 =0.
\label{embeddsu21}
\end{equation}
Using (\ref{basecartsu21}) the conditions (\ref{embeddsu21}) translate for the $p_a$'s into
\be
p_1=p_2=p_3=-p_4.
\ee
Consequently the regular embedding of $\asu$  in $\asuppp$ imply on the physical four-dimensional metric the following conditions:
\begin{equation}
\label{finalembsu21}
g_{1\, 1}=g_{2\, 2}=g_{3\, 3}=g_{4\, 4}^{-1},
\end{equation}
which is satisfied by the BPS metrics (\ref{Isom}). This completes the proof that the 
 four-dimensional BPS solutions are described  by the regular embedding of $\asu$  in $\asuppp$. It is worth noticing that this description is not valid for non-BPS solutions and it indicates again the special role played by BPS solutions in the $\mathfrak{g}^{+++}$ approach (see \cite{Englert:2003py}, \cite{Englert:2007qb}).

 \subsection{Weyl reflection in $\asuppp$ }
 In this subsection, we first discuss a definition of the Weyl group of $\asu$ and its action on BPS solutions of the $\mathcal{N}=2$ supergravity. Then, we will study  the Weyl group of $\asuppp$ and its possible consequence on the space-time signature.
 \subsubsection{Weyl reflection in $\asu$ }
 
First, we briefly recall how to construct the Weyl group of the complex  algebras $A_2$. The Weyl group $W$ of $A_2$ is generated by the two simple Weyl reflections $s_{\alpha_4}$, $s_{\alpha_5}$ associated respectively to   the simple roots $\alpha_4$ and $\alpha_5$ (see Figure \ref{fig:su21+++}a). The group contains six elements
\be
W_{A_2}=\{1,s_{\alpha_4},s_{\alpha_5},s_{\alpha_4}s_{\alpha_5},s_{\alpha_5}s_{\alpha_4},s_{\alpha_4}s_{\alpha_5}s_{\alpha_4}\}\, ,
\ee
and is isomorphic to the symmetric group $\mathrm{S}_3$ on three letters. We first note that among the six elements, three correspond to reflections: $s_{\alpha_4},s_{\alpha_5}$ and $s_{\alpha_4}s_{\alpha_5}s_{\alpha_4}$. The third transformation correspond to the Weyl reflection associated to the non-simple roots $\alpha_4+\alpha_5$ namely  $s_{\alpha_4}s_{\alpha_5}s_{\alpha_4}=s_{\alpha_4+\alpha_5}$. The action of $s_{\alpha_4+\alpha_5}$ on the simple roots of $A_2$ is:
\begin{equation} \label{eqn:weyl5} \begin{split}
s_{\alpha_4+\alpha_5}(\alpha_4)\, &= -\alpha_5\, , \\
s_{\alpha_4+\alpha_5}(\alpha_5)\, & = -\alpha_4\,  .
\end{split}
\end{equation}
The strategy used here to define the Weyl group of $\asu$ is to retain only the reflections of the Weyl group of $A_2$  associated to the roots which are invariant under the conjugation $\sigma$ fixing the real form $\asu$. Using (\ref{eqn:sigmaa4a5}) we deduce that the only invariant Weyl reflection is  $s_{\alpha_4+\alpha_5}$. Consequently,  we define the Weyl group of $\asu$ as being
$W_{ \asu } =\{ 1, s_{\alpha_4+\alpha_5} \}$. This is in agreement with the restricted root system describing $\asu$ given in Appendix \ref{app:restricted}. The restricted root system of $\asu$ is $(BC)_1$ \cite{Helgason:1978} and the Weyl group of $(BC)_1$ is generated by one restricted root $\lambda_2$  (see (\ref{eqn:restricted})) which precisely correspond to the root $\alpha_4+\alpha_5$ in $A_2$.\footnote{The fact that the restricted root system of $\mathfrak{su}(2,1)$ is of non-reduced type has interesting consequences for the behaviour of $D=4$ Einstein-Maxwell gravity in the vicinity of a space-like singularity (``BKL-limit''). For information on these aspects of Maxwell-Einstein gravity, we refer to \cite{HenneauxJulia,Henneaux:2007ej}.  }

We now determine the element $\mathcal{W}$ of  SU(2,1) corresponding to the Weyl transformation  $s_{\alpha_4+\alpha_5}$ and acting by conjugation on the coset element $\mathcal{V}$ namely: $\mathcal{V}^\prime=\mathcal{W} \, \mathcal{V} \, \mathcal{W}^{-1} $.  The conjugate action on  $\mathcal{V}$ implies a conjugate action on $\mathcal{P}=\tfrac{1}{2}\big( d\mathcal{V} \mathcal{V}^{-1}\,  - \,  \Omega_4(d\mathcal{V} \mathcal{V}^{-1}) \big)$, if  $\mathcal{W}$ pertains to the invariant subgroup under the temporal involution $\Omega_4$ namely $\mathrm{K}^{*} =\mathrm{SL}(2, \mathbb R) \times \mathrm{U(1)}$ (see (\ref{eqn:temporal}) and Appendix \ref{app:k*}). We will check below that it is indeed the case.  In order to find $\mathcal{W}$ we use  (\ref{eqn:weyl5}) which translate at the level of the $A_2$ algebra into 
\begin{equation} \label{eqn:weylal} \begin{split}
&\mathcal{W} \,  E_4 \, \mathcal{W}^{-1} = \epsilon\,  F_5,\\
&\mathcal{W} \, E_5 \, \mathcal{W}^{-1} = \epsilon\,  F_4,
\end{split}
\end{equation}
 while on the generators of $\asu$ we get (see \eqref{eqn:su21+++Chev})
\begin{equation} \label{eqn:weylalsu} \begin{split}
&\mathcal{W} \,  e_4 \, \mathcal{W}^{-1} = \epsilon\,  f_4\, ,\\
&\mathcal{W} \, e_5 \, \mathcal{W}^{-1} =-  \epsilon\,  f_5\, ,
\end{split}
\end{equation}
where $\epsilon$ is a plus or minus sign\footnote{\label{foot:sign}This arises since step operators are representations of the Weyl group up to signs.}.

 Demanding the equations (\ref{eqn:weylal}) to be satisfied and imposing  $\mathcal{W}^2=1$ determine  $\mathcal{W}$ univocally, we get:
\begin{equation}
\label{eqn:weylgf}
\mathcal{W}= \exp{[-\tfrac{\pi}{2} \, h_5]} \, \exp{[\tfrac{\pi}{2} \, (e_{4,5}+f_{4,5})]},
\end{equation}
which fixes $\epsilon=-1$.
The generators $h_5$ and $(e_{4,5}+f_{4,5})$ pertaining both to $\mathfrak{k^*}= \mathfrak{sl}(2,\mathbb{R})\oplus \mathfrak{u}(1)$  and  $\mathfrak{k}= \mathfrak{su}(2)\oplus \mathfrak{u}(1)$ (see Appendices \ref{app:k} and \ref{app:k*}), the element $\mathcal{W}$ belongs to both $\mathrm{K}^*$ and $\mathrm{K}$, ensuring the validity of the procedure to derive it.

We are know in the position to derive the effect of the Weyl transformation on the BPS solutions given by  (\ref{Isom}). Since the element $\mathcal{W} \in \mathrm{K}^*$, to see how the four charges transforms we can just conjugate by $\mathcal{W}$  the charge matrix  (\ref{eqn:charges}). We find that under $\mathcal{W}$  the charges transform as:
\begin{equation}
\label{eqn:transfochw}
(m,n,q,h) \ \stackrel{\mathcal W}{\longrightarrow}\ (-m,-n,q,h) .
\end{equation}
This Weyl transformation maps physical solutions with positive charges to unphysical solutions with negative charges.

  \subsubsection{Effect of Weyl reflections on space-time signature} \label{sec:signa}
  In this section we will focus on the Weyl group of $\asuppp$ and we will study the effect of Weyl reflections on the space-time signature $(1,3)$ of the $\mathcal N=2$ supergravity theory in $D=4$. First, recall that a Weyl transformation of a generator $T$ of a Lorentzian algebra $\mathfrak{g}^{+++}$ can be expressed as a conjugation by a group element $U_W$ of $\mathrm{G}^{+++}$: $T \longrightarrow  U_W\,  T \, U^{-1}_W $. Because of the non-commutativity of Weyl reflections with the temporal involution $\Omega_i$ (defined in (\ref{eqn:temporal}))
  \be \label{eqn:nocomu}
  U_W\, (\Omega_i T)\,  U_W^{-1}= \Omega '\,  (U_W T U^{-1}_W)\, , 
   \ee
 different Lorentz signatures $(t,s)$ (where $t(s)$ is the number of time (space) coordinates) can be obtained \cite{Keurentjes:2004bv, Englert:2004ph}. The analysis of signature changing has been done for all $\mathfrak{g}^{+++}$ that are very-extensions of a simple split Lie algebra $\mathfrak{g}$ \cite{deBuyl:2005it, Keurentjes:2005jw}. In these cases, Weyl reflections with respect to a root of gravity line\footnote{The gravity line is the set of the simple roots of the $\mathfrak {sl}(n, \mathbb R)$-part of $\mathfrak{g}^{+++}$.  It corresponds in the case of  $\asuppp$ to the roots $\alpha_1, \,\alpha_2$ and $\alpha_3$. } do not change the global Lorentz signature $(t,s)$ but it changes only the identification of the time coordinate. In fact, only Weyl reflections with respect to roots not belonging to the gravity line can change the global signature of the theory. We will now study the possible signature changing induced by Weyl reflections of the non-split real form $\asuppp$.\\
 
 The Weyl group of $\asuppp$ namely $W_{\asuppp}$ is generated by the Weyl reflection $s_{\alpha_4+\alpha_5}$ belonging to $W_{\asu}$ and by the simple Weyl reflections with respect to the roots of gravity line $s_{\alpha_1}, \, s_{\alpha_2},\, s_{\alpha_3}$. Because of the presence of the affine Weyl reflection $s_{\alpha_3}$,  the Weyl group $W_{\asuppp}$ becomes infinite-dimensional
 \be
 W_{\asuppp}=\{1, s_{\alpha_1}, \, s_{\alpha_2}, \, s_{\alpha_3},\,  s_{\alpha_4+ \alpha_5}, \ldots       \}\, .
 \ee
 
 \vspace{.3cm}
\noindent $ \bullet \ $  \textbf{The effect of the Weyl reflection } $\mathbf{s_{\alpha_1} }$

\vspace{.2cm}
 As is the case for split forms,  Weyl reflections with respect to the gravity line of $\asuppp$ will not change the global signature $(1,3)$ but it will only change the identification of time index. The roots of the gravity line are indeed not affected by arrows and  they are all non-compact roots as for split real form. Let us recall a simple example of the consequence of the Weyl reflection $s_{\alpha_1}$ on the space-time signature \cite{Englert:2004ph}. We start with the temporal involution $\Omega_1$ allowing the index $1$ to be the time index.  Applying (\ref{eqn:nocomu}) to
the Weyl reflexion $s_{\alpha_1}$  generates from 
$\Omega_i \equiv\Omega_1$ a new involution $\Omega' \equiv \Omega_2$ such that
\be \begin{split}\begin{aligned}
\label{permute}
U_1\, \Omega_1 K^2_{\ 1} \, U^{-1}_1&= \rho \, K^2_{\ 1} = \rho\, \Omega_2\, 
 K^1_{\ 2}, \\
U_1\, \Omega_1 K^1_{\ 3} \, U^{-1}_1&= \sigma \, K^3_{\ 2} =
\sigma \, \Omega_2\, 
 K^2_{\ 3} \, ,\\
U_1\, \Omega_1 K^i_{\ i +1} \, U^{-1}_1&= -\tau\,  K^{i+1}_{\ \, i} =
\tau\, \Omega_2 \,  K^i_{\ i +1}\quad i >2\, ,
\end{aligned} \end{split} \ee

\noindent where $\rho,\sigma,\tau$ are plus or minus signs (see footnote \ref{foot:sign}). The equations (\ref{permute}) illustrate the general result
that such signs always cancel in the determination of
$\Omega^\prime$ because they are identical in the Weyl transform of
corresponding positive and negative roots, as
their commutator is in the Cartan subalgebra which  forms a true
representation of the Weyl group. The content of (\ref{permute}) is
represented in Table \ref{tablin}. 
The signs below the generators of the gravity
line indicate the sign in front of the
 negative step operator obtained by the involutions $\Omega_1$ and $\Omega_2$ (see (\ref{eqn:temporal})): a
minus sign indicates that the indices in
$K^m_{\ m +1}$ are both either space or time indices while a plus sign
indicates that one index must be time and the other  space.
\begin{table}[t]
\begin{center}
\begin{tabular}{|c|ccc|c|}
\hline
&$K^1_{\ 2}$&$K^2_{\ 3}$&$K^3_{\ 4}$&time coordinate\\
\hline
$\Omega_1$&$+$&$-$&$-$&1\\
\hline$\,\Omega_2$&$+$&$+$&$-$&2\\
\hline
\end{tabular}
\caption{\sl \small Involution switches from $\Omega_1$  to
$\Omega_2$ in $\asuppp$ due to the Weyl reflection $s_{\alpha_1}$.}
\label{tablin}
\end{center}
\end{table}

The Table \ref{tablin} shows that
the  time coordinates in
$\asuppp$ must now be identified either with 2, or with all indices
$\neq 2$. We choose the first description, which leaves
unaffected coordinates attached to planes invariant under the Weyl
transformation. More generally, by Weyl reflections with respect to a root of the gravity line, it is possible to identify the time index to any $\mathfrak{sl}(4, \mathbb R)$ tensor index.

 \vspace{.3cm}
\noindent $ \bullet \ $  \textbf{The effect of the Weyl reflection } $\mathbf{s_{\alpha_4+\alpha_5} }$

\vspace{.2cm}

We will now study the effect of the particular Weyl reflection $s_{\alpha_4+\alpha_5}$ on the space-time signature $(1,3)$. We will first act with $s_{\alpha_4 +\alpha_5}$ on the generators of $A_2^{+++}$ to find then the transformation of the generators of $\asuppp$. Only the simple roots $\alpha_3, \alpha_4$ and $ \alpha_5$  are modified by this reflection. Its action  on the roots $\alpha_4$ and $\alpha_5$  is done in (\ref{eqn:weyl5}) while  on the root $\alpha_3$, it acts as
\be
s_{\alpha_4+\alpha_5} (\alpha_3)= \alpha_3 + 2 (\alpha_4+ \alpha_5) \,.
\ee
Note that the root $\alpha_3$ is transformed in a root of level $\ell=(2,2)$, the root $\alpha_4$ to a negative root of level $\ell=(0,-1)$ and the root $\alpha_5$ to a negative root of level $\ell=(-1,0)$ (see Table \ref{tab:levdeca2}).
The generators associated to roots $\alpha_3, \alpha_4$ and  $\alpha_5$ are modified respectively as 
\be \begin{split}
\mathcal{W}\, K^3_{\ 4} \mathcal{W}^{-1}&= \gamma\, \big[\overbrace{[K^3_{\ 4}, S^{44}   ]}^{2 S^{34}} ,S^{44} \big] \\\
&= \gamma\, 2\,  R^{44|34}\, ,\\
\mathcal{W}\, R^{4} \mathcal{W}^{-1} &= \epsilon\, \tilde{R}_{4}\, ,\\
 \mathcal{W}\, \tilde{R}^{4} \mathcal{W}^{-1} &= \epsilon\, R_{4}\, .
\end{split} \ee
Using the Table \ref{tab:levdecsu}, we find how the generators of $\asuppp$ transform under this Weyl reflection
\be \begin{split}
\mathcal{W}\, K^3_{\ 4} \mathcal{W}^{-1}&= \gamma\, [\overbrace{[K^3_{\ 4}, \tfrac{i}{2} s^{44}   ]}^{i s^{34}} , \tfrac{i}{2}s^{44}] \\
&= - \tfrac{1}{2} \gamma\, r^{44|34}\, ,\\
\mathcal{W}\, r^{4} \mathcal{W}^{-1} &= \epsilon\, r_{4}\, ,\\
 \mathcal{W}\, \tilde{r}^{4} \mathcal{W}^{-1} &= -  \epsilon\, \tilde{r}_{4}\, .
\end{split}\ee
If we apply (\ref{eqn:nocomu}) and (\ref{eqn:temporal}), we find the action of $\Omega'$ on these generators:
\be \label{eqn: detailsinv}\begin{split}
\gamma\,  \Omega' K^3_{\ 4} &= - \tfrac{1}{2}\, \Omega' \, (\mathcal{W} r^{44|34}   \mathcal{W}^{-1})= - \tfrac{1}{2}\mathcal{W} \underbrace{\Omega_i\,  r^{44|34}}_{- \epsilon_3 \epsilon_4 r_{44|34}}  \mathcal{W}^{-1}\\
&= \gamma (-) \epsilon_3 \epsilon_4 K^4_{\ 3}\, ,\\
\epsilon\,  \Omega' r^4 &= \Omega' ( \mathcal{W} r_4         \mathcal{W}^{-1})=   \mathcal{W} \underbrace{\Omega_i r_4 }_{- \epsilon_4 r^4}       \mathcal{W}^{-1}\, ,\\
&=  \epsilon  (-\epsilon_4) r_4\\
- \epsilon \Omega' \tilde{r}^4 &= \Omega' (\mathcal{W} \tilde{r}_4         \mathcal{W}^{-1})=   \mathcal{W} \underbrace{\Omega_i\,  \tilde{ r}_4 }_{ \epsilon_4 \tilde{r}^4}       \mathcal{W}^{-1}\\
&= - \epsilon  (\epsilon_4) \tilde{r}_4\, .
\end{split} \ee  
From (\ref{eqn: detailsinv}), one gets 
\be \label{eqn:omegaf} \begin{split} \begin{aligned}
\Omega' K^{3}_{\ 4}&=\  - \epsilon_{3}\,  \epsilon_{4} \, K ^{4}_{\ 3} & &= \Omega_i \,  K^{4}_{3}\, ,\\
\Omega' r^{4}&=\ - \epsilon_4\,  r_4& & = \Omega_i \, r^4\, ,\\
\Omega' \tilde{r}^{4}&= \ \epsilon_4\, \tilde{ r}_4 &&= \Omega_i\,  \tilde{r}^4\, .
\end{aligned} \end{split} \ee
We find in (\ref{eqn:omegaf}) that the involution $\Omega' $ acts exactly in the same way that the involution $\Omega_i$ defined by (\ref{eqn:temporal}). We can then conclude that the Weyl reflection $s_{\alpha_4+ \alpha_5}$ does not affect the signature $(1,3)$ of $\mathcal{N}=2$ supergravity theory in $D=4$.

\setcounter{equation}{0}
\section{Embedding of $\asuppp$ in $\mathfrak{e}_{11}$}
\label{sec:Embedding}

In this final section, we find a regular embedding of $\asuppp$ in the split real form of $\mathfrak{e}_{11}$ \footnote{An embedding of $\mathfrak{su}(2,1)$ in $\mathfrak{e}_{8(8)}$ has been discussed in \cite{Gunaydin:2001bt}.} . This embedding will be derived using
elegant arguments from brane physics. We will relate between themselves different extremal brane configurations of eleven-dimensional supergravity and pure $\mathcal{N}=2$ supergravity in $D=4$. We first describe the brane setting we use.

\subsection{The brane setting}

We build an extremal brane configuration  allowed by the intersection rules \cite{Argurio:1997gt,Argurio:1998cp} leading upon dimensional reduction down to four to an extremal Reissner-Nordstr\"om  electrically charged black hole solution \cite{Maldacena:1996ky}  of $\mathcal{N}=2$ supergravity in $D=4$.

The configuration, that we denote by configuration {\bf\sf A}, built out of two extremal M5 branes and two extremal M2 branes  is the following (again we choose the direction 4 to be time-like):

\begin{table}[h]
\begin{center}
\begin{tabular}{|c|cccc|ccccccc|}
\hline
Branes &1& 2 & 3 & 4 & 5 & 6 & 7 & 8 & 9 & 10 & 11\\
\hline\hline
$A_1$=M5 &\, & \, & \, &$\bullet$ & \, & \, &$\bullet$ &$\bullet$  &$\bullet$ &$\bullet$& $\bullet$ \\
\hline
$A_2$=M5 &\, & \, & \, &$\bullet$ & $\bullet$ &$\bullet$ &\, &\,   &$\bullet$ &$\bullet$& $\bullet$ \\
\hline
$A_3$=M2 &\, & \, & \, &$\bullet$ & \, &$\bullet$ &\, &$\bullet$  &\, &\, & \, \\
\hline
$A_4$=M2 &\, & \, & \, &$\bullet$ &$\bullet$ & \, &$\bullet$ &\, &\, &\,& \, \\
\hline
\end{tabular}
\end{center}
\caption{\sl \small  Configuration \textbf{\sf{A}}: the extremal brane configuration leading to a four-dimensional extremal Reissner-Nordstr\"om  electrically charged black hole. The directions 1 to $4$ are non-compact (where 4 is time) and the directions 5 to 11 are compact.}
\label{tab:braneconf1}
\end{table}

This extremal configuration is generically characterised by four different harmonic functions in three dimensions, one for each brane. Here we choose the harmonic function to be the same for all the branes: $H= 1+\frac{q}{r}$ where $r$ is the radial coordinate in the four-dimensional non-compact space-time (we denote also $\phi \in [0, 2\pi] $ and $\theta \in [0, \pi]$ the usual angles, considering spherical coordinates). The metric of this intersecting branes configuration, depending only on the $q$ parameter  is:
\begin{equation}
ds^2_{11}=-H^{-2} dx_4^2 +H^2 (dx_1^2+dx_2^2+dx_3^2) + \sum_{i=5}^{11} dx_i^2.
\label{conf11}
\end{equation}
Upon dimensional reduction down to four dimensions the metric (\ref{conf11}) is the four-dimensional extremal  
Reissner-Nordstr\"om  electrically charged black hole solution of $\mathcal{N}=2$ supergravity in $D=4$
given by (\ref{Isom}) with $m=q$ and $n=h=0$ and with $t=x^4$.

The eleven-dimensional solution is characterised by four non-zero components $A^{(i)}, \  i=1 \dots 4$, of the three form potential, one for each brane. These are given by (see for instance \cite{Argurio:1998cp}):
\begin{equation}
\label{Acompo}
A^{(1)}= A_{\phi 5 6}, \qquad A^{(2)}= A_{\phi 78}, \qquad A^{(3)}=A_{468}, \qquad A^{(4)}=A_{457}.
\end{equation}
The corresponding non-vanishing components of the  field strengths are such that $\star F^{(1)}=\star F^{(2)}=F^{(3)}= F^{(4)}=\partial_r(H^{-1})$ where $\star$ denotes the Hodge dual in eleven dimensions.  As a consequence, if we want to interpret the configuration after dimensional reduction as an electric Reissner-Nordstr\"om black hole we have to identify the four-dimensional Maxwell field strength ${}^{(4)} F$ of (\ref{eqn:SugraAction4d}) as being the dimensional reduction of the diagonal eleven-dimensional field strength  ${}^{(11)}F^{diag} \equiv \star F^{(1)}+\star F^{(2)}+  F^{(3)}+ F^{(4)}$. This gives indeed back (\ref{Isom}) with $m=q$ and $n=h=0$.

Having the eleven-dimensional origin of the electrically charged extremal Reissner-Nordstr\"om black hole, we can now easily deduce the eleven-dimensional configuration corresponding to the magnetically charged extremal Reissner-Nordstr\"om by Hodge dualising {\it in four dimensions}  (i.e. the internal coordinates $x_i,  \ i=5 \dots  11$, playing now a passive role) and uplifting back to eleven dimensions. One immediately deduces that the non-zero components  ${\tilde A}^{(i)}$ of the dual configuration are:
\begin{equation}
\label{Acompo2}
{\tilde A}^{(1)}= A_{4 5 6}, \qquad {\tilde A}{(2)}= A_{478}, \qquad {\tilde A}^{(3)}=A_{\phi 68}, \qquad {\tilde A}^{(4)}=A_{\phi 57}.
\end{equation}
From (\ref{Acompo}), we deduce that the dual configuration, denoted with the letter {\bf\sf B}, is the one given in Table \ref{tab:braneconf2}.
\begin{table}[h]
\begin{center}
\begin{tabular}{|c|cccc|ccccccc|}
\hline
Branes &1& 2 & 3 & 4 & 5 & 6 & 7 & 8 & 9 & 10 & 11\\
\hline\hline
$B_1$=M2 &\, & \, & \, &$\bullet$ &$\bullet$ & $\bullet$ & \, &\, &\, &\,& \,\\
\hline
$B_2$=M2 &\, & \, & \, &$\bullet$ & \, &\, &$\bullet$ &$\bullet$  &\, &\,& \, \\
\hline
$B_3$=M5 &\, & \, & \, &$\bullet$ &$\bullet$ &\, &$\bullet$ &\, &$\bullet$ &$\bullet$ &$\bullet$ \\
\hline
$B_4$=M5 &\, & \, & \, &$\bullet$ &\, &$\bullet$ &\, &$\bullet$ &$\bullet$ &$\bullet$&$\bullet$ \\
\hline
\end{tabular}
\end{center}
\caption{\sl \small Configuration \textbf{\sf B}: the extremal brane configuration leading to a four-dimensional extremal Reissner-Nordstr\"om  magnetically charged black hole.}
\label{tab:braneconf2}

\end{table}

The knowledge of the two dual configurations in eleven dimensions will permit us to find an embedding of 
$\asuppp$ in $\mathfrak{e}_{11}$. In order to do that we first recall how branes are encoded in the algebraic structure of $\mathfrak{e}_{11}$.

\subsection{Description of the brane configuration in $\mathfrak{e}_{11}$}

We first briefly recall the algebraic structure of $\mathfrak{e}_{11}$. The Dynkin diagram is depicted in Figure
\ref{fig:e11}.

\begin{figure}[h]
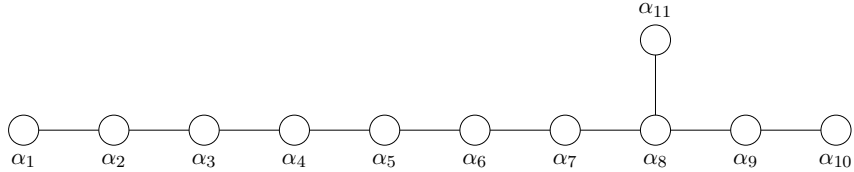
 
\begin{center}
\scalebox{.8}{
\begin{pgfpicture}{15cm}{-0.5cm}{1cm}{2.5cm}
\pgfnodecircle{Node1}[stroke]{\pgfxy(1,0.5)}{0.25cm}
\pgfnodecircle{Node2}[stroke]
{\pgfrelative{\pgfxy(1.5,0)}{\pgfnodecenter{Node1}}}{0.25cm}
\pgfnodecircle{Node3}[stroke]
{\pgfrelative{\pgfxy(1.5,0)}{\pgfnodecenter{Node2}}}{0.25cm}
\pgfnodecircle{Node4}[stroke]
{\pgfrelative{\pgfxy(1.5,0)}{\pgfnodecenter{Node3}}}{0.25cm}
\pgfnodecircle{Node5}[stroke]
{\pgfrelative{\pgfxy(1.5,0)}{\pgfnodecenter{Node4}}}{0.25cm}
\pgfnodecircle{Node6}[stroke]
{\pgfrelative{\pgfxy(1.5,0)}{\pgfnodecenter{Node5}}}{0.25cm}
\pgfnodecircle{Node7}[stroke]
{\pgfrelative{\pgfxy(1.5,0)}{\pgfnodecenter{Node6}}}{0.25cm}
\pgfnodecircle{Node8}[stroke]
{\pgfrelative{\pgfxy(1.5,0)}{\pgfnodecenter{Node7}}}{0.25cm}
\pgfnodecircle{Node9}[stroke]
{\pgfrelative{\pgfxy(1.5,0)}{\pgfnodecenter{Node8}}}{0.25cm}
\pgfnodecircle{Node10}[stroke]
{\pgfrelative{\pgfxy(1.5,0)}{\pgfnodecenter{Node9}}}{0.25cm}
\pgfnodecircle{Node11}[stroke]
{\pgfrelative{\pgfxy(0,1.5)}{\pgfnodecenter{Node8}}}{0.25cm}

\pgfnodebox{Node12}[virtual]{\pgfxy(1,0)}{$\alpha_{1}$}{2pt}{2pt}
\pgfnodebox{Node13}[virtual]{\pgfxy(2.5,0)}{$\alpha_{2}$}{2pt}{2pt}
\pgfnodebox{Node14}[virtual]{\pgfxy(4,0)}{$\alpha_{3}$}{2pt}{2pt}
\pgfnodebox{Node15}[virtual]{\pgfxy(5.5,0)}{$\alpha_{4}$}{2pt}{2pt}
\pgfnodebox{Node16}[virtual]{\pgfxy(7,0)}{$\alpha_{5}$}{2pt}{2pt}
\pgfnodebox{Node17}[virtual]{\pgfxy(8.5,0)}{$\alpha_{6}$}{2pt}{2pt}
\pgfnodebox{Node18}[virtual]{\pgfxy(10,0)}{$\alpha_{7}$}{2pt}{2pt}
\pgfnodebox{Node19}[virtual]{\pgfxy(11.5,0)}{$\alpha_{8}$}{2pt}{2pt}
\pgfnodebox{Node20}[virtual]{\pgfxy(13,0)}{$\alpha_{9}$}{2pt}{2pt}
\pgfnodebox{Node21}[virtual]{\pgfxy(14.5,0)}{$\alpha_{10}$}{2pt}{2pt}
\pgfnodebox{Node22}[virtual]{\pgfxy(11.5,2.5)}{$\alpha_{11}$}{2pt}{2pt}

\pgfnodeconnline{Node1}{Node2} \pgfnodeconnline{Node2}{Node3}
\pgfnodeconnline{Node3}{Node4}\pgfnodeconnline{Node4}{Node5}
\pgfnodeconnline{Node5}{Node6} \pgfnodeconnline{Node6}{Node7}
\pgfnodeconnline{Node7}{Node8}\pgfnodeconnline{Node8}{Node9}
\pgfnodeconnline{Node9}{Node10} \pgfnodeconnline{Node8}{Node11}

\end{pgfpicture}  }
\caption { \sl \small Dynkin diagram of  $\mathfrak{e}_{11}$.} 
\label{fig:e11}
\end{center}
\end{figure}

The Lorentzian Kac-Moody algebra $\mathfrak e_{11}$ contains a  subalgebra $\mathfrak{gl}(11, \mathbb R)$ such that $\mathfrak{sl}(11, \mathbb R) \cong A_{10}
\subset \mathfrak{gl}(11, \mathbb R) \subset
\mathfrak e_{11}$. We can again perform a level decomposition of $\mathfrak e_{11}$. The level $l$ here
is defined by the number of times the root $\alpha_{11}$ appears in the decomposition of the adjoint representation of $\mathfrak{e}_{11}$ into irreducible representation of $A_{10}$. The first levels up to $l=3$ are listed in Table \ref{tab:leve11} \cite{West:2002jj,Nicolai:2003fw}. Here, the indices are vector indices of $\mathfrak{sl}(11,\mathbb{R})$ and hence take values $a=1,\ldots,11$.

\begin{table}[h]
\begin{center}
\begin{tabular}{|c|c|c|}
\hline
$l$ &$ \mathfrak{sl}(11,\mathbb{R})$ Dynkin labels& Generator of $\mathfrak{e}_{11}$\\
\hline \hline
$0$ &$[ 1,0,0,0,0,0,0,0,0,1] $ & $K^a_{\ b}$\\
$1$ &$[ 0,0,0,0,0,0,0,1,0,0]  $ & $R^{\,abc}$\\
$2$ &$[ 0,0,0,0,1,0,0,0,0,0]  $ & $R^{\, abcdef}$\\
$3$ &$[ 0,0,1,0,0,0,0,0,0,1]  $ & $\tilde{R}^{\, abcdefgh|i}$\\

\hline
\end{tabular}
\caption{\sl \small Level decomposition of $\mathfrak{e}_{11}$ under $\mathfrak{sl}(11,\mathbb{R})$ up to level $l=3$. }
\label{tab:leve11}
\end{center}
\end{table}
The positive Chevalley generators of
$\mathfrak{e}_{11}$ are
${\tilde e}_m=\delta_m^{a} K^a{}_{a+1},\ m=1,\ldots ,10$, and  ${\tilde e}_{11}= R^{\, 9\,  
10\,11}$
where
$R^{abc}$ is the level 1 generators in $\mathfrak{e}_{11}$. One gets for the Cartan generators

\be
\begin{split}
\begin{aligned}
\label{aa} {\tilde h}_m&=\ \delta_m^{a}(K^a{}_a-K^{a+1}{}_{a+1}) \qquad \mathrm{for}  \
m=1,\dots,10\, ,\\
 {\tilde h}_{11}&=\ -\frac{1}{3}(K^1{}_1+\ldots +K^8{}_8) +\frac{2} 
{3}(K^9{}_9+
K^{10}{}_{10}+K^{11}{}_{11})\,.
\end{aligned}
\end{split}
\ee

We now recall how the extremal branes of eleven-dimensional supergravity are encoded in
the algebraic structure of $\mathfrak{e}_{11}$ (see \cite{Englert:2003py,Englert:2004it,Englert:2004ph,West:2004st}).

Each extremal brane $B_i$ corresponds  to one real root $\alpha_{B_i}$ (or one positive step operator) of $\mathfrak{e}_{11}$ and the description is always electric namely each M2 brane  is described by a definite component of the three form potential at level one and each M5 is described by a component of the six-form potential of level two. The non-zero component is the one with the indices corresponding to longitudinal directions of the extremal brane $B_i$.  The only other non-zero fields are the Cartan ones which encode the form of the metric \cite{Englert:2003py}. The intersection rules \cite{Argurio:1997gt} are neatly encoded through a pairwise orthogonality condition between the roots corresponding to each brane \cite{Englert:2004it}.

It is worthwhile to recall that such an algebraic description of extremal brane configurations extends to all space-time theories characterized by a $\mathfrak{g}^{+++}$ with simple $\mathfrak{g}$. Here, we will see that it also applies to  pure $\mathcal{N}=2$ supergravity in $D=4$ where $\mathfrak{g}^{+++}=\mathfrak{su}(2,1)^{+++}$, this will be crucial in the next subsection to uncover the embedding.

In Table \ref{tab:step}, we list the positive step operators corresponding to each brane entering in configuration {\bf\sf A} and {\bf\sf B}.
\begin{table}[h]

\begin{center}
\begin{tabular}{|c|c||c|c|}
\hline
Brane of conf. {\bf\sf A} &  step operator  &Brane of conf. {\bf\sf B} &   step operator  \\
\hline\hline
$A_1$ & $R^{4\, 7\, 8\, 9\, 10\, 11}$&  $B_1$ & $R^{4\,5\, 6}$  \\

$A_2$ & $R^{\, 4\, 5\, 6\, 9\, 10 \, 11}$ &  $B_2$ & $R^{4\, 7\, 8}$ \\

$A_3$ & $R^{4\, 6\, 8}$ &  $B_3$ &  $R^{4\, 5\, 8\, 9\, 10\, 11}$\\

$A_4$ & $R^{4\, 5\, 7}$ &  $B_4$ & $R^{4\, 6\,8\, 9\, 10\,11}$\\

\hline
\end{tabular}
\end{center}
\caption{\small The positive step operator corresponding to each brane of configurations {\bf\sf A} and {\bf \sf B}.}
\label{tab:step}
\end{table}
Since all the harmonic functions are the same in   configuration {\bf\sf A} and the dual one {\bf\sf B},  each one is characterized by an unique element of $\mathfrak{e}_{11}$.  We have
\begin{eqnarray}
\label{stepconA}
\rm{conf. \, \bf \sf{A}}& \Leftrightarrow & c\,  ( \epsilon_1 R^{4\, 7\, 8\, 9\, 10\, 11}+\epsilon_2 R^{4\, 5\, 6\, 9\, 10\, 11}+\epsilon_3 R^{4\, 6\, 8}+\epsilon_4 R^{4\, 5\, 7})\, , \\
\label{stepconB}
\rm{conf. \, \bf\sf{B}} &\Leftrightarrow& c\, (\epsilon_1 R^{4\, 5\, 6}+\epsilon_2 R^{4\, 7\, 8}+\epsilon_3 R^{4\, 5\, 7\, 9\, 10\, 11}+\epsilon_4 R^{4\, 6\, 8 \, 9 \, 10\, 11}) \, ,
\end{eqnarray}
where $c$ is a real constant and $\epsilon_i, i=1 \dots 4$ are signs. We will fix them in the next section.

\subsection{The regular embedding}

We are now in the position to find a regular embedding of $\asuppp$ in $\mathfrak{e}_{11}$.\footnote{An embedding of the split $\mathfrak{g}_2^{+++}$ in $\mathfrak{e}_{11}$ was found in \cite{Kleinschmidt:2008jj}. In this reference additional generators were added to take into account the higher rank forms that can be added consistently to the supersymmetry algebra in $D=5$ and to the tensor hierarchy~\cite{Gomis:2007gb,deWit:2008ta}.} We first discuss the non-compact Cartan generators of $\asuppp$: $h_i$ with $i=1, \ldots, 4$.

\subsubsection{The non-compact Cartan generators of $\asuppp$  }

We first recall that we can describe the Cartan fields of $\mathfrak{e}_{11}$ in two bases, the $\mathfrak{gl}(11, \mathbb R)$ one and the Chevalley base given by the ${\tilde h}_m$ (see (\ref{aa})). The relation between these two bases  (see \eqref{basecartsu21}) is:
\begin{equation}
\sum_{a=1}^{11} \, p_a K^a_{\ a}=\sum_{a=1}^{11} \, q_a\,  \tilde{h}_a\, .
\label{basecart}
\end{equation}
To find the non-compact Cartan generators of $\asuppp$ out of the eleven Cartan generators of $\mathfrak{e}_{11}$, we have simply to enforce 
\begin{equation}
p_a =0\, , \qquad a=5, \dots, 11.
\label{embc1}
\end{equation}
One can easily understand  this embedding condition in several different ways. In the brane context by noticing that the metric (\ref{conf11}) is characterized by $g_{aa}=1$ for all the longitudinal coordinates $(a=5 \dots 11)$. In a more general way this amounts to demanding that all the scalars coming from the dimensional reduction from eleven down to four should be zero. It is  a necessary condition to have a consistent truncation of eleven-dimensional supergravity to pure $\mathcal{N}=2$ supergravity in $D=4$.

Using (\ref{basecart}) we can translate the embedding condition (\ref{embc1}) in terms of the $q_a$'s using (\ref{aa}), we find
\be \begin{split} \begin{aligned} \label{embc2}
q_a &= \tfrac{a-2}{3}\,  q_{11}, \qquad a=4, \dots, 8\, , 
\\
q_9&=\tfrac{4}{3}\,  q_{11},\\
q_{10}&=\tfrac{2}{3}\,  q_{11}.
\end{aligned} \end{split} \ee
Plugging back  (\ref{embc2}) into the Cartan fields of $\mathfrak{e}_{11}$ in the Chevalley basis, we find
\begin{equation}
\sum_{a=1}^{11} q_a\, {\tilde h}_a = q_1 h_1+q_2 h_2+q_3 h_3 +\frac{q_{11}}{3} h_4,
\label{embc4}
\end{equation}
where the $h_i$ are the four non-compact Cartan generators of $\asuppp$ (see (\ref{eqn:su21+++Chev})).
This completes the discussion of the embedding for the non-compact Cartan generators.
\subsubsection{The other generators  of $\asuppp$  }

We now find  the embedding of the simple step operators and of the compact Cartan generator $h_5$.
The simple step operators corresponding to the first three nodes of Figure \ref{fig:su21+++}b are of course trivially identified with the step operators of the first three nodes of Figure \ref{fig:e11}.
We turn to the generators corresponding to the nodes 4 and 5 of Figure {\ref{fig:su21+++}b, respectively $r^4$ and $\tilde{r}^4$. An extremal Reissner-Nordstr\"om  electrically (resp. magnetically) charged black hole is a zero brane (the only longitudinal direction 4 being time-like). We recall that it is described in $\mathfrak{su}(2,1)^{+++}$ by the step operator $r^4$ (resp. $\tilde{r}^4$) \cite{Englert:2003py}. Consequently, using the brane picture expressions (\ref{stepconA}) and (\ref{stepconB}),  we have the identification 
\be \begin{split} \begin{aligned} \label{embr4}
r^4&=\tfrac{1}{\sqrt 2}\, ( \epsilon_1\,  R^{4\, 7\, 8\, 9\, 10\, 11}+\epsilon_2\,  R^{4\, 5\, 6\, 9\, 10\, 11}+\epsilon_3\,  R^{4\, 6\, 8}+\epsilon_4\, R^{4\, 5\, 7})\, ,  \\
\tilde{r}^4&=\tfrac{1}{\sqrt 2} \,  (\epsilon_1\,  R^{4\, 5\, 6}+\epsilon_2\,  R^{4\, 7\, 8}+\epsilon_3\,  R^{4\, 5\, 7\, 9\, 10\, 11}+\epsilon_4\,  R^{4\, 6\, 8 \, 9 \, 10\, 11}) \, ,
\end{aligned} \end{split} \ee
where the constant $c$ in (\ref{stepconA}) and (\ref{stepconB})  has been fixed to fulfill the normalization of    $r^4$ and $\tilde{r}^4$ in $\asuppp$ (see (\ref{eqn:bililevel1})). We still have to determine the signs $\epsilon_i$. We will fix them in the process of  determining  the compact Cartan $h_5$ of  $\asuppp$.  The  commutation relations ({\ref{comutenc45}) imply that basically $h_5$ interchanges the electric and magnetic configuration. The operator $h_5$ embedded  in   $\mathfrak{e}_{11}$ should thus correspond, in the brane picture, to the operator interchanging configuration {\bf\sf A} and {\bf\sf B} (see Tables \ref{tab:braneconf1} and \ref{tab:braneconf2}). In order to map configuration {\bf\sf A} onto configuration {\bf\sf B}, brane by brane (i.e $B_i\rightarrow {\tilde B}_i,  \quad i=1 \dots 4$), we have to perform three operations: a double T-duality in the directions 9 and 10, an exchange of the direction 6 and 7 and an exchange of the direction 5 and 8. A double T-duality in the directions 9 and 10 (followed by the exchange of the directions 9 and 10) is described in $\mathfrak{e}_{11}$ by the Weyl reflection corresponding to the simple root $\alpha_{11}$ (see Figure \ref{fig:e11}) \cite{Englert:2003zs,Elitzur:1997zn,Obers:1998rn}. The associated compact generator is: $R^{9\, 10\, 11}-R_{9\, 10\, 11}$.
The exchange of coordinates 6 and 7 (resp. 5 and 8) is generated by the compact generator 
$K^{6}_{\ 7}-K^{7}_{\ 6}$ (resp. $K^{5}_{\ 8}-K^{8}_{\ 5}$). We thus deduce that
\begin{equation}
\label{embedcartcom}
h_5=K^{6}_{\ 7}-K^{7}_{\ 6}+K^{5}_{\ 8}-K^{8}_{\ 5}+R^{9\, 10\, 11}-R_{9\, 10\, 11},
\end{equation}
we have $(h_5 | h_5)= -6$ as it should (see (\ref{eqn:h5}), (\ref{eqn:bilinearzero})).

To fix the signs in (\ref{embr4}) we use the relation $\left[ r^4, \tilde{r}_4 \right]=h_5$, we find
\be \begin{split} \begin{aligned}
r^4&=\ \tfrac{1}{\sqrt 2}\, (  R^{4\, 7\, 8\, 9\, 10\, 11}+ R^{4\, 5\, 6\, 9\, 10\, 11} - R^{4\, 6\, 8}+\ R^{4\, 5\, 7})
\label{embs1} \, ,\\
\tilde{r}^4&=\ \tfrac{1}{\sqrt 2} \,  ( R^{4\, 5\, 6}+ R^{4\, 7\, 8}- R^{4\, 5\, 7\, 9\, 10\, 11}+ R^{4\, 6\, 8 \, 9 \, 10\, 11})\, , \\
r_4&=\ \tfrac{1}{\sqrt 2}\, (  R_{4\, 7\, 8\, 9\, 10\, 11}+ R_{4\, 5\, 6\, 9\, 10\, 11} - R_{4\, 6\, 8}+\ R_{4\, 5\, 7})\, , \\
\tilde{r}_4&=\ \tfrac{-1}{\sqrt 2} \,  ( R_{4\, 5\, 6}+ R_{4\, 7\, 8}-R_{4\, 5\, 7\, 9\, 10\, 11}+ R_{4\, 6\, 8 \, 9 \, 10\, 11}) .
\end{aligned} \end{split} \ee
We can the check that the definitions (\ref{embedcartcom})-(\ref{embs1}) together with the $h_i, \, i=1\dots 4$ (see (\ref{embc4})) satisfy all the relations of $\asuppp$. 

The expressions (\ref{embedcartcom})-(\ref{embs1}) and (\ref{embc2})-(\ref{embc4}) define thus  a regular embedding of  the non-split $\asuppp$ in the split form of $\mathfrak{e}_{11}$, proving the algebraic counterpart of the truncation of maximal supergravity to the $\mathcal{N}=2$ theory.


\section*{Acknowledgments}
We thank Marc Henneaux and Ella Jamsin for collaboration in the early stages of this project.
We are grateful to  Guillaume Bossard, Hermann Nicolai and Kellogg Stelle for informing us about their work \cite{Bossard:2009at} and for clarifying exchanges.
We are greatly indebted to Riccardo Argurio for interesting discussions on brane physics and have benefited from stimulating conversations with Ling Bao, Fran\c{c}ois Dehouck, Bengt E.W. Nilsson, Jakob Palmkvist, Philippe Spindel, Amitabh Virmani and Vincent Wens.
L.H. is a Senior Research Associate of the Fonds de la Recherche
Scientifique--FNRS, Belgium, A.K. is a Research Associate of the Fonds de la Recherche
Scientifique--FNRS, and N.T. is a FRIA bursar of Fonds de la Recherche Scientifique--FNRS. Work supported in part by IISN-Belgium
(conventions 4.4511.06 and 4.4514.08), by the European Commission FP6 RTN
programme MRTN-CT-2004-005104 and by the Belgian Federal Science Policy Office
through the Interuniversity Attraction Pole P6/11.

\newpage

\appendix
\setcounter{equation}{0}

\setcounter{equation}{0}
\section{Non-linear $\sigma$-models over coset spaces}
 \label{app:SigmaModel}

Here we recall how to define a non-linear $\sigma$-model over a coset space $\mathrm G/\mathrm H$. 
Let $\mathrm G$ be a connected Lie group. Consider its real Lie algebra $\mathfrak{g}$ and an involution $\imath: \mathfrak{g} \rightarrow \mathfrak{g}$. Using the two eigenspaces of $\imath$ we can write the Lie algebra as the direct sum $\mathfrak{g} = \mathfrak{h} \oplus \mathfrak{p}$ where here $\mathfrak{h}$ is the span of the generators fixed under $\imath$ and not the Cartan subalgebra. It  does hence constitute a subalgebra. Let $\mathrm H$ be the closed Lie group corresponding the subalgebra $\mathfrak{h}$. We can now consider $\mathrm H$ as a topological subspace of $\mathrm G$ and define the coset space $\mathrm G/\mathrm H$ as the set of left cosets $\mathrm Hg$. As $\mathrm H$ is closed $\mathrm G/\mathrm H$ can be endowed with a smooth manifold structure. It is now natural to consider $\mathrm G/\mathrm H$ as the base manifold of a $\mathrm H$-principal fiber bundle $\mathrm H \hookrightarrow \mathrm  G \overset{\pi} \rightarrow \mathrm G/\mathrm H$. In fact, in the case when $\imath= \theta$ defines a Cartan decomposition of $\mathfrak{g}$, this bundle is trivial due to the global Iwasawa decomposition.

Consider a $k$-dimensional manifold $W$ with metric $h$. For simplicity we will take $W$ to have a vanishing affine connection. We can now define a $\sigma$-model for smooth maps $\mathcal{V} : W \rightarrow \mathrm G/\mathrm H$, such that $\mathcal{V} : p \mapsto \mathrm H\mathcal{V}(p)$. Locally $\mathcal{V}$ can be described, using the exponential map, by a map $v: W \rightarrow \mathfrak{p}$. Let $k : W \rightarrow \mathrm H$ be a smooth map, $\mathcal{V}$ and $k\mathcal{V}$ define hence the same map into the coset. We will call such maps $k$ gauge transformations for reasons to be clear below. Our main three criteria on the Lagrangian $\mathcal{L}_{\sigma}$ of the $\sigma$ model is that it should be invariant under $\mathrm G$ acting globally on $\mathcal{V}$ from the right, under gauge transformations and such that its equations of motions should be second order in derivatives of $\mathcal{V}$. It is hence natural to define the action in terms of the $\mathfrak{g}$-valued Maurer-Cartan form  on $\mathrm G$ restricted to $\mathrm G/\mathrm H$:
\be 
\mathrm{d}\mathcal{V} \mathcal{V}^{-1}= \mathcal{P}+\mathcal{Q}\,,
\ee
where 
\be \label{eqn:pandq} \mathcal{P} = P_{\imath} (\mathrm{d}\mathcal{V}\mathcal{V}^{-1}), \qquad
\mathcal{Q} = (1-P_{\imath}) (\mathrm{d}\mathcal{V}\mathcal{V}^{-1})\,, 
\ee
with $P_{\imath} : \mathfrak{g} \rightarrow \mathfrak{p}$ is the projection onto the coset algebra.  We therefore write our Lagrangian as
\begin{equation}
\label{eqn:SigmaModelAction}
\mathcal{L}_{\sigma} = \sqrt{|h|} h^{\mu \nu} (\mathcal{P}_{\mu} | \mathcal{P}_{\nu}) ,
\end{equation}
where $( \cdot | \cdot )$ is the Killing form of $\mathfrak{g}$ and $h = \mathrm{det}\ h_{\mu \nu}$. This Lagrangian is manifestly invariant under the right action of $\mathrm G$ on $\mathcal{V}$, and also under gauge transformations acting from the left on $\mathcal{V}$. \\

To derive the equations of motion, consider a variation of $\mathcal{V}$ from the left, i.e. $\mathcal{V}(x) \rightarrow \mathcal{V}'(x)=  e^{\epsilon(x)}\mathcal{V}(x)$, with $\epsilon(x) \in \mathfrak{p}$ infinitesimal. As the action is invariant under the action of local $\mathrm K$-transformations from the left, this is a non-trivial deformation only for $\mathfrak{p}$-valued $\epsilon(x)$. This transformation gives
\begin{equation}
\delta \mathcal{P} =  \mathrm{d}\epsilon + [\mathcal{Q}, \epsilon]\, ,
\end{equation}
 and the equations of motion thus become
\begin{equation}
\label{eqn:SigmaMotion}
\partial_{\mu} (\sqrt{|h|}  h^{\mu \nu} \mathcal{P}_{\nu}) - \sqrt{|h|} h^{\mu \nu}[ \mathcal{Q}_{\mu}, \mathcal{P}_{\nu} ] = 0.
\end{equation}
The $\mathfrak{h}$-valued field $\mathcal{Q}$ transforms under gauge transformations as a connection, i.e. if $k$ is a gauge transformation the connection $\mathcal{Q}'$ derived from $k\mathcal{V}$ is given by $\mathcal{Q}' = k\mathcal{Q}k^{-1} + \mathrm{d}k k^{-1}$. Its appearance in (\ref{eqn:SigmaMotion}) supports this point of view, so will hereafter refer to $\mathcal{Q}$ as the connection. Note that $\mathcal{Q}$ and $\mathcal{P}$ are not independent, but both derived from the map $\mathcal{V}$. The equations of motion (\ref{eqn:SigmaMotion}) is hence invariant under both gauge transformations and global $\mathrm G$-transformations. We can furthermore use Noethers theorem to derive a gauge invariant Lie algebra valued Noether ($k-1$)-form
\begin{equation}\label{eqn:currapp}
\mathcal J^{\mu} =  \sqrt{|h|} h^{\mu \nu} \mathcal{V}^{-1} \mathcal P_{\nu} \mathcal{V} ,
\end{equation}
which is conserved by virtue of the equations of motion. The Noether-form transform in the adjoint representation of $\mathrm G$, so that when $\mathcal{V} \rightarrow \mathcal{V}'=  \mathcal{V} g^{-1}$,
\begin{equation}
\mathcal J^{\mu}  \rightarrow \mathcal{J}'^{\mu}=  g\mathcal{J}^{\mu} g^{-1}.
\end{equation}
Note also that $\mathcal J$ being conserved implies the equations of motion, so (\ref{eqn:SigmaMotion}) and (\ref{eqn:currapp}) are equivalent, which is the natural consequence of the arbitrariness in defining the action of $\mathrm G$ from the right or from the left when deriving the equations of motion.

Let us consider the dynamics of this model. Choose a grading of $\mathfrak{g}$ such that
\begin{equation}
\mathfrak{g} = \underset{\ell}\bigoplus \ \mathfrak{g}_\ell, 
\end{equation}
respected by the involution $\imath$, in the sense that $\imath(\mathfrak{g}_\ell) \subset \mathfrak{g}_{-\ell}$. The restricted root space decomposition \cite{Helgason:1978} provide for example such a grading, another is given by the level decomposition under a regularly embedded subalgebra. We can then choose a base  of every level $\ell$ and $-\ell$ in terms of generators $E^{(\ell)}$ and $F^{(\ell)}$ such that $\imath(E^{(\ell)}) = -F^{(\ell)}$. Note that for a finite Lie algebra, the spaces $\mathfrak{g}_\ell$ are zero for $|\ell|$ bigger than some given $\mathrm N$, and if $\mathrm{dim}\ \mathfrak{g}_\ell > 1$, $E^{(\ell)}$ (and $F^{(\ell)}$) has some additional index, enumerating these generators.  The algebra-valued function $\mathcal{P}$ can now be expanded, with respect to this grading \cite{Damour:2004zy}, as
\begin{equation}
\label{eqn:generalP}
\mathcal{P} = \frac{1}{2}P_{(0)} K^{(0)} + \frac{1}{2} \sum_{\ell \geq 1} P_{(\ell)} (E^{(\ell)} + F^{(\ell)})\, ,
\end{equation}
and we write $K^{(0)}$ for the elements in $\mathfrak{p}$ at level zero, i.e. $K^{(0)} \in P_{\imath}(\mathfrak{g}_0)$, defining $J^{(0)}$ to span their complement in $\mathfrak{g}_0$. The connection $\mathcal{Q}$ is similarly written
\begin{equation}
\label{eqn:generalQ}
\mathcal{Q} = \frac{1}{2}Q_{(0)} J^{(0)} + \frac{1}{2} \sum_{\ell \geq 1} P_{(\ell)} (E^{(\ell)} - F^{(\ell)})\, .
\end{equation}
Inserting these expressions into (\ref{eqn:SigmaMotion}) we see that, as all the generators are linearly independent, (\ref{eqn:SigmaMotion}) split into one equation for every generator. These equations are to be interpreted as equations of motion, but also as generalized Bianchi identities and constraints, as the parameters $P_{(l)}$ are not all simultaneously physical fields. More explicitly, inserting the expansions \eqref{eqn:generalP} and \eqref{eqn:generalQ} into (\ref{eqn:SigmaMotion}), we get the equation of motion for $P_{(0)}$,

\begin{equation} \begin{split}
\label{eqn:levelzeromotion}
\frac{1}{\sqrt{|h|}} \partial_{\mu} (\sqrt{|h|} h^{\mu \nu} {P_{(0)}}_{\nu}) K^{(0)} + \frac{1}{2} {P_{(0)}}^{\mu}{Q_{(0)}}_{\mu}[J^{(0)}, K^{(0)}] &\\
 + \sum_{\ell \geq 1} {P_{(\ell)}}^{\mu} {P_{(\ell)}}_{\mu} [ E^{(\ell)}, F^{(\ell)}] = 0 ,&
\end{split} \end{equation}
and for the $P_{(\ell)}$'s we get
\be
\label{eqn:higherlevelmotion} \begin{split}
&\frac{1}{\sqrt{|h|}} \partial_{\mu} (\sqrt{|h|} h^{\mu \nu} {P_{(\ell)}}_{\nu})(E^{(\ell)}+F^{(\ell)}) + \frac{1}{2}  {Q_{(0)}}^{\mu} {P_{(\ell)}}_{\mu}  [J^{(0)}, E^{(\ell)}+F^{(\ell)}] \\
&- \frac{1}{2}  {P_{(0)}}^{\mu} {P_{(\ell)}}_{\mu}  [K^{(0)}, E^{(\ell)}-F^{(\ell)}]\\
&+ \frac12 \sum_{\substack{k,m \geq 1\\ k-m = \ell}} {P_{(k)}}^{\mu} {P_{(m)}}_{\mu} [ E^{(k)}-F^{(k)}, E^{(m)}+F^{(m)}] = 0 .
\end{split}\ee

\setcounter{equation}{0}
 \section{Generalities on $\asu$  } \label{app:su21}
 In this section, we will see how to fix $\asu$ from the complex algebra $A_2= \mathfrak{sl}(3, \mathbb{C})$. We will also give a complete list of its generators.
 
 \subsection{$\asu$: definitions}
The real form $\asu$ is the Lie algebra of $3 \times 3$ complex traceless matrices X, subject to the constraint
\be
\eta \, X+ X^{\dag} \, \eta=0 \, ,
\ee
with
\be
\eta= \left(\begin{array}{ccc}0 & 0 & -1 \\0 & 1 & 0 \\-1 & 0 & 0\end{array}\right)\, .
\ee
 This algebra is a non-split real form of the complex Lie algebra $A_2= \mathfrak{sl}(3, \mathbb{C})$ which can be written as
 \be
 \mathfrak{sl} (3,\mathbb{C})= \sum _{k=4}^{5} \mathbb{C} F_k \oplus \mathbb{C} F_{4,5} \oplus \sum _{k=4}^{5} \mathbb{C} H_k \oplus \sum _{k=4}^{5} \mathbb{C} E_k \oplus \mathbb{C} E_{4,5} , 
 \ee
 where the generators ${H_i, \, E_i,\,  F_i}$ of $\mathfrak{sl}(3,\mathbb{C})$ have the following matrix realization in the fundamental representation
 
 \begin{align} \label{eqn:chevgenA2}
 &H_4 = \left(\begin{array}{ccc}1 & 0 & 0 \\0 & -1 & 0 \\0 & 0 & 0\end{array}\right),  &&H_5= \left(\begin{array}{ccc}0 & 0 & 0 \\0 & 1 & 0 \\0 & 0 & -1\end{array}\right),  \\
 &E_4= \left(\begin{array}{ccc}0 & 1 & 0 \\0 & 0 & 0 \\0 & 0 & 0\end{array}\right), &&E_5=\left(\begin{array}{ccc}0 & 0 & 0 \\0 & 0 & 1 \\0 & 0 & 0\end{array}\right), \quad E_{4,5}= [E_4, E_5]= \left(\begin{array}{ccc}0 & 0 & 1 \\0 & 0 & 0 \\0 & 0 & 0\end{array}\right) \nn,
 \end{align}
 as well as 
 \be
 F_i=(E_i)^{T}.
\ee 

The conjugation $\sigma$, \footnote{ If $\mathfrak{g}$ is a real form of the complex Lie algebra $\mathfrak{g}_{\mathbb C}$, it defines a conjugation on $\mathfrak{g}_{\mathbb C}$. Conversely, if $\s$ is a conjugation on $\mathfrak{g}_{\mathbb C}$, the set $\mathfrak{g}_{\s}$ of elements of $\mathfrak{g}_{\mathbb C}$ fixed by $\sigma$ provides a real form of $\mathfrak{g}_{\mathbb C}$. Thus, on $\mathfrak{g}_{\mathbb C}$, real forms and conjugations are in one-to-one correspondence~\cite{Helgason:1978}. }  that fixes $\asu$ may be read off from its Tits-Satake diagram \cite{Helgason:1978,  Henneaux:2007ej} displayed in Figure \ref{figa:su21} with the following action on the simple roots of $A_2$:
 \be
 \alpha_4 + \sigma (\alpha_5) = \alpha_5+ \sigma (\alpha_4).
  \ee
Since there are no black nodes, this implies 
\be \label{eqn:sigmaroots}
 \sigma (\alpha_4) = \alpha_5, \   \sigma (\alpha_5) = \alpha_4\, .
  \ee
\begin{figure}[t]
\begin{center}
\begin{pgfpicture}{1cm}{0cm}{1cm}{3cm}

\pgfnodecircle{Node1}[stroke]{\pgfxy(1,0.5)}{0.25cm}
\pgfnodecircle{Node2}[stroke]
{\pgfrelative{\pgfxy(0,2)}{\pgfnodecenter{Node1}}}{0.25cm}

\pgfnodebox{Node6}[virtual]{\pgfxy(1,0)}{$\alpha_{4}$}{2pt}{2pt}
\pgfnodebox{Node7}[virtual]{\pgfxy(1,3)}{$\alpha_{5}$}{2pt}{2pt}
\pgfnodeconnline{Node1}{Node2} 

\pgfsetstartarrow{\pgfarrowtriangle{4pt}}
\pgfsetendarrow{\pgfarrowtriangle{4pt}}
\pgfnodesetsepend{5pt}
\pgfnodesetsepstart{5pt}
\pgfnodeconncurve{Node2}{Node1}{-10}{10}{1cm}{1cm}

\end{pgfpicture}
\caption {\label{figa:su21} \sl \small The Tits-Satake diagram of $\asu$.} 
\end{center}
\end{figure}
\noindent Thus on the generators of $\mathfrak{sl}(3, \mathbb{C})$ we have
  \be \begin{split}
  \begin{aligned}\label{eqn:sigmagens}
\s (H_4)&=H_5, & \s (H_5)&= H_4,  \\
 \s (E_4)&=E_5, &  \s (E_5)&= E_4, & \s (E_{4,5})&=-  E_{4,5}\,, \\
 \s (F_4)&=F_5, & \s (F_5)&= F_4, & \s (F_{4,5})&=-  F_{4,5}\,. 
\end{aligned} \end{split} \ee
\noindent The generators of $\asu$ correspond to the ones which are fixed by $\sigma$ and they can be written in terms of the generators of $\mathfrak{sl}(3, \mathbb{C})$ (\ref{eqn:chevgenA2}) as

\be \label{eqn:su21gens}\begin{split} 
\begin{aligned}
 h_4&=H_4+ H_5 &&=  \left(\begin{array}{ccc}1 & 0 & 0 \\0 & 0 & 0 \\0 & 0 & -1\end{array}\right), & h_5&= i (H_4-H_5)&&= \left(\begin{array}{ccc}i & 0 & 0 \\0 & -2 i & 0 \\0 & 0 & i\end{array}\right)\, , \\
e_4&= E_4 +E_5&&= \left(\begin{array}{ccc}0 & 1 & 0 \\0 & 0 & 1 \\0 & 0 & 0\end{array}\right), & f_4&= F_4 +F_5&&=\left(\begin{array}{ccc}0 & 0 & 0\\1 & 0 & 0 \\0 & 1 & 0\end{array}\right),  \\
 e_5&= i (E_4-E_5)&&= \left(\begin{array}{ccc}0 & i & 0 \\0 & 0 & -i \\0 & 0 & 0\end{array}\right) ,& f_5&= i (F_4-F_5)&&= \left(\begin{array}{ccc}0 & 0 & 0 \\ i & 0 & 0 \\0 & -i & 0\end{array}\right) , \\
 e_{4,5}&=i E_{4,5}&&= \left(\begin{array}{ccc}0 & 0 & i \\0 & 0 & 0 \\0 & 0 & 0\end{array}\right)\, , & f_{4,5}&=i F_{4,5}  &&=\left(\begin{array}{ccc}0 & 0 & 0 \\0 & 0 & 0 \\i & 0 & 0\end{array}\right), 
\end{aligned}
\end{split}
\ee
where $h_4$, $h_5$ are the generators of the Cartan subalgebra $\mathfrak{h}$,  $e_4,\, e_5$ and $ e_{4,5}$ are positive generators while $f_4,\,f_5$ and $ f_{4,5}$ are negative ones.

\subsection{$\mathfrak{k}= \mathfrak{su}(2)\oplus \mathfrak{u}(1)$}
\label{app:k}
The subalgebra $\mathfrak{su}(2)\oplus \mathfrak{u}(1)$ is defined as the maximal compact subalgebra of $\asu$. It is given as the fixed point set under the Cartan involution $\theta$,
\be
\mathfrak{k} = \big\{ x \in \asu : \, \theta(x)= x\big\}\, .
\ee
From the Tits-Satake diagram of $\asu$ (see Figure \ref{figa:su21}) we infer the following action of the Cartan involution on the simple roots
\be
\theta (\alpha_4) =-  \alpha_5, \quad \theta (\alpha_5) = - \alpha_4\, .
  \ee
On the Borel generators of $\asu$, this corresponds to
\be \begin{split} \begin{aligned} \label{eqn: invsu21}
\theta(h_4)&= -\, h_4,\ &  \theta(h_5)&= h_5, \\ 
\theta(e_4)&= -\, f_4, \ & \theta(e_5)&= \, f_5, \ \theta(e_{4,5})= \,f_{4,5} .
\end{aligned} \end{split}
\ee

We find that the subalgebra $\mathfrak{k}$ is generated by:
\be \begin{split}
\tilde{u}&= \frac{1}{2} \, (e_{4,5}+f_{4,5})+ \frac{1}{6} h_5, \\
\tilde{t}_1&= \frac{1}{2} \, \big(h_5 - (e_{4,5}+f_{4,5})\big), \\
 \tilde{t}_2&= \frac{1}{\sqrt{2} } \, (e_4-f_4), \\
   \tilde{t}_3&= \frac{1}{\sqrt{2} }\,  (e_5+ f_5)\, ,
\end{split}
\ee
where $\tilde{u}$ is the $\mathfrak{u}(1)$ generator and the $\tilde{t}_i$ generate a $\mathfrak{su}(2)$ subalgebra.
We thus have the following Cartan decomposition of $\asu$,
\be
\asu = \mathfrak{k} \oplus \mathfrak {p}= \mathfrak{su}(2) \oplus \mathfrak{u}(1) \oplus \mathfrak{p}\,,  
\ee
where $\mathfrak{p}$ are the subset of $\asu$ which are anti-invariant under $\theta$.

\subsection{The restricted root system of $\asu$}
\label{app:restricted}

Let $\mathfrak{a}$ be the maximal abelian subalgebra of $\mathfrak{p}$ which can be diagonalized over $\mathbb R$. Then
\be
\mathfrak{a}= \mathfrak{p} \cap \mathfrak{h} = \mathbb{R}\, h_4 .
\ee

\noindent The eigenvalues under the adjoint action of $h_4$ which is the only diagonalizable generator, are the following
\begin{align} \label{eqn:comute}
[h_4, e_4]&= e_4,\  && [h_4, e_5]= e_5 ,  &&[h_4, e_{4,5}]=2\,  e_{4,5},\\  \label{eqn:comutf}
[h_4, f_4]&= - f_4\ , \  &&[h_4, f_5]=-  f_5\ , \ &&[h_4, f_{4,5}]=- 2\,  f_{4,5}.
\end{align}
The generator $h_5$ is not diagonalizable over $\mathbb{R}$. Indeed, we have the following commutations relations
\begin{align}
[h_5, e_4]&=3 \, e_5, \ && [h_5, e_5]=-3 \, e_4, \ &&[h_5, e_{4,5}]=0, \label{eqn:comuth5e}\\
[h_5, f_4]&= -3 f_5, \  &&[h_5, f_5]=3  f_4, \ &&[h_5, f_{4,5}]=0. \label{eqn:comuth5f}
\end{align}
According to (\ref{eqn:comute}) and (\ref{eqn:comutf}), we may decompose $\asu$ into a direct sum of eigenspaces labelled by elements of the dual space $\mathfrak{a}^*$
\be
\asu = \bigoplus_{\lambda} \mathfrak{g}_{\lambda} \ , \ \mathfrak{g}_{\lambda}=\{ x \in \asu :\forall h \in \mathfrak{a}, \mathrm{ad}h (x)= \lambda(H)\, x \}. 
\ee
The non-compact Cartan $h_4$ then generates a $5$-grading
 of $\asu$ which is given by
 \be
 \asu= \mathfrak{g}_{(-2)} \oplus  \mathfrak{g}_{(-1)} \oplus \mathfrak{h} \oplus \mathfrak{g}_{(+1)} \oplus \mathfrak{g}_{(+2)}.
 \ee
\noindent One trivial subspace is $\mathfrak{h}$. The other nontrivial subspaces define the restricted root spaces of $\asu$ with respect to $\mathfrak{a}$ and the restricted roots are the $\lambda \in \mathfrak{a}^*$. It is now easy to determine the positive restricted root system $\sigma$ of $\asu$
\be \label{eqn:restricted}
\lambda_1(h_4)= 1 \quad, \quad \lambda_2(h_4)= 2 =   2 \lambda_1.
\ee 
Hence, the restricted root system displayed in Figure \ref{fig:restricted}, consists of the restricted root $\lambda_1$ which has multiplicity $2$ and the highest reduced root $2 \lambda_1$, with multiplicity $1$. This can be identified with the non-reduced root system $(BC)_1$ \cite{Helgason:1978, HenneauxJulia,Henneaux:2007ej}. See also \cite{Gunaydin:2007qq} for a recent analysis of $\mathfrak{su}(2,1)$ from a more representation-theoretic perspective. 

\begin{figure}[h]
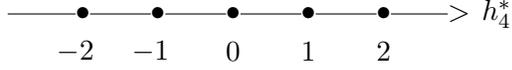

\begin{center}
\begin{pgfpicture}{3.5cm}{- 0.5cm}{1cm}{1.5cm}

\pgfnodebox{Node1}[virtual]{\pgfxy(1,0.5)}{$\bullet$}{0pt}{0pt}
\pgfnodebox{Node2}[virtual]
{\pgfrelative{\pgfxy(1,0)}{\pgfnodecenter{Node1}}}{$\bullet$}{0pt}{0pt}
\pgfnodebox{Node3}[virtual]
{\pgfrelative{\pgfxy(1,0)}{\pgfnodecenter{Node2}}}{$\bullet$}{0pt}{0pt}
\pgfnodebox{Node4}[virtual]
{\pgfrelative{\pgfxy(1,0)}{\pgfnodecenter{Node3}}}{$\bullet$}{0pt}{0pt}
\pgfnodebox{Node5}[virtual]
{\pgfrelative{\pgfxy(1,0)}{\pgfnodecenter{Node4}}}{$\bullet$}{0pt}{0pt}
\pgfnodebox{Node55}[virtual]
{\pgfrelative{\pgfxy(1,0)}{\pgfnodecenter{Node5}}}{$>$}{0pt}{0pt}
\pgfnodebox{Node0}[virtual]
{\pgfrelative{\pgfxy(-1,0)}{\pgfnodecenter{Node1}}}{}{0pt}{0pt}

\pgfnodebox{Node6}[virtual]{\pgfxy(0.9,0)}{$-2$}{2pt}{2pt}
\pgfnodebox{Node7}[virtual]{\pgfxy(1.9,0)}{$-1$}{2pt}{2pt}
\pgfnodebox{Node8}[virtual]{\pgfxy(3,0)}{$0$}{2pt}{2pt}
\pgfnodebox{Node9}[virtual]{\pgfxy(4,0)}{$1$}{2pt}{2pt}
\pgfnodebox{Node10}[virtual]{\pgfxy(5,0)}{$2$}{2pt}{2pt}
\pgfnodebox{Node60}[virtual]{\pgfxy(6.5,0.5)}{$h_4^*$}{2pt}{2pt}

\pgfnodeconnline{Node1}{Node2} \pgfnodeconnline{Node2}{Node3}
\pgfnodeconnline{Node3}{Node4} \pgfnodeconnline{Node4}{Node5} 
\pgfnodeconnline{Node5}{Node55} \pgfnodeconnline{Node0}{Node1} 

\end{pgfpicture}
\caption { \sl \small The restricted root system of $\asu$ labeled by the eigenvalues of $h_4$.} 
\label{fig:restricted}
\end{center}
\end{figure}

Let $\Sigma $ be the subset of nonzero restricted roots and $\Sigma ^{+} $ the set of positive roots, we define a nilpotent subalgebra of $\asu$ as 
\be
\mathfrak{n}_+= \bigoplus _{\lambda \in \Sigma^+} \mathfrak{g}_{\lambda}. 
\ee
Then, the algebraic Iwasawa decomposition of the Lie algebra $\asu$ reads
 \be \label{eqn:Iwasawasu21} \begin{split}
 \begin{aligned}
\asu& =  \ \mathfrak{k} \oplus \mathfrak{a}  \oplus \mathfrak{n}_+,\\
 &= \big( \mathfrak{su}(2) \oplus \mathfrak{u}(1)\big) \oplus \mathbb{R} h_4 \oplus \big(\mathbb{R} e_4 \oplus \mathbb{R} e_5 \oplus \mathbb{R} e_{4,5}\big) .
\end{aligned} \end{split} \ee
It is only $\mathfrak{a}$ that appears in the Iwasawa decomposition of $\asu$ (\ref{eqn:Iwasawasu21}) and not the full Cartan subalgebra $\mathfrak{h}$ since its compact part $h_5$ belong to $\mathfrak{k}$. This implies that when constructing the coset Langragians (\ref{eqn:lag3ds}) and (\ref{eqn:3dStationaryLagrangian}) respectively on the cosets $\mathcal C$ and  $\mathcal C^*$, the only part that will show up in the Borel gauge is the Borel subalgebra
\be
\mathfrak{b}_+= \mathfrak{a} \oplus \mathfrak{n}_+.
\ee

\subsection{$\mathfrak{k^*}= \mathfrak{sl}(2,\mathbb{R})\oplus \mathfrak{u}(1)$ } \label{app:k*}

The scalar part of the reduced  Lagrangian (\ref{eqn:3dStationaryLagrangian}) was identified with a non-linear $\s$-model constructed on the coset  $\mathcal{C}^*= \su/ \mathrm{SL}(2,\mbb R) \times \mathrm{U}(1)$. The generators of the algebra $ \mathfrak{sl}(2,\mathbb{R})\oplus \mathfrak{u}(1)$ associated to the quotient group of this coset are invariant 
 under the $\Omega_4$-involution, defined in (\ref{eqn:temporal}): 
\be
\mathfrak{k^*} = \big\{ x \in \asu :\, \Omega_4(x)= x \big\}\, ,
\ee
where this involution $\Omega_4$ acts on the generators
of the Borel subalgebra of $\asu$ as:
 \be \label{eqn:omega4} \begin{split} \begin{aligned}
\Omega_4(h_4)&= -\, h_4,\  & \Omega_4(h_5)&= h_5,\\
  \Omega_4(e_4)&= \, f_4, \ &\Omega_4(e_5)&=- \, f_5, \ \Omega_4(e_{4,5})= \,f_{4,5}\, .
\end{aligned} \end{split}
\ee
The subalgebra $\mathfrak{k}^*$ is generated by:
\be \label{eqn:nocompactsub} \begin{split}
u&= - \frac{1}{2} \, (e_{4,5} +f_{4,5})+ \frac{1}{6} h_5,\\
  t_1&= \frac{1}{2} \, \big(h_5 + (e_{4,5}+f_{4,5})\big),\\
   t_2&= \frac{1}{\sqrt{2} } \, (e_4+f_4), \\
    t_3&= \frac{1}{\sqrt{2} }\,  (e_5 - f_5)\, ,
\end{split}\ee
where $u$ is the $\mathfrak{u}(1)$ generator and the $t_i$ generate a $\mathfrak{sl}(2, \mathbb{R})$ subalgebra.

\setcounter{equation}{0}
\section{The generators of $\mathfrak{so}(2,2)$}
\label{app:so22}

The real form $\mathfrak{so}(2,2)$ of the complex algebra $D_2= A_1 \oplus A_1$, is defined as the set of matrices
\begin{equation}
X  = \left(\begin{array}{cc}X_1 & X_2 \\{X_2}^T & X_3\end{array}\right),
\end{equation}
where all $X_i$ are real $2\times2$ matrices, and $X_1$ and $X_3$ are skew symmetric \cite{Helgason:1978}. It is therefore spanned by the six generators
\begin{eqnarray}
b_1 = &\left(\begin{array}{cccc}0 & -1 & 0 & 0 \\1 & 0 & 0 & 0 \\0 & 0 & 0 & 1 \\0 & 0 & -1 & 0\end{array}\right)  & b_2 = \left(\begin{array}{cccc}0 & 0 & -1 & 0 \\0 & 0 & 0 & -1 \\-1 & 0 & 0 & 0 \\0 & -1 & 0 & 0\end{array}\right) \nonumber\\
b_3 = & \left(\begin{array}{cccc}0 & 0 & 0 & 1 \\0 & 0 & -1 & 0 \\0 & -1 & 0 & 0 \\1 & 0 & 0 & 0\end{array}\right) & b_4 = \left(\begin{array}{cccc}0 & 1 & 0 & 0 \\-1 & 0 & 0 & 0 \\0 & 0 & 0 & 1 \\0 & 0 & -1 & 0\end{array}\right) \\
b_5 =&\left(\begin{array}{cccc}0 & 0 & 0 & -1 \\0 & 0 & -1 & 0 \\0 & -1 & 0 & 0 \\-1 & 0 & 0 & 0\end{array}\right) & b_6 = \left(\begin{array}{cccc}0 & 0 & 1 & 0 \\0 & 0 & 0 & -1 \\1 & 0 & 0 & 0 \\0 & -1 & 0 & 0\end{array}\right) \nonumber \\
\nonumber
\end{eqnarray}
and the choice of base here is to streamline the analysis in Section \ref{sec:ActionOnCharges}. In fact, $\mathfrak{so}(2,2) \cong \mathfrak{sl}(2,\mathbb{R}) \oplus  \mathfrak{sl}(2,\mathbb{R})$, which is easily seen in this basis as $b_1,b_2$ and $b_3$ generate one $\mathfrak{sl}(2,\mathbb{R})$ summand, and $b_4,b_5$ and $b_6$ the other. The two compact generators are $b_1$ and $b_4$.

 \setcounter{equation}{0}
\section{Details on $\asuppp$ level decomposition} \label{app:su21+++}
\subsection{Commutators and bilinear forms of $A_2^{+++}$}\label{a2+++}
The level decomposition of the complex algebra $A_2^{+++}$ under its $A_3= \mathfrak{sl}(4, \mathbb R)$ subalgebra is performed in Section \ref{sec:levdecompo} and displayed in Table \ref{tab:levdeca2}. At level $\ell=(0,0)$, we have a $\mathfrak{gl}(4, \mathbb{R})= \mathfrak{sl}(4,\mathbb{R}) \oplus \mathbb{R}$ algebra generated by $K^{a}_{\ b}$ $(a, b = 1, \ldots, 4)$, as well as an extra scalar generator $T$. Their relations are
\be \begin{split} \begin{aligned}
\left[K^{a}_{\ b}, K^{c}_{\ d}\right] = \delta^c_b\,  K^{a}_{\ d} - \delta^a_d \, K^{c}_{\ b}\, ,\ \ \  \left[T, K^{a}_{\ b}\right]=0 \, ,\\
(K^{a}_{\ b} |K^{c}_{\ d}) = \delta^a_d \delta^c_b- \delta^a_b\delta^c_d, \  \ \ (T|T)= \frac{2}{9}  ,\ \  \ (T | K^{a}_{\ b})=0\, .
\end{aligned} \end{split}
\ee
All objects transform as $\mathfrak{gl}(4,\mathbb{R})$ tensors in the obvious way. The $T$ commutator relations are
\be
\left[T, R^a \right]= \frac{1}{2} \, R^a, \ \ \left[T, \tilde{R}^a \right]=-  \frac{1}{2} \, \tilde{R}^a .
\ee 
The negative step operators are obtained from the positive ones by lowering the indices. The commutations relations between a positive generator and the negative one are given by
\be \begin{split} \begin{aligned}
\left[ R^a,R_{b}\right] &= \delta^{a}_{b} (- \tfrac{1}{2} K  + 3\, T) + K^{a}_{\ b}\, ,\\
\left[ \tilde{R}^a,\tilde{R}_{b}\right] &= \delta^{a}_{b} (- \tfrac{1}{2} K  - 3\, T) + K^{a}_{\ b}\, ,\\
\left[ R^{ab},R_{cd}\right] &=- 3\,  \delta^{ab}_{cd} K+ 6 \, \delta^{[ a}_{[ c}\, K^{b]}_{\ d]}\, ,\\
\left[ S^{ab},S_{cd}\right] &=-  \bar{\delta}^{ab}_{cd} K+ 2 \, \delta^{(a}_{( c}\, K^{b)}_{\ d)}\, , 
\end{aligned} \end{split} \ee
with
\be
K= K^1_{\ 1}+  K^2_{\ 2}+ K^3_{\ 3}+ K^4_{\ 4}\, ,
\ee
and the bilinear forms are given  by
\be \begin{split}\begin{aligned}
(R^a | R_{b}) &= \delta^{a}_{b},  & (\tilde{R}^a | \tilde{R}_{b})  &= \delta^{a}_{b},\\
(R^{ab} | R_{cd})&= 3\,   \delta^{ab}_{cd},  & (S^{ab} | S_{cd}) &= \bar{ \delta}^{ab}_{cd}, 
\end{aligned} \end{split} \ee
where
\be \begin{split} \begin{aligned}
\delta^{ab}_{cd} &:=  \tfrac{1}{2} (\delta^a_c\, \delta^b_d- \delta^b_c\, \delta^a_d)\, ,\\
\bar{  \delta}^{ab}_{cd} &:=  \tfrac{1}{2} (\delta^a_c\, \delta^b_d+ \delta^b_c\, \delta^a_d)\, .
\end{aligned} \end{split} \ee
The generators of different rank commute in the following non-trivial way:
\be \begin{split} \begin{aligned}
\left[ S^{ab}, R_c\right] &= - \delta^{(a}_{c}\,  \tilde{R}^{b)}, && \left[ S^{ab}, \tilde{R}_c\right] =  \delta^{(a}_{c}\, R^{b)}\, ,\\
\left[ R^{ab}, R_c\right] &= - 3 \,  \delta^{[a}_{c}\, \tilde{R}^{b]}, && \left[ R^{ab}, \tilde{R}_c\right] = - 3 \,  \delta^{[a}_{c}\, R^{b]}, \\
\left[ S^{ab}, R_{cd} \right]&=0 \, .
\end{aligned} \end{split} \ee
We identify the Chevalley generators of $A_2^{+++}$ as
\be \label{eqn:chevbis}
\begin{split}
\begin{aligned}
&H_1 = K^1_{\ 1} - K^2_{\ 2}, &\qquad&E_1= K^1_{\ 2}\, ,\\
&H_2 = K^2_{\ 2} - K^3_{\ 3}, &\qquad&E_2= K^2_{\ 3}\, ,\\
&H_3 = K^3_{\ 3} - K^4_{\ 4}, &\qquad&E_3= K^3_{\ 4}\, ,\\
&H_4 = - \tfrac{1}{2}\, K + K^4_{\ 4} +3\, T, &\qquad&E_4= R^{4}\, ,\\
&H_5 = - \tfrac{1}{2}\, K + K^4_{\ 4} -3\, T, &\qquad&E_5= \tilde{R}^{4}\, .
\end{aligned} \end{split}
\ee

\subsection{Commutators and bilinear forms of $\asuppp$} \label{app:comsu21+++}
The level decomposition of $\asuppp$ under its $\mathfrak{sl}(4, \mathbb R)$ subalgebra is performed in Section \ref{sec:levdecompo} and displayed in Table \ref{tab:levdecsu}.
The generators at opposite levels commute as follows
\be \begin{split} \begin{aligned}
\left[ r^a, r_b \right] &= - \delta^a_b\, K + 2\, K^a_{\ b}\, ,\\
\left[ \tilde{r}^a, \tilde{r}_b \right] &= \delta^a_b\, K - 2\, K^a_{\ b}\, ,\\
\left[ r^a, \tilde{r}_b \right] &= 6\, i \,\delta^a_b\, T \, ,\\
\left[ s^{a b}, s_{c d} \right] &= 4\,  \bar{\delta}^{ab}_{cd}\, K  - 8\, \delta^{(a}_{( c}\, K^{b)}_{\ d)}   \, ,\\
\left[ r^{a b}, r_{c d} \right] &= - 12\,  \delta^{ab}_{cd}\, K  + 24\, \delta^{[a}_{[c}\, K^{b]}_{\ d]}   \, ,\\
\left[ r^{a b}, s_{c d} \right] &=  0 \,.
\end{aligned} \end{split} \ee
The generators of different rank commute in the following non-trivial way:
\be \begin{split} \begin{aligned}
\left[ s^{ab}, r_c\right] &= - 2\,  \delta^{(a}_{c}\,  \tilde{r}^{b)}, &&  \left[ s^{ab}, \tilde{r}_c\right] = -2\,  \delta^{(a}_{c}\, r^{b)}\, ,\\
\left[ r^{ab}, r_c\right] &= - 6 \,  \delta^{[a}_{c}\, r^{b]}, &&  \left[ r^{ab}, \tilde{r}_c\right] = 6 \,  \delta^{[a}_{c}\, \tilde{r}^{b]}\, .
\end{aligned} \end{split} \ee
The generators are normalized as
\be \begin{split} \begin{aligned}
(r^a | r_{b}) &= 2\, \delta^{a}_{b},  & (\tilde{r}^a | \tilde{r}_{b}) &= - 2\,  \delta^{a}_{b},\\
(r^{ab} | r_{cd})&= \,12 \,  \delta^{ab}_{cd} , & (s^{ab} | s_{cd}) &= -4 \, \bar{ \delta}^{ab}_{cd}\, .
\end{aligned} \end{split} \ee

\newpage

 \addcontentsline{toc}{section}{References}
\bibliographystyle{jhep}
\bibliography{Refdata}
\end{document}